\documentclass[11pt, letterpaper]{article}
\usepackage[utf8]{inputenc}
\usepackage{fullpage,times}

\usepackage[colorlinks=true,linkcolor=black,citecolor=black,urlcolor=black]{hyperref}

\usepackage{mathptmx}
\usepackage{subfig}

\usepackage{bm}
\usepackage[dvips]{color}
\usepackage[pdftex]{graphicx}
\usepackage{paralist,tabularx}
\usepackage{amsmath,amsfonts,amssymb,amsthm}
\usepackage{latexsym}
\usepackage{enumerate}
\usepackage{mathtools}
\usepackage{wrapfig}

\usepackage{algorithm, algorithmicx, algpseudocode} 
\usepackage{algorithmicx}

\usepackage{algorithm}
\usepackage{algorithmicx}
\usepackage{algpseudocode}
\usepackage[algo2e]{algorithm2e}

\newtheorem{theorem}{Theorem}[section]
\newtheorem{lemma}[theorem]{Lemma}

\newtheorem{corollary}[theorem]{Corollary}
\newcounter{definition}[section]
\newenvironment{definition}[1][]{\refstepcounter{definition}\par\medskip
   \noindent \textbf{Definition~\thedefinition. #1} \rmfamily}{\medskip}	
\newenvironment{reminder}[1]{\smallskip
	\noindent {\scshape \textbf{Reminder of #1 }}\em}{
}

{\makeatletter
 \gdef\xxxmark{%
   \expandafter\ifx\csname @mpargs\endcsname\relax % in minipage?
     \expandafter\ifx\csname @captype\endcsname\relax % in figure/caption?
     \marginpar{xxx}% not in a caption or minipage, can use marginpar
          \else
       xxx % notice trailing space
    \fi
   \else
     xxx % notice trailing space
   \fi}
 \gdef\xxx{\@ifnextchar[\xxx@lab\xxx@nolab}
 \long\gdef\xxx@lab[#1]#2{\textbf{[\xxxmark #2 ---{\sc #1}]}}
 \long\gdef\xxx@nolab#1{\textbf{[\xxxmark #1]}}
}

\newcommand{\dOVm}{$\#$OV$^{\mu,d}$}

\newcommand{\tO}{\tilde{O}}

\newcommand{\eps}{\epsilon}

\def \Z {\mathbb Z}
\def \R {\mathbb R}

\newcommand{\kOV}{$k$-OV}
\newcommand{\kSUM}{$k$-SUM}
\newcommand{\kXOR}{$k$-XOR}
\newcommand{\kxor}{\kXOR}
\newcommand{\aczkc}{ACZ$k$C}
\newcommand{\zkc}{Z$k$C}

\newcommand{\kov}{\kOV}

\newcommand{\klcs}[1][]{\ifthenelse{\equal{#1}{}}{$k$-LCS}{${#1}$-LCS}}

\newcommand{\kwlcs}[1][]{\ifthenelse{\equal{#1}{}}{$k$-WLCS}{${#1}$-WLCS}}

\newcommand{\kNLstC}[1][]{\ifthenelse{\equal{#1}{}}{$k$-NLstC}{${#1}$-NLstC}}

\newcommand{\kELstC}[1][]{\ifthenelse{\equal{#1}{}}{$k$-ELstC}{${#1}$-ELstC}}

\newcommand{\czkc}[1][]{\ifthenelse{\equal{#1}{}}{FZ$k$C}{FZ${#1}$C}}

\newcommand{\czkch}[1][]{\ifthenelse{\equal{#1}{}}{FZ$k$CH}{FZ${#1}$CH}}

\newcommand{\cfkc}[1][]{\ifthenelse{\equal{#1}{}}{F$\mathfrak{f}k$C}{F$\mathfrak{f}{#1}$C}}

\newcommand{\cfkch}[1][]{\ifthenelse{\equal{#1}{}}{F$\mathfrak{f}k$CH}{F$\mathfrak{f}{#1}$CH}}

\newcommand{\ckov}[1][]{\ifthenelse{\equal{#1}{}}{F$k$-OV}{F${#1}$-OV}}

\newcommand{\ckovh}[1][]{\ifthenelse{\equal{#1}{}}{F$k$-OVH}{F${#1}$-OVH}}

\newcommand{\cksum}[1][]{
\ifthenelse{\equal{#1}{}}
{F$k$-SUM}
{F${#1}$-SUM}
}

\newcommand{\cksumh}[1][]{\ifthenelse{\equal{#1}{}}{F$k$-SUMH}{F${#1}$-SUMH}}

\newcommand{\fzkc}[1][]{\ifthenelse{\equal{#1}{}}{FZ$k$C}{FZ${#1}$C}}

\newcommand{\ckxor}[1][]{\ifthenelse{\equal{#1}{}}{F$k$-XOR}{F${#1}$-XOR}}

\newcommand{\ckxorh}[1][]{\ifthenelse{\equal{#1}{}}{F$k$-XORH}{F${#1}$-XORH}}

\newcommand{\ckfunc}[1][]{\ifthenelse{\equal{#1}{}}{F$k$-$\mathfrak{f}$}{F${#1}$-$\mathfrak{f}$}}

\newcommand{\ckfunch}[1][]{\ifthenelse{\equal{#1}{}}{F$k$-$\mathfrak{f}$H}{F${#1}$-$\mathfrak{f}$H}}

\newcommand{\ksum}[1][]{\ifthenelse{\equal{#1}{}}{$k$-SUM}{${#1}$-SUM}}

\newcommand{\chpgh}{CHGHP}

\newcommand{\ckfunct}{\textrm{F}k\textrm{-}f}

\newcommand{\czt}{FZT}

\newcommand{\NDMT}{PMT}

\newcommand{\Unif}{\textsc{Unif}}

\newcommand{\GoodPolys}{Good Low-Degree Polynomials}
\newcommand{\goodPoly}{good low-degree polynomial}
\newcommand{\gPol}[1]{GLDP(#1)}

\newenvironment{proofof}[1]{{\bf Proof of #1.  }}{\hfill$\Box$}

\title{New Techniques for Proving Fine-Grained Average-Case Hardness}
\author{Mina Dalirrooyfard\\{MIT, minad@mit.edu} \and Andrea Lincoln\\ {MIT, andreali@mit.edu} \and Virginia Vassilevska Williams\\ {MIT, virgi@mit.edu} }
\date{}
\begin{document}

\maketitle
\thispagestyle{empty}
\begin{abstract}
%\begin{abstract}
%Fine-grained complexity (FGC) is a vibrant research area with many applications. 
The recent emergence of fine-grained cryptography strongly motivates developing an average-case analogue of Fine-Grained Complexity (FGC).

Prior work [Goldreich-Rothblum 2018, Boix-Adser{\`{a}} et al. 2019, Ball et al. 2017] developed worst-case to average-case fine-grained reductions (WCtoACFG) for certain algebraic and counting problems over natural distributions and used them to obtain a limited set of cryptographic primitives.
To obtain stronger cryptographic primitives based on standard FGC assumptions, ideally, one would like to develop WCtoACFG reductions from the core hard problems of FGC, Orthogonal Vectors (OV), CNF-SAT, $3$SUM, All-Pairs Shortest Paths (APSP) and zero-$k$-clique. Unfortunately, it is unclear whether these problems actually are hard for any natural distribution. It is known, that e.g.~OV can be solved quickly for very natural distributions [Kane-Williams 2019], and in this paper we show that even counting the number of OV pairs on average has a fast algorithm.

This paper defines new versions of OV, $k$SUM and  zero-$k$-clique that are both worst-case and average-case fine-grained hard assuming the core hypotheses of FGC. We then use these as a basis for fine-grained hardness and average-case hardness of other problems.
The new problems represent their inputs in a certain ``factored'' form. We call them ``factored''-OV, ``factored''-zero-$k$-clique and ``factored''-$3$SUM. We show that factored-$k$-OV and factored $k$SUM are equivalent and are complete for a class of problems defined over Boolean functions.
Factored zero-$k$-clique is also complete, for a different class of problems.

Our hard factored problems are also simple enough that we can reduce them to many other problems, e.g.~to edit distance, $k$-LCS and versions of  Max-Flow. We further consider counting variants of the factored problems and give WCtoACFG reductions for them for a natural distribution. Through FGC reductions we then get \emph{average-case} hardness for well-studied problems like regular expression matching from standard \emph{worst-case} FGC assumptions. 

To obtain our WCtoACFG reductions, we formalize the framework of [Boix-Adser{\`{a}} et al. 2019] that was used to give a WCtoACFG reduction for counting $k$-cliques.  We define an explicit property of problems such that if a problem has that property one can use the framework on the problem to get a WCtoACFG self reduction. We then use the framework to slightly extend Boix-Adser{\`{a}} et al.'s average-case counting $k$-cliques result to average-case hardness for counting arbitrary subgraph patterns of constant size in $k$-partite graphs. 

The fine-grained public-key encryption scheme of [LaVigne et al.'20] is based on an average-case hardness hypothesis for the {\em decision} problem, zero-$k$-clique, and the known techniques for building such schemes break down for algebraic/counting problems.
Meanwhile, the WCtoACFG reductions so far have only been for counting problems. To bridge this gap, we show that for a natural distribution, an algorithm that detects a zero-$k$-clique with high enough probability also implies an algorithm that can count zero-$k$-cliques with high probability. This gives hope that the FGC cryptoscheme of [LaVigne et al.'20] can be based on standard FGC assumptions.

\end{abstract}
\newpage

\section{Introduction}
\label{sec:Intro}

%story:

%\xxx{motivation/story: from avg-case to defining factored problems}

%\xxx{Factored problems are cool 1: they are hard from the regular versions, and they all reduce to FOV and ..}

%\xxx{Factored problems are cool 2: They are not too hard, there are "normal" (aka non-factored) problems that they reduce to, such as partitioned matching triangle, st reachability, a version of flow problems and LCS (a problem that existed before)}

%\xxx{side-motivation: why we define these "natural" problems? we want to say that our factored problems can solve problems from different kinds, flow/reachability/string etc }

%\xxx{finally getting avg case hardness: We can get avg case hardness for st-reachability and partitioned matching triangles (and ST-g-flow if we define counting flow correctly) through all the factored problems.}

%\xxx{all these avg case results at the end come from our framework below. we get avg case hardness for H-counting/detection as well, which was stated open by Adsera et al.}

%\xxx{TODO: Give a high level explanation of our results showing that detection (at the right range) is the hardest average case problem.}

%\xxx{If we delete any theorems from here (we likely won't want all of these theorem statements in the intro). We should fix the reminder environments in the body to be back to regular theorems and lemmas again. }
%{\bf ================ A draft intro  start =============== }

Fine-grained complexity (FGC) is an active research area that seeks to understand why many problems of interest have particular running time bounds $t(n)$ that are easy to achieve with known techniques, but have not been improved upon significantly in decades, except by $t(n)^{o(1)}$ factors.
FGC has produced a versatile set of tools that have resulted in surprising {\em fine-grained} reductions that together with popular hardness hypotheses explain the running time bottlenecks for a large variety of problems \cite{virgiSurvey}. The reductions of FGC have, for example, explained the difficulty of improving over the $n^{2-o(1)}$ time algorithms for Longest Common Subsequence (LCS) by giving a tight reduction from $k$-SAT, and thus showing that an improved LCS algorithm would violate 
the Strong Exponential Time Hypothesis (SETH)  \cite{LCSisHard}.

There are three main problems, with associated hardness hypotheses about their running times, that FGC primarily uses as sources of hardness reductions (see \cite{virgiSurvey}). The three core hard problems are All Pairs Shortest Paths  (APSP), hypothesized to require $n^{3-o(1)}$ time in $n$-node graphs\footnote{All hypotheses are for the word-RAM model of computation with $O(\log n)$ bit words.}, the $3$SUM problem, hypothesized to require $n^{2-o(1)}$ time on $n$ integer inputs, and the Orthogonal Vectors (OV) problem, hypothesized to require $n^{2-o(1)}$ time for $n$ vector inputs of dimension $\omega(\log n)$
(the OV hypothesis is implied by SETH \cite{ryanThesis}). 

While it is unknown whether these three hypotheses are equivalent, some work suggests they might not be \cite{CarmosinoGIMPS16}.
There is a problem, Zero Triangle, on $n$ node graphs that requires $n^{3-o(1)}$ time under both the $3$SUM and the APSP hypothesis \cite{vw10j,vw09j}. Zero Triangle asks if an $n$ node graph with integer edge weights contains a triangle whose three edge weights sum to $0$. A natural extension of Zero Triangle, zero-$k$-clique (where one wants to detect a $k$-clique with edge weight sum $0$), is conjectured to require $n^{k-o(1)}$ time.
There are also some simple to define problems on $n$ node graphs that require $n^{3-o(1)}$ time under three core hardness hypotheses (SETH, APSP and $3$SUM): Matching Triangles and Triangle Collection \cite{matchingTriangles}. 

% When basing hardness on unproven hardness hypotheses it is of course desirable to choose those that are the most plausible. The hypothesis that at least one of the three core hypotheses is true is an example of a more plausible hypothesis. Previous work defined the matching triangles problem and showed it required $n^{3-o(1)}$ time if any one of the three core hypotheses was true \cite{matchingTriangles}. 
%%
Recently there has been increased interest in developing average-case fine-grained complexity (ACFGC), with a new type of {\em fine-grained cryptography} as a main motivation \cite{BallWorstToAvg,BallRSV18,GoldreichR18, fgCrytpo, UniformCliqueABB}. The main goal is to identify a problem $P$ that requires some $t(n)^{1-o(1)}$ time on average for an easily sampled distribution, and then to build interesting cryptographic primitives from this problem, where any honest party only needs to run a very fast algorithm, in some $t'(n)\leq O(t(n)^{c})$ time for $c$ much smaller than $1$, while an adversary would need to run at least in $t(n)^{1-o(1)}$ time, unless problem $P$ can be solved fast on average.

% This work is generally motivated by either finding explicit distributions over which problems of interest are average case hard (e.g. \cite{BallWorstToAvg,GoldreichR18, UniformCliqueABB}) or building fine-grained cryptography primitives \cite{BallWorstToAvg,BallRSV18, UniformCliqueABB}. 

To obtain average-case fine-grained hard problems, one would like to be able to obtain worst-case to average-case fine-grained reductions for natural problems that are hypothesized to be fine-grained hard in the worst-case\footnote{Well, even more ideally, one would like to use problems that are provably unconditionally average-case hard, such as the problems from the known time-hierarchy theorems, but these problems are difficult to work with and there are no known techniques to build cryptography from them.}. This is what prior work does. 

The problems for which fine-grained worst-case to average-case hardness reductions are known are mostly algebraic or counting problems, such as counting $k$-cliques \cite{goldreichrothblum20,GoldreichR18,BallRSV18,UniformCliqueABB}, or some problems involving polynomials. Some limited cryptographic primitives have been obtained from such problems, e.g. fine-grained proofs-of-work \cite{BallRSV18,BallWorstToAvg}.
Building fine-grained one-way functions or fine-grained public key cryptography based on any worst-case FGC hardness assumption is still an open problem. Such primitives have been developed, based on plausible assumptions about the average-case complexity of zero-$k$-clique \cite{fgCrytpo}. This motivates the following question: {\em  Is there a fine-grained worst-case to average-case reduction for zero-$k$-clique?} 

As prior work showed  worst-case to average-case case reductions for counting cliques, %(and these naturally extend to counting zero-$k$-cliques **) 
a natural approach to obtaining worst-case to average-case reductions for the detection variant of zero-$k$-clique is to give a fine-grained reduction from counting to decision. A tight reduction is not known for the worst-case version of the problem. It turns out that a fine-grained reduction from counting to decision for zero-$k$-clique is possible in the average-case for a natural distribution with certain parameters, if the detection probability is high enough. We prove this in Section \ref{sec:countToDetect}. While the parameters are currently not good enough to imply a worst-case to average-case reduction for (the decision version of) zero-$k$-clique, the reduction gives hope that the fine-grained public-key scheme of \cite{fgCrytpo} can eventually be based on a standard FGC (worst-case) hardness assumption.

The next natural question is whether worst-case to average-case reductions are possible for the other core problems of FGC, and in particular for OV (as it is as far as we know unrelated to zero-$k$-clique). Consider the most natural distribution for OV: given a fixed probability $p\in (0,1)$, one generates $n$ vectors of dimension $d=\omega(\log n)$ by selecting for each vector $v$ and $i\in [d]$ independently, $v_i$ to be $1$ with probability $p$ and $0$ otherwise. Kane and Williams~\cite{ryanAvgCaseOV} showed that for every $p$, there is an $\eps_p>0$ and an $O(n^{2-\eps_p})$ time algorithm that solves OV on instances generated from the above distribution with high probability. Thus, for this distribution (if the OV conjecture is true), there can't be a fine-grained $(n^2,n^2)$-worst-case to average-case reduction for OV. 
In Section \ref{sec:OV} we also show that even the counting version of OV, in which one wants to determine the number of pairs of orthogonal vectors, has a truly-subquadratic time algorithm that works with high probability over the same distribution. Thus, even counting OV cannot be average-case $n^{2-o(1)}$-hard. (Though, it could be fine-grained average-case hard for a different time function. We leave this to future work.)

The first key contribution of this paper is in defining a new type of problem, a ``factored problem'' that is fine-grained hard from a core FGC assumption, whose counting version is average-case hard for a natural distribution again under a core FGC assumption, and that is also simple enough so that one can reduce it to well-studied problems and develop average-case hardness for them.

While developing worst-case to average-case reductions for our factored problems, we formalize the worst-case to average-case fine-grained reductions framework of Boix et al. \cite{UniformCliqueABB}. We identify a property of problems (the existence of a ``good polynomial'') that makes it possible for these problems to have such a worst-case to average-case reduction. 
Originally, \cite{UniformCliqueABB} gave average-case hardness for counting $k$-Cliques in Erd\"os-Renyi graphs using their framework.
Along the way of generalizing their framework, we also obtain a worst-case to average-case reduction for counting copies of $H$ for any $k$-node $H$, where the distribution for the average-case instance is again for Erd\"o-Renyi graphs. We achieve this using a new technique we call Inclusion-Edgesclusion.

In the rest of the introduction we will present our results mentioned in the above two paragraphs.

\subsection{The factored problems}
We call the problems we introduce ``factored problems'' (a full formal definition is in Section~\ref{sec:Prelims}).
To define them, let us first define a {\em factored vector}. Let $b$ and $g$ be positive integers. A $(g,b)$-factored vector, $v$, is made up of $g$ sets $v[1], \ldots, v[g]$. Each set is a subset $v[i] \subseteq \{0,1\}^b$. Roughly speaking, a factored vector $v$ represents many $b\cdot g$ binary vectors, namely a concatenation $x_1,x_2,\ldots,x_g$ for each choice of a $g$-tuple of vectors $x_i\in v[i]$ for all $i$.  For example, for $g=2$ and $b=3$, let $v$ be a factored vector where
$v[0] = \{001,010\}$ and $v[1] = \{010,110\}$. A natural interpretation of $v$ is that it is a set of the following $4$ binary vectors, by concatenating each member of $v[0]$ with each member of $v[1]$, that is $\{001010, 001110,010010,010110\}$.

Now, consider a function $f$ that takes a $2b$-bit input $x_1,\ldots,x_b,y_1,\ldots, y_b$ and returns a value in $\{0,1\}$; we can consider $f$ as a Boolean function. Then, for two factored vectors $v$ and $v'$ and a coordinate $i\in [g]$, we can consider the number of pairs of $b$-bit vectors $x\in v[i], y\in v'[i]$ that $f$ accepts. This is $accept_f(v,v',i):=\sum_{x\in v[i],y\in v'[i]} f(x_1,\ldots,x_b,y_1,\ldots,y_b)$, where $x=x_1\ldots x_b$ and $y=y_1\ldots y_b$. If we take the product $\prod_{i=1}^g accept_f(v,v',i)$, we would obtain the number of pairs of $b\cdot g$-length vectors represented by $v$ and $v'$ that are accepted by $f$, where $f$ is said to accept a pair of $b\cdot g$-length vectors if it accepts each of the $g$ pairs of chunks of $b$-length subvectors between positions $(i-1)b+1$ to $ib$ for $i\in [g]$.

Then we can define the factored problem for $f$, \ckfunc[2]~that given two sets $S$ and $T$ of $n$ $(g,b)$-factored vectors, computes the sum $\sum_{v\in S,v'\in T} \prod_{i=1}^g accept_f(v,v',i)$, i.e. the total number of pairs of vectors represented by vectors in $S$ and $T$ that are accepted by $f$. For technical reasons, we restrict the
 values $g=o(\lg(n)/\lg\lg(n))$ and $b = o(\lg(n))$, so that each factored vector can be represented with at most $gb2^b$ bits ($g$ sets of at most $2^b$ vectors of length $b$). 
 
Depending on the function $f$, we get different versions of a factored problem. If $f$ on $b$-length vectors $x$ and $y$, returns $1$ iff $x\cdot y=0$, then we get the factored OV problem \ckov[2]. If $f$ returns $1$ if the XOR of $x$ and $y$ is $0$, we get the \ckxor[2]~problem, and if $f$ returns $1$ iff $x+y=0$ when viewed as integers, we get the \cksum[2]~problem.

More generally, $f$ can be defined over $k\cdot b$-length vectors, for integer $k\geq 2$, taking $k$-tuples of $b$-length binary vectors to $\{0,1\}$. Then analogously we can define \ckfunc~to compute the number of $k$-tuples of vectors represented by some $k$-tuple of factored vectors, one from each $n$-sized input set $S_i$, $i\in [k]$, so that $f$ accepts the $k$-tuple. This way we can define \ckov, \ckxor, \cksum etc, the factored versions of \kov, \kxor~and \ksum.

Similarly to these problems defined on $k$-tuples of sets of factored vectors, we define problems reminiscent to $k$-clique. Here $f$ is a function that takes $\binom{k}{2}$-tuples of $b$-length vectors to $\{0,1\}$, one is given a graph whose edges are labeled by factored vectors and the factored $f$ $k$-clique problem, \cfkc[], asks to compute the number of $\binom{k}{2}$-tuples of vectors that are accepted by $f$ and are represented by the factored vectors labeling the edges of a
$k$-clique in the graph. We focus in particular on the factored zero-$k$-clique problem, \czkc[], in which $f$ corresponds to returning whether the sum of $\binom{k}{2}$ $b$-bit numbers is $0$.

\subsection{Results for factored problems}

We will summarize the results around our factored problems below. 
They appear in sections \ref{sec:compressedVariants} and \ref{sec:harderProblem}. We give a visual summary of our results in Figure \ref{fig:results}. We use the shortened names for many of the problems in the figure. 
The results will concern both counting and decision versions of our factored problems. The decision versions ask whether the count is nonzero, whereas the counting versions ask for the exact count. When we want the counting version, we will place $\#$ in front of the name of the problem. See the Preliminaries (Section \ref{sec:Prelims}) for more details.
%For formal definitions of all these problems please see the Preliminaries Section \ref{sec:Prelims}. 

\begin{figure}[ht]
    \centering
    %The version without the line for xor:
    %\includegraphics[width=0.9\textwidth]{ResultSummary.eps}
    \includegraphics[width=0.9\textwidth]{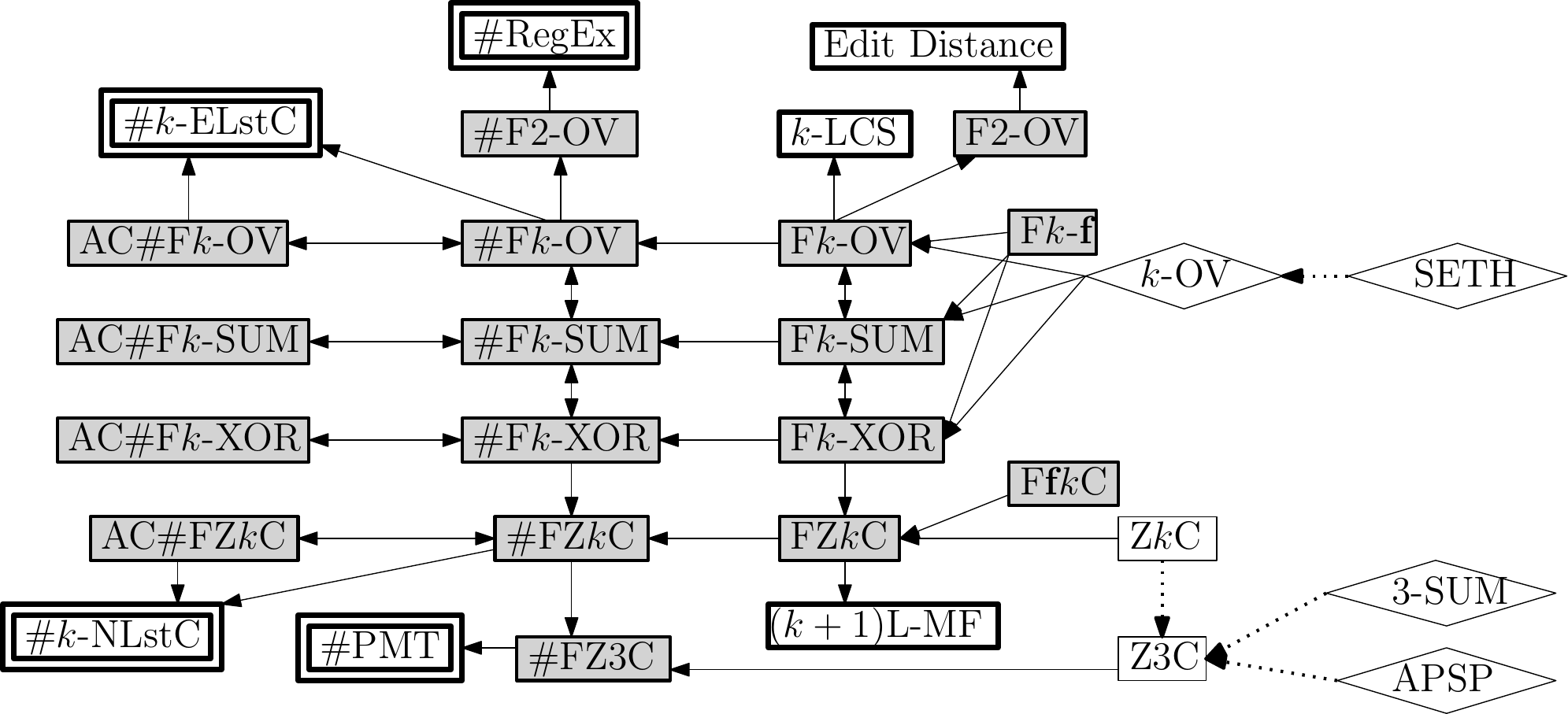}
    \caption{A summary of the reductions to and from factored problems in the paper. The problems in diamonds are the core problems of FGC. The full lines are reductions from this paper, while doted lines are pre-existing reductions. 
    The problems in gray boxes are our factored problems. The  problems in thick lined boxes are the problems we reduce from factored problems. For the problems surrounded by a thick-lined double box we have generated an explicit average case distributions on which they are hard (but it is not the uniform distribution).
    These results appear in Sections \ref{sec:compressedVariants} and \ref{sec:harderProblem}.}
    \label{fig:results}
\end{figure}

\paragraph{Summary.}
We first provide an overview summary of our results.

First we show that the factored versions of $k$-OV, $k$-SUM and $k$-XOR are all $n^{k-o(1)}$-fine-grained hard under SETH. We also show that the factored version of zero-$3$-clique (\fzkc[3]) is $n^{3-o(1)}$-fine-grained hard based on any of the three core hypotheses of FGC (SETH, or the APSP or $3$-SUM hypothesis).  Additionally, we show that the counting versions of these factored problems are as hard in their natural uniform average-case as they are in the worst case. Moreover, we show that many natural problems, like counting regular expression matchings, reduce from our factored problems. This even implies fine-grained average-case hardness for these problems over some explicit distributions. 

Thus our factored problems do three things simultaneously:
 %We will show that these factored problems are hard from SETH. We will even give a particular factored problem, factored zero-$3$-clique (\fzkc[3]), that is $n^{3-o(1)}$ hard if any of the three core hypotheses of FGC are true.  Additionally, we will show that the counting versions of these factored problems are as hard in their natural uniform average-case as they are in the worst case. We show that many natural problems, like counting regular expression matches, reduce from our factored problems (including average-case hardness over some explicit distribution). So these factored problems we introduce do three things simultaneously:
\begin{itemize}
\item Instead of trying to use the uniform average-case of the core problems of FGC as central problems in a network of average-case reductions, we can use the factored versions of the core problems in FGC. For example, the counting variant of factored OV (\#\ckov[2]) is hard in its uniform average case from the worst-case OV hypothesis. 
 Generically, our factored problems serve as an alternative central problem for average-case hardness. To demonstrate this, in Section \ref{sec:harderProblem}, we give reductions from counting factored problems to four problems in graph algorithms and sequence alignment (including counting regular expression matchings).
 
  \item The factored versions of the core problems are sufficiently expressive that they are complete for the large class of factored problems. In particular, \ckov[], \ckxor[], and \cksum~are complete for the class of problems of the form \ckfunc[]~over all $f$, while \fzkc[]~is complete for the class of problems \cfkc[]~over all $f$. 
  Despite this expressiveness we are still able to reduce our factored problems to many natural problems. In section \ref{sec:harderProblem} we give fine-grained reductions from our factored problems to  $k$-LCS, Edit Distance and a labeled version of Max Flow.  
 
\item Abboud et al. \cite{matchingTriangles} gave two problems, Triangle Collection and Matching Triangles that are hard from all three core assumptions in FGC. They also showed that one can reduce  Triangle Collection \footnote{Actually a version of the problem that is still hard under all three assumptions.} to several natural problems in graph algorithms. Unfortunately, however, neither Triangle Collection, nor Matching Triangles are known to be hard on average.
One of our factored problems, \fzkc[3]~is also hard from all three core assumptions. Moreover, the counting version of \fzkc[3]~is additionally $n^{3-o(1)}$ hard in the average-case from all three core assumptions of FGC. Thus, problems that reduce from counting \fzkc[3]~get average-case hardness for some explicit average-case distribution. We give two examples of problems that reduce from counting \fzkc[3]~in Section \ref{sec:harderProblem}. Hence if you are interested in average-case hardness then counting \fzkc[3]~might be a better source for reductions than, say Matching Triangles or Triangle Collection.

\end{itemize}

\paragraph{Fine-grained hardness for factored problems.}
Here we show that our factored problems are fine-grained hard under standard FGC hypotheses.

We first show that a single call to a factored problem solves its non-factored counterpoint. 
\begin{theorem} \label{thm:unfactoredToFactored}
In $O(n)$ time, one can reduce an instance of size $n$ of $k$-OV, $k$-XOR, $k$-SUM and Z$k$C to a single call to an instance of size $\tO(n)$  of \ckov[], \ckxor[], \cksum[]~and \czkc, respectively.
\end{theorem}

The above theorem holds both in the decision and counting context. It gives fine-grained hardness for the factored variants of all our problems, under the hypothesis that the original variants are hard.  Note that  $k$-XOR, $k$-SUM have $\tO(n^{ \lceil k/2 \rceil})$ time algorithms. However, we have $n^{k-o(1)}$ conditional lower bounds for all of \ckov[], \ckxor[], \cksum[]~and \czkc. So, while we do get fine-grained hardness from the $k$-XOR and $k$-SUM hypotheses, this hardness is not tight. The hardness is tight from the $k$-OV and Z$k$C hypotheses however. 

Now we give fine-grained hardness for \czkc[3]~under all three core hypotheses from FGC.
\begin{theorem}
If \czkc[3]~(even for $b=o(\log n)$ and $g=o(\log(n)/\log\log(n))$) can be solved in $O(n^{3-\eps})$ time for some constant $\eps>0$, then SETH is false, and there exists a constant $\eps'>0$ such that $3$-SUM can be solved in $O(n^{2-\eps'})$ time and APSP can be solved in $O(n^{3-\eps'})$ time.
\label{thm:fzkcIsHardFromAll}
\end{theorem}

%Thus, \czkc[3]~is a very hard problem.

%\xxx{TODO: cover the results of sections \ref{sec:compressedVariants} and \ref{sec:harderProblem} here}

%\paragraph{Theorems from Section \ref{sec:compressedVariants}}

\paragraph{Worst-case to average-case reductions for factored problems.}
We show that our factored problems admit fine-grained worst-case to average-case reductions. Our first theorem about this is a worst-case to average-case fine-grained reduction for the counting version of \ckfunc~for a natural distribution (defined in Definition \ref{def:muavgCasefactored}). The proof appears in Section \ref{sec:compressedVariants}.

\begin{theorem}
\label{thm:worstToAvgckf}
Let $\mu$ be a constant such that $0 < \mu <1$. 
Suppose that average-case \#\ckfunc[]$^{\mu}$~(see definition \ref{def:muavgCasefactored}, this is an iid distribution which has ones with probability $\mu$) can be solved in time $T(n)$ with probability at least $1-1/(lg(n)^{kg}\lg\lg(n)^{kg})$. Then worst-case \#\ckfunc[]~can be solved in time $\tO(T(n))$ \footnote{Note that given that $g= o(\lg(n)/\lg\lg(n)) $ then a probability of $1-1/n^{\epsilon}$ will be high enough for any $\epsilon>0$.}.

When $\mu =1/2$ average-case \#\ckfunc[]$^{\mu}$ is average-case \#\ckfunc[]. 
\end{theorem}
 
 Thus, if we have worst-case fine-grained hardness for \#\ckfunc~for some $f$, then we get average-case hardness for the same problem over a natural distribution. In particular, in the corollary below we obtain average-case hardness for \#\ckov[], \#\cksum[], \#\ckxor[], based on the standard FGC hardness of \kov, \ksum, \kxor~(as implied by Theorem~\ref{thm:unfactoredToFactored}).
 
\begin{corollary}
\label{cor:worstToAvg}
%By Theorem \ref{thm:worstToAvgckf}, we have the following result: 
If average-case \#\ckov[]~can be solved in time $T(n)$ with probability $1-1/(lg(n)^{gk}\lg\lg(n)^{gk})$  then worst-case \#\ckov[]~ can be solved in time $\tO(T(n))$ \footnotemark[\value{footnote}].

If average-case \#\cksum[] can be solved in time $T(n)$ with probability $1-1/(lg(n)^{gk}\lg\lg(n)^{gk})$  then worst-case \#\cksum[]~can be solved in time $\tO(T(n))$ \footnotemark[\value{footnote}].

If average-case \#\ckxor[]~can be solved in time $T(n)$ with probability $1-1/(lg(n)^{gk}\lg\lg(n)^{gk})$  then worst-case \#\ckxor[]~can be solved in time $\tO(T(n))$ \footnotemark[\value{footnote}].

\end{corollary}

Similarly, we obtain fine-grained average-case hardness for \#\cfkc, based on the fine-grained worst-case hardness of \#\cfkc. 

\begin{theorem}
\label{thm:worstToAvgclique}
%We will use the definition of average-case \#\cfkc[]$^\mu$ from Definition \ref{def:muavgCasefactored}. 
Let $\mu$ be a constant and $0<\mu<1$.
If average-case \#\cfkc[]$^{\mu}$~(see Definition \ref{def:muavgCasefactored}, this is an iid distribution which has ones with probability $\mu$) can be solved in time $T(n)$ with probability $1-1/(lg(n)^{k^2g}\lg\lg(n)^{k^2g})$  then worst-case \#\cfkc[]~can be solved in time $\tO(T(n))$ \footnotemark[\value{footnote}].

When $\mu=1/2$ average-case \#\cfkc[]$^{\mu}$ is average-case \#\cfkc[].
\end{theorem}

By Theorem \ref{thm:worstToAvgclique}, we have the following result for \#\czkc[]~in particular.
 
\begin{corollary}
\label{cor:worstToAvgClique}
If average-case \#\czkc[]~can be solved in time $T(n)$ with probability $1-1/(lg(n)^{k^2g}\lg\lg(n)^{k^2g})$ then worst-case \#\czkc[]~can be solved in time $\tO(T(n))$ \footnotemark[\value{footnote}].
\end{corollary} 

Thus in particular we obtain fine-grained average-case hardness for counting factored zero-$3$-cliques, based on the hardness of zero-$3$-clique, and thus based on the APSP and $3$-SUM hypotheses.

\paragraph{Completeness for \ckov[],\cksum[], \ckxor[]~and \czkc[].}

Let $k\geq 2$ be a fixed integer. Consider the class of problems \ckfunc[]~defined over all boolean functions $f$ on $kb$-length inputs. Our first sequence of results show that \ckov[],\cksum[]~and  \ckxor[]~are complete for the class, so that a $T(n)$ time algorithm for any of these problems would imply an $\tilde{O}(T(n))$ time algorithm for \ckfunc[]~for any $f$.

To prove this, we first show that \ckxor[]~is complete for the class:
\begin{theorem}
\label{thm:ckxorIscomplete}
If we can solve \#\ckxor[]~with $g$ sets of $k^3b$ length vectors in time $T(n)$ then, for any $f$, we can solve a \#\ckfunc[]~instance with $g$ sets of $b$ length vectors in time $T(n)+\tilde{O}(n)$ time. 
%\label{thm:ckxorIscomplete}
\end{theorem}

We then show that \ckov[], \cksum[]~and \ckxor[]~are equivalent.
\begin{theorem}
\label{thm:compressedKVarEquiv}
If any of \#\ckov[], \#\cksum[], or \#\ckxor[]~can be solved in $T(n)$ time then all of \#\ckov, \#\cksum[], and \#\ckxor[]~can be solved in $\tO(T(n))$ time.
%\label{thm:compressedKVarEquiv}
\end{theorem}

The above two theorems imply the final completeness theorem:
\begin{theorem}
\label{thm:allCompressedAreComplete}
If any of \#\ckov[], \#\cksum[], or \#\ckxor[]~can be solved in $T(n)$ time then \#\ckfunc[]~can be solved in $\tO(T(n))$ time.
%\label{cor:allCompressedAreComplete}
\end{theorem}

We also consider the class of problems (\#)\cfkc[]~defined by Boolean functions $f$ on $\binom{k}{2}b$-length inputs. We show that (\#)\czkc~is complete for this class.

\begin{theorem}
If (\#)\czkc~can be solved in $T(n)$ time then (\#)\cfkc[]~for any $f$, can be solved in $\tO(T(n)+n^2)$ time.
\label{thm:fzkcIsComplete}
\end{theorem}

Thus our factored problems corresponding to core problems in FGC, are the hard problems for natural classes of factored problems.

\paragraph{Fine-grained hardness for well-studied problems, based on the hardness of factored problems.}
The results we mention here appear in Section \ref{sec:harderProblem}. The main upshot is that the factored problems are both hard and also simple enough to imply hardness for basic problems in graph and string algorithms.
Some of the results are based on the hardness of \fzkc[3]~which implies hardness from all of SETH, 3-SUM and APSP. Some come from \ckfunc[]~which implies hardness from SETH. 

% \paragraph{Theorems from Section \ref{sec:harderProblem}}
% In this section we give a bunch of results. Some work for counting and thus also average case. Some come from \fzkc[3]~which implies hardness from all of SETH, 3-SUM and APSP. Some come from \ckfunc[]~which implies hardness from SETH. 

%\textbf{\emph{4.1 Partitioned Matching Triangles Problem:}}
\noindent \textbf{\emph{Partitioned Matching Triangles.}}
First we define the {\em Partitioned Matching Triangles} problem (\NDMT) as follows: Given $g$ disjoint $n$-node graphs with node colors, is there a triple of colors $a,b,c$ so that every one of the $g$ graphs contains a triangle whose nodes are colored by $a,b,c$? The counting variant of \NDMT~is to count the total number of such $g$-tuples of colored triangles.
%the total number of color triples that have a triangle in each of the $g$ graphs. 

Abboud et al. \cite{matchingTriangles} consider the related Matching Triangles problem mentioned earlier in the introduction, and show that it is hard from all three core FGC hypotheses. In the Matching Triangles problem one is given an integer $T$ and a node-colored graph $G$ and one wants to know if there is a triple of colors $a,b,c$ so that there are at least $T$ triangles in $G$ colored by $a,b,c$. 

We observe first that for the particular parameters for which Matching Triangles is shown to be hard in \cite{matchingTriangles}, one can actually reduce Matching Triangles in a fine-grained way to Partitioned Matching Triangles (\NDMT), so that the latter problem is also hard from all three hypothesis. Furthermore, we give a powerful reduction to \NDMT~from \czkc[3]. Moreover, our reduction also holds between the counting versions of the problems, so that we get fine-grained {\em average-case} hardness for counting \NDMT~under all three hypotheses as well.
% Abboud et al. \cite{} implicitly show that \NDMT is hard from SETH, the $3$-SUM Hypothesis and the APSP Hypothesis, in their proof that Matching Triangles is hard. Here we show that \NDMT is also hard from \czkc[3].***

\begin{theorem}
\label{thm:czkc-to-ndmt}
If (\#)Partitioned Matching Triangles can be solved in $T(n)$ time, then we can solve (\#)\czkc[3]~in time $\tilde{O}(T(n)+n^2)$.
\end{theorem}

% \textbf{\emph{4.2 $k$-color Node Labeled st Connectivity :}}
\noindent\textbf{\emph{$\bm{k}$-color Node Labeled $st$ Connectivity.}}
In the $k$-color Node Labeled $st$ Connectivity Problem (\kNLstC) one is given an acyclic graph $G=(V,E)$ with two designated nodes $s,t\in V$, and colors on all nodes in $V\setminus \{s,t\}$ from a set of colors $C$. One is then asked whether there is a path from $s$ to $t$ in $G$ using at most $k$ node colors.

We give a fine-grained reduction from \czkc[]~to \kNLstC~that also holds between the counting versions. Here in the counting version of \kNLstC~we want to output the number of $s$-$t$ paths through at most $k$ colors, mod $\lceil2^{2k\lg^2(n)}\rceil$.

\begin{theorem} \label{thm:knlstc-hardness}
%\label{thm:knlstc-hardness}
If a $O(|C|^{k-2}|E|^{1-\epsilon/2})$ or $O(|C|^{k-2-\epsilon}|E|)$ time algorithm exists for (counting mod $2^{2k\lg^2(n)}$) \kNLstC~then a $O(n^{k-\epsilon})$ algorithm exists for (\#)\czkc.
\end{theorem}

The conditional lower bound of $(|C|^{k-2}|E|)^{1-o(1)}$ resulting from the above theorem is tight. In Appendix \ref{app:algorithms}~we give the corresponding algorithm.
%\begin{theorem}
%(Counting mod $R$) \kNLstC~has a $\tO(|C|^k + |C|^{k-2}|E|)$ time algorithm for all $k\geq 2$ (when $\lg(R)$ is sub-polynomial).
%\end{theorem}

\noindent\textbf{\emph{$\bm{k}$-color Edge Labeled $st$ Connectivity.}}
The $k$-color Edge Labeled $st$ Connectivity problem (\kELstC[]) asks for a given acyclic graph with colored edges and given source $s$ and target $t$, if there is a path from $s$ to $t$ that uses only $k$ colors of edges. 

We give conditional hardness for both the decision and counting version of the problem (where the counts are mod a small $R$). This also implies average-case hardness for the counting mod $R$ problem under all three hardness hypotheses of FGC. 

\begin{theorem}
\label{kelstcCKFUNC}
If a $\tilde{O}(|E||C|^{k-1-\epsilon})$ or $\tilde{O}(|E|^{1-\epsilon}|C|^{k-1})$ time algorithm exists for (counting mod $ 2^{2k\lg^2(n)}$) \kELstC[], then a $\tilde{O}(n^{k-\epsilon})$ algorithm exists for (\#)\ckfunc[].
%can solve \kELstC[]~(\#\kELstC[]~mod $R = \Omega( 2^{2k\lg^2(n)})$) in time $\tilde{O}(|E||C|^{k-1-\epsilon})$ or $\tilde{O}(|E|^{1-\epsilon}|C|^{k-1})$ then you can solve \ckfunc[]~(\#\ckfunc[])~in time $\tilde{O}(n^{k-\epsilon})$ for any constant $\epsilon>0$.
\end{theorem}

This is tight.  Note this algorithm is slower (by a factor of $|C|$) than the node-labeled version, however it is optimal. The corresponding algorithm is in a theorem from Appendix \ref{app:algorithms}.

%\begin{theorem}
%\label{thm:countkelstc}
%There is an algorithm for (counting mod $R$) \kELstC[]~st-reachability that runs in time $\tilde{O}(|C|^{k-1}|E|)$ (when $\lg(R) =n^{o(1)}$).
%\end{theorem}

%\noindent\textbf{\emph{4.3 (k+1)-labeled-max-flow }}
\noindent\textbf{\emph{$\bm{(k+1)}$ Labeled Max Flow.}}
The $(k+1)$ Labeled Max Flow problem studied in \cite{mfml}
 asks, given a capacitated graph $G=(V,E)$ where the edges have colors, and $s,t\in V$, if there is a maximum flow from the source $s$ to the sink $t$ where number of distinct colors of the edges with non-zero flow is at most $k+1$.

\begin{theorem}
\label{klmf}
If $(k+1)$L-MF can be solved in $T(n)$ time, then we can solve \czkc~in time $\tilde{O}(T(n)+n^2)$.
\end{theorem}

This implies an $n^{k-1-o(1)}$ lower bound for $k$L-MF under all three FGC hypotheses.
We also show that for the particular structured version of the problem given in our reduction, this lower bound is tight.

\noindent\textbf{\emph{Regular Expression Matching.}}
The Regular Expression Matching problem (studied e.g. in \cite{regularexpression}) takes as input a regular expression (pattern) $p$ of size $m$ and a sequence of symbols (text) $t$ of length $n$, and asks if there is a substring of $t$ that can be derived from $p$.
The counting version of the problem, \#Regular Expression Matching asks for the number of subset alignments of the pattern in the text mod an integer $R$, where $R=n^{o(1)}$. A classic algorithm constructs and simulates a non-deterministic finite automaton corresponding to the expression, resulting in the rectangular $O(mn)$ running time for the detection version of the problem.

We give hardness from \#\ckov[2]~(mod $R$) which in turn implies average-case fine-grained hardness for counting regular expression matchings mod $R$, from SETH.

\begin{theorem}
Let $R$ be an integer where $\lg{(R)}$ is subpolynomial. If you can solve (\# mod ${R}$) regular expression matching in $T(n)$ time, then you can solve (\# mod $R$) \ckov[2]~in $\tO(T(n)+n)$ time
\label{thm:regex}
\end{theorem}
Again, we show in Appendix \ref{app:algorithms} that for the particular ``type'' of pattern used in our reduction, this lower bound is tight.

\noindent\textbf{\emph{LCS and Edit Distance.}}
The $k$-LCS problem is a basic problem in sequence alignment. Given $k$ sequences $s_1,\ldots,s_k$ of length $n$, one is asked to find the longest sequence that appears in every $s_i$ as a subsequence. 
\klcs~can be solved in $O(n^k)$ time with dynamic programming and requires $n^{k-o(1)}$ time under SETH, via a reduction from \kov~\cite{LCSisHard}. Here we show that $k$-LCS is also fine-grained hard via a reduction from \ckov[].

\begin{theorem}
\label{thm:kLCSisWorstCaseHard}
A $T(n)$ time algorithm for \klcs~with alphabet size $O(k)$ implies a $\tO(T(n))$ algorithm for \ckov[].
\end{theorem}

The Edit Distance problem is another famous sequence alignment problem. Here one is given two $n$ length sequences $a$ and $b$ and one needs to compute the minimum number of symbol insertions, deletions and substitutions needed to transform $a$ into $b$.
Edit Distance can be solved in $O(n^2)$ time via dynamic programming, and requires $n^{2-o(1)}$ time under SETH, via a reduction from OV \cite{BackursI15,BackursI18}.

In section \ref{sec:harderProblem} we show that edit distance is also fine-grained hard from \ckov[2].

\begin{theorem}
\label{thm:editDistisWorstCaseHard}
A $T(n)$ time algorithm for Edit Distance implies a $\tO(T(n))$ algorithm for \ckov[2].
\end{theorem}

%% trying to cover OV results
\subsubsection{Counting OV is Easy on Average}
As mentioned earlier in the introduction we show that counting orthogonal vectors over the uniform distribution is easy in the average-case. Let \dOVm~be the problem of solving orthogonal vectors on instances generated by sampling $n$ vectors iid from the distribution over $d$ bit vectors where every bit in the vector is sampled iid from the distribution that returns $1$ with probability $\mu$ and returns $0$ with probability $1-\mu$.

\begin{theorem}
\label{thm:avgCaseCountOVisEasy}
For all constant values of $\mu$ and all values of $d$ there exists constants $\epsilon>0$ and $\delta>0$ such that there is an algorithm for \dOVm~that runs in time $\tO(n^{2-\delta})$ with probability at least $1-n^{-\epsilon}$.
\end{theorem}

%% Covering the count to detect results
\subsubsection{Counting to Detection for \texorpdfstring{\zkc}{Lg}}
Our worst-case to average-case reductions show hardness for counting problems. We mentioned earlier in the introduction that stronger cryptographic primitives have been built from detection problems than from counting problems. 
In this paper we show that in the sufficiently low error regime there is a counting to detection reduction for the zero-$k$-clique problem. Unfortunately, this does \emph{not} give a fine-grained one-way function from worst-case assumptions. However, it makes progress towards bridging the gap between the problems we can show hard from the worst-case and those we can build powerful cryptographic primitives from. 

\begin{definition}
An average case instance of \zkc~(\aczkc) with range $R$ takes as input a complete $k$-partite graph with $n$ nodes in each partition. Every edge has a weight chosen iid from $[0,R-1]$. A clique is considered a zero $k$ clique if the sum of the edges is zero mod $R$.
\end{definition}

%\begin{theorem}
%\label{thm:countToDecision}
%Let $p(n)$ be a monotonically non-increasing function. %%
%Given a decision algorithm for average-case \zkc~that runs in time $O(n^{k-\eps})$ for some $\eps>0$ and succeeds with probability at least $1-p(n)$, there is a counting algorithm that runs in $O(n^{k-\eps'})$ for some $\eps'>0$ and succeeds with probability at least $1 -2^{-\lg^2(n)}-p(n^{1/25})n^k$.
%\end{theorem}

%We can plug in a particular function for the error term, $p(n)$, to make a cleaner statement about the connection between counting and detection in the low error regime. 

\begin{theorem}
\label{thm:countToDecision}
Given a decision algorithm for \aczkc~that runs in time $O(n^{k-\eps})$ for some $\eps>0$ and succeeds with probability at least $1-n^{-\omega(1)}$, there is a counting algorithm that runs in $O(n^{k-\eps'})$ time for some $\eps'>0$ and succeeds with probability at least $1-n^{-\omega(1)}$, where $\omega(1)$ here means any function that is asymptotically larger than constant.  
\end{theorem}

%%Covering both the framework appendix and the average case hardness for counting subgraphs H in one fell swoop
\subsubsection{Worst-Case to Average-Case Reductions}
We define the notion of a \goodPoly~for the problem $P$ (a \gPol{$P$}). We define the properties of a \goodPoly~in Definition \ref{def:goodPoly}. Intuitively these properties are that the function must be low degree, count the output of the problem, and have well structured monomials.
We show that any problem $P$ that has a \gPol{$P$} is hard in its uniform average case in appendix \ref{sec:Framework}. We do this using techniques from Boix{-}Adser{\`{a}} et al \cite{UniformCliqueABB}. We use the \gPol{$\cdot$} framework to show uniform average-case hardness for our counting factored problems (in section \ref{sec:compressedVariants}). We give the framework theorem statement below.

%In Boix et. al they show fine-grained average case hardness for counting $k$-cliques in Erd{\H{o}}s-R{\'{e}}nyi graphs \cite{UniformCliqueABB}.  We take the techniques they use and generalize this into a framework. 

%We define a ``good low-degree polynomial'' and prove that for any counting problem, $\#P$ that can be represented by a polynomial over a finite field that meets our definition of a good low-degree polynomial, we can show hardness for average-case $\#P$. Specifically, over the totally uniform average case (every bit of the input is selected iid with some constant probability $p$) if you can solve $\#P$ with probability $1-1/n^{\Omega(1)}$ in time $T(n)$ you can solve $\#P$ in the worst case in time $\tO(T(n))$ with probability at least $1-1/n^{\lg(n)}$. 

%We use this framework for the average-case hardness of our factored problems as well as the average-case harness of counting sub-graphs (even non-clique subgraphs). This framework makes applying the ideas from Boix et. al (\cite{UniformCliqueABB}) to new problems easy and fast. We hope others will find this as well. 

\begin{theorem}
\label{thm:framework}
Let $\mu$ be a constant such that $0 < \mu <1$.
Let $P$ be a problem such that a function $f$ exists that is a \gPol{$P$}, and let $d$ be the degree of $f$. 
Let $A$ be an algorithm  that runs in time $T(n)$ such that when $\vec{I}$ is formed by $n$ bits each chosen iid from $Ber[\mu]$:
$$Pr[A(\vec{I}) = P(\vec{I})] \geq 1-1/\omega\left( \lg^d(n)\lg\lg^d(n) \right).$$
Then there is a randomized algorithm $B$ that runs in time $\tO(n + T(n))$ such that
for \emph{any} for $\vec{I} \in \{0,1\}^n$:
$$Pr[B(\vec{I}) = P(\vec{I})] \geq 1-O\left(2^{-\lg^2(n)}\right).$$
\end{theorem}

Boix{-}Adser{\`{a}} et al show that counting $k$ cliques is as hard in Erd{\H{o}}s-R{\'{e}}nyi graphs as it is in the worst case. We use the \gPol{$\cdot$}~framework a second time to slightly generalize their result to show that counting any subgraph $H$ in an Erd{\H{o}}s-R{\'{e}}nyi graph is at least as hard as counting subgraphs $H$ in worst case $k$-partite graphs (in section \ref{sec:SubgraphsWithKNodes}). 
%The intuitive version of the below statement is the following. Let $A$ be  a $T(n)$ time algorithm for counting the number of subgraphs $H$ in an Erd{\H{o}}s-R{\'{e}}nyi graph, $G$ with probability $1- \epsilon/\tilde{\Omega}(1)$. Given algorithm $A$ there is a  $\tO(T(n) + |G|)$ time algorithm for worst case counting of subgraphs $H$ that succeeds with probability $1-O(\epsilon)$.  

\begin{theorem}
\label{thm:ACSubgraphCountToWCSubgraphCount}
%Let $G$ be a $n$ node Erd{\H{o}}s-R{\'{e}}nyi graph. 
%Let $1/b$ be the edge probability where $b$ is a constant. 
Let $H$ have $e$ edges and $k$ vertices where $k =o(\sqrt{\lg(n)})$. 
Let $A$ be an average-case algorithm for counting subgraphs $H$ in Erd{\H{o}}s-R{\'{e}}nyi graphs with edge probability $1/b$ which takes $T(n)$ time with probability $1-2^{-2k} \cdot b^{-k^2} \cdot (\lg(e)\lg\lg(e))^{-\omega(1)}$.

Then an algorithm exists to count subgraphs $H$ in $k$-partite graphs in time $\tilde{O}(T(n))$ with probability at least $1-\tO(2^{-\lg^2(n)})$. 
\end{theorem}

\subsection{Organization of the Paper}
In the preliminaries section \ref{sec:Prelims} we give a formal definition of our factored problems. We also define the problems that we use throughout the paper, and we give an introduction of the average-case framework which is defined formally in Appendix \ref{sec:Framework}.
%We give preliminaries in section \ref{sec:Prelims}.
We show that the factored problems are hard, and give the worst-case to average-case reductions for the factored problems in section \ref{sec:compressedVariants}. 
%In section \ref{sec:compressedVariants}, we show how different factored problems are connected to each other, as well as how non-factored versions of each problem reduces to its factored version. 
In section \ref{sec:harderProblem}, we show that our factored problems can show hardness for many natural non-factored problems.
We use the same framework that gives average-case hardness for the factored problems to show that counting arbitrary subgraphs in random graphs is hard in section \ref{sec:SubgraphsWithKNodes}. We give a fast algorithm for counting OV over the uniform average-case in section \ref{sec:OV}. We give counting to detection reduction for average-case zero-$k$-clique with high probability in section  \ref{sec:countToDetect}. Finally, we list problems that seem like promising future work in section \ref{sec:FutureWork}.
%In section \ref{sec:SubgraphsWithKNodes}, we give an average case hardness result for subgraph detection for any $k$-node subgraphs, and finally in section \ref{sec:OV} we give a truly sub-quadratic algorithm for average case counting OV 

We give the efficient algorithms for our factored problems and the problems that reduce from our factored problems in appendix \ref{app:algorithms}. We give the framework that generalizes the techniques of Boix{-}Adser{\`{a}} et al. in appendix \ref{sec:Framework}.

\section{Preliminaries}
\label{sec:Prelims}

We cover useful preliminaries for sections \ref{sec:compressedVariants} and \ref{sec:harderProblem} in this section. We include preliminaries for Section  \ref{sec:SubgraphsWithKNodes}, Appendix \ref{sec:Framework}, and proofs of algorithm running times in Appendix \ref{app:algorithms}.

%\xxx{add definitions of the framework cause they are used later [Andrea: Which parts do we want? The definitions of good polynomials and the final theorem statement?]}

\subsection{Hypotheses about Core Problems of Fine-Grained Complexity}

\begin{definition} [The $3$-SUM Hypothesis \cite{C3sum}]
	In the \emph{\ksum[]~problem}, we are given an unsorted list $L$ of $n$ values (over $\Z$ or $\R$) and want to determine if there are $a_1, \ldots, a_k \in L$ such that $\sum_{i=1}^k a_i = 0$. The counting version of \ksum[]~asks how many sets of $k$ numbers $a_1, \ldots, a_k \in L$ sum to zero. 
	
The \ksum[]~hypothesis states that that the \ksum[]~problem requires $n^{\lceil k/2 \rceil -o(1)}$ time \cite{C3sum}.

This is equivalent to saying no $n^{\lceil k/2 \rceil-\eps}$ time algorithm exists for \ksum[]~for constant $\eps>0$. 

	\label{def:ksum}
\end{definition}

\begin{definition}[APSP Hypothesis \cite{CAPSP}]
APSP takes as input a graph $G$ with $n$ nodes (vertices), $V$ and $m$ edges, $E$. These edges are given weights in $[-R,R]$ where $R = O(n^c)$ for some constant $c$. We must return the shortest path length for every pair of vertices $u,v \in V$. The length of a path is the sum of the edge weights for all edges on that path.

The APSP Hypothesis states that the APSP problem requires $n^{3-o(1)}$ time when $m=\Omega(n^2)$.
	
\label{def:apsp}
\end{definition}

\begin{definition}[Strong Exponential Time Hypothesis (SETH) \cite{cseth}]
Let $c_k$ be the smallest constant such that there is an algorithm for $k$-CNF SAT that runs in $O(2^{c_k n +o(n)})$ time.

SETH states that there is no constant $\eps>0$ such that $c_k \leq 1-\eps$ for all constant $k$.
\end{definition}

Intuitively SETH states that there is no constant $\epsilon>0$ such that there is a $O(2^{n(1-\epsilon)})$ time algorithm for $k$-CNF SAT for all constant values of $k$.

\begin{definition} [The $k$-OV Hypothesis \cite{ryanThesis}]
In the \emph{\kOV~problem}, we are given $k$ unsorted lists $L_1,\ldots, L_k$ of $n$ zero-one vectors of length $d$ as input. If there are $k$ vectors $v_1 \in L_1,\ldots,v_k \in L_k$ such that for $\forall i\in[1,d]~\exists j\in[1,k]$ such that $v_i[j]=0$ we call these $k$ vectors an orthogonal $k$-tuple. One should return true if there is an orthogonal $k$-tuple in the input. The counting version of \kOV~(\#\kOV) asks for the number of orthogonal $k$-tuples. 
	
The \kOV~hypothesis states that that the \kOV~problem requires $n^{k-o(1)}$ time \cite{ryanThesis}.

This is equivalent to saying no $O(n^{k-\eps})$ time algorithm exists for \kOV for constant $\eps>0$. 

	\label{def:kOV}
\end{definition}
\subsection{Graphs}

\begin{definition}
Let $H=(V_H,E_H)$ be a $k$-node graph with $V_H=\{x_1,\ldots,x_k\}$. 

An $H$-partite graph is a graph with $k$ partitions $V_1,\ldots,V_k$. This graph must only have edges between nodes $v_i \in V_i$ and $v_j \in V_j$ if e $(x_i,x_j)\in E_H$. (See Figure \ref{fig:Hpartite})
\end{definition}

\begin{figure}[ht]
\centering
\includegraphics[width=0.8\textwidth]{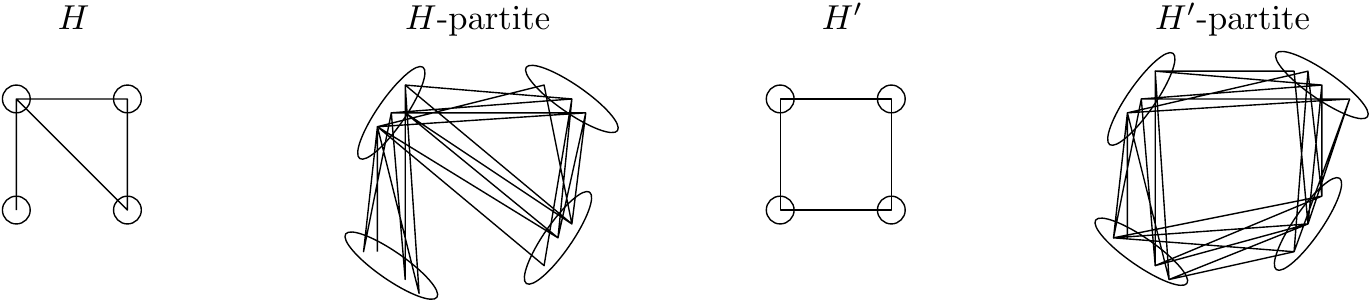}
\caption{An example of the corresponding $H$-partite graphs. }
\label{fig:Hpartite}
\end{figure}

\subsection{\GoodPolys}
We define the \goodPoly~for a problem $P$ ($\gPol{P}$). In Appendix \ref{sec:Framework} we provide a framework which shows that if a problem $P$ has a  $\gPol{P}$ then $P$ is hard over the uniform average case. The proof of this framework is a generalization of the proof in Boix et al. \cite{UniformCliqueABB}. We use this to show average-case hardness for counting versions of factored problems and counting subgraphs in sections \ref{sec:compressedVariants} and \ref{sec:SubgraphsWithKNodes} respectively. 

\begin{definition}
Let the polynomial  $f$ have $n$ inputs $x_1,\ldots, x_n$.
%We say $f$ is strongly $d$-partite if one can partition the inputs into $d$ sets $S_1,\ldots, S_d$ such that $f$ is can be written as $$f(\cdot) = \sum_{i\in [1,m]} g_{f}(x_{1,i}, \ldots, x_{d,i}).$$ Where every variable $x_{j,i}$ is from the partition $S_j$. 
We say $f$ is \textit{strongly $d$-partite} if one can partition the inputs into $d$ sets $S_1,\ldots, S_d$ such that $f$ can be written as a sum of monomials $\sum_i x_{1,i}\cdots x_{d,i}$, where every variable $x_{j,i}$ is from the partition $S_j$. 
That is, if there is a monomial $x_{i_1}^{c_1} \cdots x_{i_k}^{c_k}$ in $f$ then it must be that $c_j =1 $ and for all $j \ne \ell$ if $x_{i_j} \in S_m$ then $x_{i_\ell} \notin S_m$.
\end{definition}

\begin{definition}
Let $P(\vec{I})$ be the correct output for problem $P$ given input $\vec{I}$. 
\end{definition}

\begin{definition}
Let $n$ be the input size of the problem $P$, let $P$ return an integer in the range $[0,p-1]$ where $p$ is a prime and $p < n^c$ for some constant $c$. 
A \emph{\goodPoly} for problem $P$  (\gPol{P})  is a polynomial $f$ over a prime finite field $F_{p}$ where:
\begin{itemize}
    %\item There are at most polynomial monomials in $f$.
    \item If $\vec{I} = b_1,\ldots, b_n$, then 
    $f(b_1,\ldots, b_n) = f(\vec{I})=P(\vec{I})$ where $b_i$ maps to either a zero or a one in the prime finite field.
    \item The function $f$ has degree $d = o\left( \lg(n)/\lg\lg(n) \right)$.
    \item The function $f$ is strongly $d$-partite. 
\end{itemize}
\label{def:goodPoly}
\end{definition}

\subsection{Factored Problems}
\label{subsec:prelimsCompressed}

We introduce a more expressive extension of \kSUM, \kOV, \kXOR, and \zkc.
At a high level this extension takes every number or vector from the original problems and splits them up into $g =o(\lg(n)/\lg\lg(n))$ groups of numbers or vectors with bit representations of size $b =o(\lg(n))$. If the original numbers had length $\ell$, then $\ell \approx b\cdot g$. Then, we allow each group to contain multiple numbers or vectors. 

We start by giving a definition of \ckov[], then we give a small example of \ckov[2]. Next, we follow up with the analogously defined \cksum[],\ckxor[], and \czt. Finally, we give algorithms for these problems in the Appendix \ref{app:algorithms}.

\subsubsection{\texorpdfstring{F$k$-OV}{Lg}, Intuition and Examples}

\begin{definition}
A $(g,b)$-{\em factored vector} $v$ is defined by $g$ sets $(v[1],\ldots,v[g])$ where each $v[i]\subseteq \{0,1\}^b$ is a set of $b$-dimensional binary vectors.

%If $v$ is a factored vector let $v[i]$ be the $i^{th}$ set of vectors. The set of vectors will be a subset of $\{0,1\}^b$.\\

%It will be defined in terms of the orthogonality notion of \kOV. 
For a set of vectors $\vec{w_1},\ldots,\vec{w_k}$ of the same dimension $d$,
let $isOrthogonalTuple(\vec{w_1},\ldots,\vec{w_k})$ return $1$ iff $\vec{w_1},\ldots,\vec{w_k}$ are orthogonal, i.e. iff $\sum_{a=1}^d \prod_{j=1}^k w_j[a]~= 0$, where $w_j[a]$ is the $a^{th}$ bit of the vector $w_j$.

Now we define a useful operator, $\circ$ for a set $\{Z_1,\ldots, Z_k\}$ where each $Z_i$ is a set of $d$-dimensional binary vectors as follows. 
% are $k$-orthogonal vectors in the sense of the \kOV~problem. 
% $$\circ(v_1[i],\ldots,v_k[i]) =  \sum_{\vec{w_1}\in v_1[i]~\ldots \vec{w_k} \in v_k[i]} isOrthogonalTuple(\vec{w_1},\ldots,\vec{w_k}).$$
$$\circ(Z_1,\ldots,Z_k) :=  \sum_{\vec{w_1}\in Z_1, \ldots, \vec{w_k} \in Z_k} isOrthogonalTuple(\vec{w_1},\ldots,\vec{w_k}).$$

Now, given $k$ $(g,b)$-factored vectors $v_1,\ldots,v_k$ the number of orthogonal vectors within those factored vectors is $\circledcirc(v_1,\ldots,v_k) := \Pi_{i=0}^{g-1}\circ(v_1[i],\ldots,v_k[i]).$

The input to \ckov[]~is $V_1,\ldots, V_k$, where each $V_j$ is a set of $n$ $(g,b)$-factored vectors, where $g = o(\lg(n)/\lg\lg(n))$ and $b= o(\lg(n))$.
The total number of orthogonal vectors in a given \ckov[]~instance is 
$$\sum_{v_1,\ldots,v_k \in V_1,\ldots, V_k}\circledcirc(v_1,\ldots,v_k).$$

%On input $V_1,\ldots, V_k$, where each $V_j$ is a set of $n$ $(g,b)$-factored vectors (with $g = o(\lg(n)/\lg\lg(n))$ and $b= o(\lg(n))$), 
The \ckov[]~problem
asks to determine whether $\sum_{v_1,\ldots,v_k \in V_1,\ldots, V_k}\circledcirc(v_1,\ldots,v_k)>0$.
% returns true if the number of vectors in a \ckov[]~instance is non-zero. 
\end{definition}

\paragraph{An Example: }
We give a small example bellow. Consider \ckov[2]~where $g=2$ and $b={3}$. We give an example of factored vectors $u$, $v$ and $w$:
\begin{align*}
   &u[0] = \{001,010\} \text{~~} & u[1] = \{001,010\}\\
   &v[0] = \{000,010,110\} \text{~~} & v[1] = \{110,101\}\\
   &w[0] = \{\} \text{~~} & w[1] = \{000,011,100,111\}
\end{align*}

First, note that $\circledcirc(w,u) = \circledcirc(w,v) = 0$ trivially because $w[0]$ is the empty set. Empty sets are valid in this factored representation, but, rather degenerate. 
Next, note that $\circledcirc(v,u)$ is $4 \cdot 2 = 8$. For $\circ(u[0],v[0])$ all of $(001,000), (001,010), (001, 110),$ and $ (010,000)$ are orthogonal.  For $\circ(u[1],v[1])$ both $(001,110),$ and $(010,101)$ are orthogonal.  

\paragraph{A Natural Interpretation:} We can generate a $\kOV$ instance by interpreting a factored vector as representing $|v_1|\cdot \ldots \cdot |v_k|$ vectors. For example $u$ in the above example would represent the following list of vectors:
$$ 001001, 001010, 010001, 010010.$$
As another example $v$ would represent the following list of vectors:
$$ 000110, 000101, 010110, 010101, 110110, 110101.$$
Finally, $W$ represents no vectors, because $w[0]$ is the empty set. 

However, the number of vectors that can be represented by a single factored vector that has a $g2^b$ sized representation is $2^{bg}$. While $g2^b$ is sub-polynomial, $2^{bg}$ can be super polynomial (e.g. if $b=g=\lg(n)^{3/2}$)!

\subsubsection{Definitions for \texorpdfstring{F$k$-$\mathfrak{f}$}{Lg}, \texorpdfstring{F$k$-SUM}{Lg}, \texorpdfstring{F$k$-XOR}{Lg}, and FZT}

\begin{definition}
Let $f:(\{0,1\}^b)^{\times k} \rightarrow \{0,1\}$ be a function taking $k$ $b$-dimensional binary vectors to $\{0,1\}$. We can view $f$ as a Boolean function.

Let us define an operator for $f$, $\circ_f$, that takes $k$ factored vectors $a_1,\ldots,a_k$ and computes the number of $k$-tuples of vectors, one in each $a_i$, that $f$ accepts:
$$\circ_f(a_1,\ldots,a_k) =  \sum_{\vec{w_1}\in a_1 \ldots \vec{w_k} \in a_k} f(\vec{w_1},\ldots,\vec{w_k}).$$

If $v$ is a $(g,b)$-factored vector let, for $i\in [g]$, $v[i]$ be the $i^{th}$ set of vectors in $v$.

Given $(g,b)$-factored vectors $v_1,\ldots,v_k$ the number of $k$-tuples of vectors accepted by $f$ within those factored vectors is $\circledcirc_f(v_1,\ldots,v_k) = \Pi_{i=0}^{g-1}\circ_f(v_1[i],\ldots,v_k[i]).$

For each $f$, we define a problem \ckfunc[].
The input to \ckfunc[]~is $k$ sets, $V_1,\ldots, V_k$, of $n$ $(g,b)$-factored vectors each, where $g = o(\lg(n)/\lg\lg(n))$ and $b= o(\lg(n))$.  
%The set of vectors will be a subset of $\{0,1\}^b$.

The total number $k$-tuples of vectors accepted by $f$ in a given \ckfunc[]~instance is 
$$\ckfunct(V_1,\ldots,V_k):=\sum_{v_1,\ldots,v_k \in V_1,\ldots, V_k}\circledcirc_f(v_1,\ldots,v_k).$$

The \ckfunc[]~problem returns true iff $\ckfunct(V_1,\ldots,V_k)>0$.
%$\sum_{v_1,\ldots,v_k \in V_1,\ldots, V_k}\circledcirc_f(v_1,\ldots,v_k)>0.$ 
More generally, the counting version $\#$\ckfunc[]~of \ckfunc[]~asks to compute the quantity $\ckfunct(V_1,\ldots,V_k)$. %$\sum_{v_1,\ldots,v_k \in V_1,\ldots, V_k}\circledcirc_f(v_1,\ldots,v_k).$ 

% the number of vectors in a \ckfunc[]~instance is non-zero. 
\end{definition}

\begin{definition}
\ckxor[]~is the problem \ckfunc[]~where $f$ is $1$ if the componentwise XOR of the $k$ given vectors is the $0$ vector:
$$f(v_1, \ldots, v_k) =
\begin{cases}
1, & \text{if } v_1 \oplus \ldots \oplus v_k = \vec{0} \\
0, & \text{else}
\end{cases}.$$
\end{definition}

\begin{definition}
\cksum[]~is the problem \ckfunc[]~where $f$ that checks if the sum of the $k$ vectors is the $0$ vector:
$$f(v_1, \ldots, v_k) =
\begin{cases}
1, & \text{if } v_1 + \ldots + v_k = 0\\
0, & \text{else}
\end{cases}.$$
\end{definition}
%\cfkc

\begin{definition} For an integer $k$, $\ell=\binom{k}{2}$ and a given function $f:\{0,1\}^{b\ell}\rightarrow \{0,1\}$, construed as taking $\ell$-tuples of $b$-length binary vectors to $\{0,1\}$,
let $\#$\cfkc~be the problem of counting cliques in a graph whose edges are labeled with factored vectors, where a clique is counted with multiplicity the number of $\ell$-tuples of vectors that $f$ accepts and that appear in the $\ell$ factored vectors labeling the edges.

More formally, we change the definition of the operation $\circledcirc_f(\cdot)$ to take as input $k$ vertices $v_1,\ldots,v_k$ of a given graph $G=(V,E)$ whose edges $(x,y)\in E$ are labeled by $(g,b)$-factored vectors $e_{x,y}$:
$$
\circledcirc'_{f}(v_1,\ldots,v_k) = isClique(v_1,\ldots,v_k)\cdot\Pi_{i=0}^{g-1}\circ_f(e_{v_1,v_2}[i],e_{v_1,v_3}[i], \ldots,e_{v_{k-1},v_k}[i]).
$$

Above $isClique(v_1,\ldots,v_k)$ outputs $1$ if $v_1,\ldots,v_k$ form a $k$-clique in $G$, and otherwise outputs $0$. 

We keep the definition of $\circ_f(\cdot)$ the same as before, but now its input is a list of $\ell$ sets of vectors that are the $i$th group of vectors of the factored vectors labeling the clique edges: $$\circ_f(e_1[i],\ldots,e_{\ell}[i]) =  \sum_{\vec{w_1}\in e_1[i]~\ldots \vec{w_\ell} \in e_{\ell}[i]} f(\vec{w_1},\ldots,\vec{w_\ell}).$$

Finally, we let $\#$\cfkc be the problem of computing 
$$\textrm{F}fk\textrm{C}(G):=\sum_{v_1,\ldots,v_k\in V} \circledcirc'_f(v_1,\ldots,v_k).$$

Here, unlike for $\#$F$\ell$-f, we are only counting the sums of factored vectors when those factored vectors are on a set of $\ell=\binom{k}{2}$ edges that form a $k$ clique. Let \cfkc be the detection version of the problem that returns $1$ if $\textrm{F}fk\textrm{C}(G)>0$ and $0$ otherwise.  
%\xxx{TODO: Factored Zero k clique} 
\end{definition}

\begin{definition}
Factored Zero $k$-Clique, \czkc~is the \cfkc~problem~where $f$ is the sum function for $\binom{k}{2}$ variables defined in the definition of \cksum[]~. 
\end{definition}

\begin{definition}
Factored Zero Triangle, \czt~is \czkc[3]. 
\end{definition}

\subsubsection{Hypotheses for Factored Problems}

First we will define the hypotheses for our factored list problems. 

In many lemma, theorem and definition statements we will use a structure where we put $(\#)$ before several problem or hypothesis names. This structure means that the statement is true for all non counting versions, or for all counting versions.  For example, in the first line below the two implies statements are:\\
``The  \ckov[]~hypothesis (i.e.\ckovh[]) states that \ckov[]~requires $n^{k-o(1)}$ time.''
 \\
 and ``The  \#\ckov[]~hypothesis (i.e.\#\ckovh[]) states that  \#\ckov[]~requires $n^{k-o(1)}$ time.''.

\begin{definition}
The  (\#)\ckov[]~hypothesis (i.e.(\#)\ckovh[]) states that  (\#) \ckov[]~requires $n^{k-o(1)}$ time. 

The  (\#)\cksum[]~hypothesis (i.e.(\#)\cksumh[]) states that  (\#) \cksum[]~requires $n^{k-o(1)}$ time.

The  (\#)\ckxor[]~hypothesis (i.e.(\#)\ckxorh[]) states that  (\#)\ckxor[]~requires $n^{k-o(1)}$ time.

The  (\#)\ckfunc[]~hypothesis (i.e.(\#)\ckfunch[]) states that  (\#)\ckfunc[]~requires $n^{k-o(1)}$ time. 
\end{definition}

Now we will define the hypotheses for our factored clique problems. 

\begin{definition}
The  (\#)\czkc~hypothesis (i.e.(\#)\czkch) states that (\#) \czkc~requires $n^{k-o(1)}$ time. 

The  (\#)\cfkc~hypothesis (i.e.(\#)\cfkch) states that  (\#) \cfkc~requires $n^{k-o(1)}$ time. 
\end{definition}

%\xxx{TODO: do this for all the other problems}

\subsubsection{Average-Case for Factored Problems}

We will separate the average-case distribution of factored problems into the normal case and a more-general parameterized case. 

\begin{definition}[More General Average-Case]
Let $S_{b,\mu}$ be a distribution over sets of vectors from $\{0,1\}^b$. A set drawn from $S_{b,\mu}$ includes every vector $w \in \{0,1\}^b$ with probability $\mu$. 

Let $D_{g,b,\mu}$ be a distribution over factored vectors $v$ where all $g$ sets of $v[i]$ are sampled iid from $S_{b,\mu}$.

The average-case distribution for \#\ckfunc[]$^{\mu}$ samples every factored vector in its input iid from $D_{g,b,\mu}$.

The average-case distribution for \#\cfkc[]$^{\mu}$ samples every factored vector in its input iid from $D_{g,b,\mu}$.
\label{def:muavgCasefactored}
\end{definition}

For the average-case we use in this paper we use $\mu=1/2$. We feel this is the most natural distribution for our problem. We will occasionally call this the ``uniform average-case'' to emphasize that every set $v[i]$ in every factored vector is chosen uniformly at random from all possible subsets of $\{0,1\}^b$.

\begin{definition}
The average-case distribution for \#\ckfunc[]~samples every factored vector in its input iid from $D_{g,b,1/2}$.

The average-case distribution for \#\cfkc[]~samples every factored vector in its input iid from $D_{g,b,1/2}$.
\label{def:avgCaseFactored}
\end{definition}

\subsection{Problems harder than factored problems}
Here we define problems that later are shown to be hard via reductions from the factored problems. We state the known results for each, and a simple algorithm for each is given in Appendix \ref{app:algorithms} that matches the lower bound we prove later.

\begin{definition}
The \textbf{Partitioned Matching Triangles (\NDMT)} problem takes as input $g=O(\log{n}/\log\log{n})$ disjoint $n$-node graphs with node colors, and asks if there is a triple of colors with a triangle of that color triple in each of the $g$ graphs. The counting version of the problem, \textit{$\#$PMT}, asks for the number of such $g$-tuples of colored triangle.
%Given a 3 partite graph $G=(A,B,C)$ with colored nodes and an integer $g=O(\log{n}/\log\log{n})$, is there a triple of distinct colors $a,b,c$ such that there are at least $g$ node-disjoint triangles $(x,y,z)$ in $G$ in which $x$ has color $a$, $y$ has color $b$, and $z$ has color $c$? (Are there $g$ node-disjoint triangles with “matching” colors?)
\end{definition}

This problem is very similar to the \textit{$\Delta$ Matching Triangles} problem defined in \cite{matchingTriangles}, where given an $n$-node graph $G$ with node colors, the problems asks if there is a triple color with $\Delta$ triangles of that color triple in $G$.

In \cite{matchingTriangles}, $3$SUM, APSP and SETH are reduced to $\Delta$ Matching Triangles where the instances produced can be represented as instances of Partitioned Matching Triangles instance for $g=\Delta$. So Partitioned Matching Triangles is hard from $3$SUM, APSP and SETH. A related problem to PMT is the \textit{node disjoint triangle packing} problem which asks to find a maximum size node-disjoint triangle packing in a given graph (see for example \cite{caprara2002packing}). PMT is a natural mix of the $\Delta$-matching-triangle and the node disjoint triangle packing problems.
%Note that $g=O(\log{n}/\log\log{n})$ for our purposes.

\begin{definition}
\textbf{Node Labeled $k$-Color $st$ Connectivity (\kNLstC)} takes as input a directed graph $G$ with edge set $E$ and vertex set $V$, two special nodes $s$ and $t$, and a proper coloring of the vertices $c:~V\setminus \{s,t\}\rightarrow C$, where $C$ is a set of colors, so that the endpoints of every edge have different colors and $s$ and $t$ have all their neighbors colored distinctly.
The input graph $G$ is a layered graph, the vertex set $V$ is partitioned into $V_0, V_1,\ldots,V_\ell$, such that every directed edge goes from a node in set $V_i$ to a node in set $V_{i+1}$ for some $i\in \{0,\ldots,\ell-1\}$.
%Every vertex in $V$ is colored a color from $C$, except $s$ and $t$. 
%There are no edges from a vertex of a given color $c$ to another vertex of the same color in the input. 
%Furthermore, $s$ and $t$ will each connect to at most one node of every color. 
The \kNLstC~problem asks if there is a path from $s$ to $t$ that uses only $k$ colors of nodes (where $s$ and $t$ are not counted for colors).%, and it returns false otherwise.
\end{definition}

We also consider the problem of Counting \kNLstC~mod $R$, in which we ask for the total number of paths from $s$ to $t$ that use at most $k$ colors of nodes. We will generally use values of $R$ such that $\lg(R)$ is subpolynomial, as this allows us to represent the count with a subpolynomial number of bits.

\begin{definition}
The \textbf{$k$ Edge Labeled (directed/undirected) $st$ Connectivity (\kELstC[])} problem takes as input a directed acyclic graph $G=(V,E)$, two special vertices $s$ and $t$ and a coloring of the edges $c:E\rightarrow C$, where $C$ is a set of colors. 

% Each edge has an associated color (the set of colors has size $|C|$). Two special nodes are marked, s and t. 

\kELstC[]~asks, given this input can you pick $k$ colors such that there is a path from $s$ to $t$ using only edges that are colored by one of those $k$ colors? The counting version of \kELstC[], $\#$\kELstC[]~asks for the number of paths from $s$ to $t$ mod $R$ that use only $k$ colors, where $\lg(R) = n^{o(1)}$.
\end{definition}

\begin{definition}
The \textbf{Bounded Labeled Maximum Flow (BL-MF)} problem \cite{mfml} takes as input a directed, capacitated, and edge-labeled graph $G=(V,E)$ with a source node $s\in V$, a sink node $t\in V$, and a positive integer $k$, and asks if there is a maximum flow $x$ from $s$ to $t$ in $G$ such that the
total number of different labels corresponding to arcs $(i,j)\in E$ with non-zero flow is less than or equal to $k$. For fixed constant $k$, we refer to the problem as $k$L-MF.
\end{definition}

BL-MF is the decision version of the maximum flow with the \textit{minimum number of labels (MF-ML)} problem where we seek a maximum flow from $s$ to $t$ that uses the minimum number of labels. \cite{mfml} uses this problem to model the purification of water during the distribution process. They show that BL-MF is NP-complete. Let \textit{BL-MF*} be a slightly more restricted version of BL-MF where the number of edges of each label is $o(n)$ and the edges attached to the sink and source have a special label $l^*$. 
We show a lower bound of $O(n^{k-1})$ for $k$L-MF* (and thus $k$L-MF) for fixed $k$, and show that it has a matching algorithm as well.% We call BL-MF with $k$ being a fixed number $k$-labeled max flow or $kL-MF$ for short.

\begin{definition}
The \textbf{Regular Expression Matching} problem \cite{regularexpression} takes as input a regular expression (pattern) $p$ of size $m$ and a sequence of symbols (text) $t$ of length $n$, and asks if there is a substring of $t$ that can be derived from $p$.
The counting version of the problem, \#Regular Expression Matching asks for the number of subset alignments of the pattern of the pattern in the text mod an integer $R$, where $R=n^{o(1)}$.
\end{definition}

A classic algorithm constructs and simulates a non-deterministic finite automaton corresponding to the expression, resulting in the rectangular $O(mn)$ running time. 
%, and show that it has a $\t{O}(n^3)$ algorithm and a matching lower bound.

\begin{definition}
The (counting) \textbf{$k$-Longest Common Subsequence ((\#)\klcs)} problem (see for example \cite{LCSdef}) takes as input $k$  sequences $P_1, \ldots, P_k$ of
length $n$ over an alphabet $\Sigma$. Let $\ell$ be the length of the longest sequence $X$ such that $X$ appears in all of $P_1, \ldots, P_k$ (in the same order). 

\klcs~ asks for the value of $\ell$, while \#\klcs~asks to compute $\ell$ and also the total {\em number} of common subsequences of length $\ell$. 

More formally, define $C_{\#}(X_i)$ to be the total number $k$-tuples of $\ell$ sequence locations in each of our $k$ strings such that those locations map onto the sequence $X_i$ for all $k$ strings when $X_i$ is of length $\ell$. Let $X_1, X_2, \ldots X_j$ be all possible sequences of length $\ell$ that appear in all of $P_1, \ldots, P_k$. For the \#\klcs~problem we ask for the value of $\ell$ and the value of $\Pi_{i=1}^j C_{\#}(X_i)$.
\end{definition}

%\begin{definition}
%\textbf{Maximum Flow with the Minimum Number of Labels (MF-ML)} \cite{mfml} Given a directed edge capacitated and labeled graph $G=(V,E)$, what is the maximum flow on $G$, from a source $s$ to a sink $t$, using the minimum number of different labels.
%\end{definition}

%\cite{mfml} uses this problem to model the purification of water during the distribution process. They show that this problem is NP-complete. However, we add a small constraint to the decision variant of this problem 

\begin{definition}
The \textbf{Edit Distance} problem (see for example \cite{BackursI15}) takes as input two sequences $x$ and $y$ over an alphabet $\Sigma$, and asks to output the edit distance $EDIT(x, y)$ which is equal to the minimum number of symbol insertions, symbol deletions or symbol substitutions needed to transform $x$ into $y$.
\end{definition}

\section{Factored Problems are Hard}
\label{sec:compressedVariants}
In this section we will first show the simple result that \ckov, \cksum[], \ckxor[], and \czt~are all at least as hard as their non-factored variants. 
Second, we will show a worst-case to average-case reduction from \ckov[]~to itself. We will also show the corresponding worst-case to average-case reductions for \cksum, \ckxor[], and \czt.
Third, we will show many worst-case reductions between these factored problems. Notably, \ckov[], \cksum[], and \ckxor[]~are all equivalent up to sub-polynomial factors. Additionally, \czt~is $n^{3-o(1)}$ hard from \ckov[3]~(and thus equivalently hard from \cksum[3], and \ckxor[3]). Notably this means that the \cfkch[3]~is implied by SETH, the $3$-SUM hypothesis, and the APSP hypothesis. 
Figure \ref{fig:reductions1} summarized the reductions of this section. 

Remember that algorithms for these problems are given in Appendix \ref{app:algorithms}. We give $O(n^{k+o(1)})$ algorithms for \ckov[], \cksum[], \czkc[]~and \ckxor[].

\begin{figure}[h]
  \centering
  \includegraphics[width=.6\linewidth]{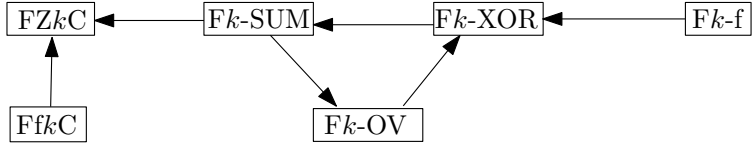}
  \caption{Map of the reductions}
  \label{fig:reductions1}
\end{figure}

%TODO: Factored zero k clique is hard from k-ov and regular k-clique. 

\subsection{Factored Versions are Harder}
\label{subsec:compressedVsUncompressed}

Consider any problem where we have $k$ sets of vectors $V_1,\ldots,V_k$ and we want to compute the number of $k$-tuples of vectors $v_1\in V_1,\ldots, v_k\in V_k$ of length $l=bg$, such that $\hat{f}(v_1,\ldots,v_k)=1$, for some function $\hat{f}:\{0,1\}^{l\times k}\rightarrow \{0,1\}$. Call this problem $k$-$\hat{f}$. Note that \ksum, $k$XOR, and $k$OV are examples of such problems. We show that these problems can be solved using their factored version. 

For any vector $v$ of length $bg$ and for any $j=1,\ldots,g$, let $v^j$ be the subvector of $v$ that starts at the $(j-1)b+1$th bit and ends at the $jb$th bit. Suppose that there is a function $f:\{0,1\}^{b\times k}\rightarrow \{0,1\}$ such that $\hat{f}(v_1,\ldots,v_k)=1$ if and only if $\Pi_{j=1}^g f(v_1^j,\ldots,v_k^j)=1$. We call $f$ the factored version of $\hat{f}$. In other words, the function $\hat{f}$ can be applied more locally, on subvectors of length $g$. Note that for most problems including the problems we work with, this property holds. 

Now we can easily reduce $k$-$\hat{f}$ to \ckfunc[]. Let the resulting \ckfunc[]~instance be the following: For any $bg$ length vector $v$ in the $k$-$\hat{f}$ instance, let the factored version of $v$ have sets $v[j]=\{v^j\}$. By the property mentioned, it is straightforward to see that this instance of \ckfunc[]~is equivalent to $k$-$\hat{f}$.

For the \ksum~problem it is less obvious how to solve it with \cksum[]. For the \ksum~problem we can use the nearly linear hash functions to reduce all numbers to the range $[-n^k,n^k]$ \cite{Patrascu10}. Additionally, we can reduce \ksum~in the range $[-n^k,n^k]$ to a version where every number is instead a vector with $g$ numbers with $b$ bits each \cite{losingWeight}, where $g \cdot b = k \lg(n)$. We consider a sum of $k$ vectors to be a zero sum if the vectors sum to the zero vector. To ask if $k$ numbers of length $k\lg(n)$ sum to zero, we can instead ask if $k$ vectors sum to the zero vector where the vectors have length $g$ and the numbers are each $b$ bits. But, we need to guess the $g-1$ carries, this is a total of $O(k^g)$ guesses. If $g = o(\lg(n))$ this is sub-polynomial, and so we can go through all these guesses. This vectorized version of the \ksum~problem can be directly solved by \cksum[]~as mentioned above.

Similar to the approach for solving $k$-$\hat{f}$ problems using \ckfunc[], we can reduce Z$k$C to \czkc. Here the function $\hat{f}$ (which is the sum function) gets $\binom{k}{2}$ vectors as input instead of $k$ vectors, and these vectors should have the property that they form the edges of a $k$-clique in the graph. If $f$ is the factored version of $\hat{f}$, then we have that $\hat{f}(e_1,\ldots,e_{\binom{k}{2}})=1$ and $e_1,\ldots ,e_{\binom{k}{2}}$ are the edges of a $k$-clique if and only if $\Pi_{j=1}^g f(e_1^j,\ldots,e_{\binom{k}{2}}^j)=1$ and $e_1,\ldots e_{\binom{k}{2}}$ are edges of a $k$-clique. So again, the \czkc~instance that is equivalent to the Z$k$C instance is that for each vector $e$ of an edge, we let $e[j]=\{e^j\}$. 
So we have the following theorem. Note that our reductions form a one-to-one correspondence between each solution in a $k$-$\hat{f}$ instance and the corresponding \ckfunc[]~instance, and hence they work for the counging version of our problems as well.

\begin{reminder}{Theorem \ref{thm:unfactoredToFactored}}
In $O(n)$ time, one can reduce an instance of size $n$ of $k$-OV, $k$-XOR, $k$-SUM and Z$k$C to a single call to an instance of size $\tO(n)$  of \ckov[], \ckxor[], \cksum[]~and \czkc~respectively.
\end{reminder}
\begin{proof}
We can split a number or vector in the original problem into a vector with $g$ numbers of length $b$. This reduction step is trivial for \kov, \zkc, and \kxor. To be explicit:
\begin{itemize}
    \item \kov: Let $d = o(\lg^2(n))$ be the dimension. Given $k$ lists of $n$ vectors $L_1,\ldots,L_k$ we will produce $k$ lists of factored vectors $L_1',\ldots,L'_k$. Let $v[x:y]$ be a vector formed by taking all the bits from $x^{th}$ bit to the $y^{th}$ bit. For every vector $v_i \in L_i$ take the vector $v_i[bj +1, b(j+1)] = v_i^j$ where $b = \sqrt{d}$ and $j \in [1,g]$ where $g = \sqrt{d}$. We create a factored vector $v_i'$ from $v_i$ by creating a vector where the $j^{th}$ subset of $\{0,1\}^b$ is just a set with the single vector $v_i^j$.
    \item \kxor: There is a random reduction for \kxor~which shrinks the vectors to length $d=k\lg(n)$ bits. Given two lists of $n$ vectors $L_1$ and $L_2$ we will produce two lists of factored vectors $L_1'$ and $L_2'$. Let $v[x:y]$ be a vector formed by taking all the bits from $x^{th}$ bit to the $y^{th}$ bit. For every vector $v_i \in L_i$ take the vector $v_i[bj +1, b(j+1)] = v_i^j$ where $b = \sqrt{d}$ and $j \in [1,g]$ where $g = \sqrt{d}$. We create a factored vector $v_i'$ from $v_i$ by creating a vector where the $j^{th}$ subset of $\{0,1\}^b$ is just a set with the single vector $v_i^j$.
    \item \zkc~and \ksum: We reduce the range of numbers with linear hash functions to the range $[-n^k,n^k]$. We want to split each $k\lg(n)$ bit number into $g=\sqrt{k\lg(n)}$ numbers length $b = \sqrt{k\lg(n)}$ bits. If we guess all $g$ caries then we can replace the question of if $k$ (${\binom{k}{2}}$) numbers sum to zero to if $g$ sets of $k$ (${\binom{k}{2}}$) numbers each sum to zero (see \cite{losingWeight}). So, for all $O(k^{g})$ possible guesses of carries we form a factored vector for an edge by having $g$ subsets of $\{0,1\}^b$ that each have one number. The $j^{th}$ set has the $j^{th}$ number created by splitting the original number (possibly updated by our guess of the carry). 
    %\item [\ksum:] As above we reduce the range of the numbers to $[-n^k,n^k]$, and we want to split each $k\lg(n)$ bit number into $g=\sqrt{k\lg(n)}$ numbers length $b = \sqrt{k\lg(n)}$ bits. If we guess all $g$ caries then we can replace the question of if $k$ numbers sum to zero to if $g$ sets of $k$ numbers each sum to zero (see \cite{losingWeight}).   So, for all $k^{g}$ possible guesses of carries we form a factored vector for a number from the original list by having $g$ subsets of $\{0,1\}^b$ that each have one number. The $j^{th}$ set has the $j^{th}$ number created by splitting the original number (possibly updated by our guess of the carry). 
\end{itemize}
\end{proof}

%for function f, let $\hat{f}$ be extension of f on $b.g$ bit numbers For xor and ov it is obvious, for sum not so obvious. for zero clique it is similar to sum. 

\subsection{Worst-Case to Average-Case Reductions for Factored Problems}
\label{subsec:WCtoACforCompressed}

We will use our framework from Section \ref{sec:Framework}
 to show that these factored versions are as hard on average as they are in the worst case. 
 
  \paragraph{\#\ckfunc[]}
 
We give a polynomial for \#\ckfunc[]. We represent every factored vector $v$ with $g2^b$ variables. The variable $x_{v[i]}(\vec{s})$ is a $1$ if $s \in v[i]$ and $0$ otherwise. We create such a variable for all $i \in [0,g-1]$ and all $s \in \{0,1\}^b$.
 Let $S_f$ be the subset of $k$ tuples of vectors in  $\{0,1\}^b$ such that $f(s_1,\ldots,s_k)=1$.
 
 $$f_{ckfunc}(\vec{X}) = \sum_{v_1 \in V_1, \ldots, v_k \in V_k} \left( \prod_{i \in [0,g-1]} \left( \sum_{(s_1,\ldots, s_k)\in S_f}  x_{v_1[i]}(\vec{s_1}) \cdots x_{v_k[i]}(\vec{s_k}) \right) \right).$$
 
 \begin{lemma}
 $f_{ckfunc}(\vec{X})$ is a \gPol{\#\ckfunc[]}
 (see Definition \ref{def:goodPoly})%(see Appendix \ref{app:prelims2} and Appendix \ref{sec:Framework} for an explanation of \gPol{f}). 
 \label{lem:funcIsGoodPol}
 \end{lemma}
 \begin{proof}
 We will show that each property of a good polynomial is met by $f_{ckfunc}$. 
 \begin{itemize}
   % \item \textit{There are at most polynomial monomials in $f_{ckfunc}$:} There are at most $n^k \cdot g \cdot 2^{bk} \cdot = \tO(n^k)$ monomials.
    
    \item \textit{If $\vec{I} = b_1,\ldots, b_n$, then 
    $f_{ckfunc}(b_1,\ldots, b_n) = f_{ckfunc}(\vec{I})=P(\vec{I})$ where $b_i$ maps to either a zero or a one in the prime finite field}: $f_{ckfunc}$ and \#\ckfunc[]~count the same thing. Note that the inner summation is computing $\circ_f$, the product is computing $\circledcirc_f$. Thus the overall sum is computing \#\ckfunc[].
    \item \textit{The function $f_{ckfunc}$ has degree $d = o\left( \lg(n)/\lg\lg(n) \right)$:} $f_{ckfunc}$ has degree $kg$ which, when $k$ is constant is $o(\lg(n)/\lg\lg(n))$ by the definition of $g$.
    \item \textit{The function $f_{ckfunc}$ is strongly $d$-partite:} Every monomial is formed by exactly one copy of a $x_{v_j[i]}(\vec{s})$ variable for every $j \in [0,g-1]$ and $i\in[1,k]$. These form our partitions and make the function strongly $kg$ partite. 
\end{itemize}
\end{proof}
 
 Now we can say that the average case version of \#\ckfunc[]~is as hard as the worst case version. 
 
\begin{reminder}{Theorem \ref{thm:worstToAvgckf}}
%We will use the definition of average case  \#\ckfunc[]$^{\mu}$ from Definition \ref{def:muavgCasefactored}. 
Let $\mu$ be a constant such that $0 < \mu <1$. 
If average-case \#\ckfunc[]$^{\mu}$~(see Definition \ref{def:muavgCasefactored}) can be solved in time $T(n)$ with probability $1-1/(lg(n)^{kg}\lg\lg(n)^{kg})$  then worst-case \#\ckfunc[]~can be solved in time $\tO(T(n))$ \footnote{Note that given that $g= o(\lg(n)/\lg\lg(n)) $ then a probability of $1-1/n^{\epsilon}$ will be high enough for any $\epsilon>0$.}.

When $\mu =1/2$ average-case \#\ckfunc[]$^{\mu}$ is average-case \#\ckfunc[]. 
\end{reminder}
\begin{proof}
This follows from Theorem \ref{thm:framework} and Lemma \ref{lem:funcIsGoodPol}. The dimension of the \gPol{\#\ckfunc[]} is $kg$. By our construction of $f_{ckfunc}$ every set has every possible string as a variable. By the construction of the framework from theorem \ref{thm:framework}, every bit will be selected as a $1$ uniformly at random with probability $\mu$. So, given the construction of $f_{ckfunc}$ every set will have every possible string included with probability $\mu$. So the distribution induced by our framework matches our defined average-case distribution. 

Finally by definition \ref{def:avgCaseFactored} when $\mu =1/2$ average-case \#\ckfunc[]$^{\mu}$ is average-case \#\ckfunc[]. 
%\label{thm:ckfuncThm}
\end{proof}
 
 %%%%%%%%%%%%%%%%
 
\begin{reminder}{Corollary \ref{cor:worstToAvg}}
By Theorem \ref{thm:worstToAvgckf}, we have the following result: 

If average-case \#\ckov[]~can be solved in time $T(n)$ with probability $1-1/(lg(n)^{gk}\lg\lg(n)^{gk})$  then worst-case \#\ckov[]~ can be solved in time $\tO(T(n))$ \footnotemark[\value{footnote}].

If average-case \# \cksum[]~can be solved in time $T(n)$ with probability $1-1/(lg(n)^{gk}\lg\lg(n)^{gk})$  then worst-case \# \cksum[]~can be solved in time $\tO(T(n))$ \footnotemark[\value{footnote}].

If average-case \# \ckxor[]~can be solved in time $T(n)$ with probability $1-1/(lg(n)^{gk}\lg\lg(n)^{gk})$  then worst-case \# \ckxor[]~can be solved in time $\tO(T(n))$ \footnotemark[\value{footnote}].

\end{reminder}

%%%%%%%
  \paragraph{\#\cfkc[]}
  
Now we will give the \#\cfkc[]~polynomial. Once again we will represent every factored vector $v$ with $g2^b$ variables. The variable $x_{v[i]}(\vec{s})$ is a $1$ if $s \in v[i]$ and $0$ otherwise. We create such a variable for all $i \in [0,g-1]$ and all $s \in \{0,1\}^b$.
Once again, let $S_f$ be the subset of $\binom{k}{2}$ tuples of vectors in  $\{0,1\}^b$ such that $f(s_1,\ldots,s_{\binom{k}{2}})=1$.
Finally for convenience let $E_1,\ldots, E_{\binom{k}{2}}$ be the $\binom{k}{2}$ partitions of edges in the input of \cfkc[]~and let $\ell = \binom{k}{2}$ to make notation easier to read. Let $S_{E}$ be the set of all $\ell$ tuples of edges $e_1,\ldots, e_\ell$ that form a clique. In an abuse of notation we will also use $e_i$ to represent the factored vector associated with the edge $e_i$.
 
 $$f_{ffkc}(\vec{X}) = \sum_{e_1, \ldots, e_\ell \in S_E} \left( \prod_{i \in [0,g-1]} \left( \sum_{(s_1,\ldots, s_\ell)\in S_f}  x_{e_1[i]}(\vec{s_1}) \cdots x_{e_\ell[i]}(\vec{s_\ell}) \right) \right).$$
  
\begin{lemma}
 $f_{ffkc}(\vec{X})$ is a \gPol{\#\cfkc[]} (see Definition \ref{def:goodPoly}). 
 \label{lem:funcIsGoodPolclique}
 \end{lemma}
 \begin{proof}
 We will show that each property of a good polynomial is met by $f_{ffkc}$. 
 \begin{itemize}
   % \item \textit{There are at most polynomial monomials in $f_{ffkc}$:} There are at most $n^k \cdot g \cdot 2^{b{\ell}} = \tO(n^{k^2})$ monomials.
    
    \item \textit{If $\vec{I} = b_1,\ldots, b_n$, then 
    $f_{ffk}(b_1,\ldots, b_n) = f_{ffkc}(\vec{I})=P(\vec{I})$ where $b_i$ maps to either a zero or a one in the prime finite field}: $f_{ffkc}$ and \#\ckfunc[]~count the same thing. Note that the inner summation is computing $\circ_f$, the product is computing $\circledcirc'_f$. Thus the overall sum is computing \#\ckfunc[].
    \item \textit{The function $f_{ffkc}$ has degree $d = o\left( \lg(n)/\lg\lg(n) \right)$:} $f_{ffkc}$ has degree $\ell g< k^2g$ which, when $k$ is constant is $o(\lg(n)/\lg\lg(n))$ by the definition of $g$.
    \item \textit{The function $f_{ffkc}$ is strongly $d$-partite:} Every monomial is formed by exactly one copy of a $x_{e_j[i]}(\vec{s})$ variable for every $j \in [0,g-1]$ and $i\in[1,\ell]$. These form our partitions and make the function strongly $\ell g$ partite. 
\end{itemize}

\end{proof}

\begin{reminder}{Theorem \ref{thm:worstToAvgclique}}
%We will use the definition of average-case \#\cfkc[]$^\mu$ from Definition \ref{def:muavgCasefactored}. 
Let $\mu$ be a constant and $0<\mu<1$.
If average-case \#\cfkc[]$^{\mu}$~(see Definition \ref{def:muavgCasefactored}) can be solved in time $T(n)$ with probability $1-1/(lg(n)^{k^2g}\lg\lg(n)^{k^2g})$  then worst-case \#\cfkc[]~can be solved in time $\tO(T(n))$ \footnotemark[\value{footnote}].

When $\mu=1/2$ average-case \#\cfkc[]$^{\mu}$ is average-case \#\cfkc[].
\end{reminder}
\begin{proof}
This follows from Theorem \ref{thm:framework} and Lemma \ref{lem:funcIsGoodPolclique}. The dimension of the \gPol{\#\cfkc[]} is $\binom{k}{2}g<k^2g$. By our construction of $f_{ffkc}$ every set has every possible string as a variable. By the construction of the framework from theorem \ref{thm:framework}, every bit will be selected as a $1$ uniformly at random with probability $\mu$. So, given the construction of $f_{ffkc}$ every set will have every possible string included with probability $\mu$. So the distribution induced by our framework matches our defined average-case distribution. 

Finally by definition \ref{def:avgCaseFactored} when $\mu =1/2$ average-case \#\cfkc[]$^{mu}$ is average-case \#\cfkc[]. 
\end{proof}
 
 %%%%%%%%%%%%%%%%
 
\begin{reminder}{Corollary \ref{cor:worstToAvgClique}}
By Theorem \ref{thm:worstToAvgclique}, we have the following result: 

If average-case \# \czkc[]~can be solved in time $T(n)$ with probability $1-1/(lg(n)^{k^2g}\lg\lg(n)^{k^2g})$ then worst-case \#\czkc[]~can be solved in time $\tO(T(n))$ \footnotemark[\value{footnote}] .
\end{reminder} 

% \paragraph{\ckov[]}

%\begin{theorem}
%Let average-case \ckov[]~instances be drawn from the distribution where every set $v[i]$ is picked independently and uniformly at random from all subsets of $\{0,1\}^b$.

%Average-case \ckov[]~can be solved in time $T(n)$ with probability  then worst-case \ckov[]~can be solved in time $\tO(T(n))$.
%\end{theorem}
%\begin{proof}
%This follows from Theorem %\ref{thm:ckfuncThm}.
%\end{proof}

%\paragraph{\cksum[]} 
%\begin{theorem}
%Let average-case \cksum[]~instances be drawn from the distribution where every set $v[i]$ is picked independently and uniformly at random from all subsets of $\{0,1\}^b$.

%Average-case \cksum[]~can be solved in time $T(n)$ with probability  then worst-case \cksum[]~can be solved in time $\tO(T(n))$.
%\end{theorem}
%\begin{proof}
%This follows from Theorem %\ref{thm:ckfuncThm}.
%\end{proof}

%\paragraph{\ckxor[]} 
%\begin{theorem}
%Let average-case \ckxor[]~instances be drawn from the distribution where every set $v[i]$ is picked independently and uniformly at random from all subsets of $\{0,1\}^b$.

%Average-case \ckxor[]~can be solved in time $T(n)$ with probability  then worst-case \ckxor[]~can be solved in time $\tO(T(n))$.
%\end{theorem}
%\begin{proof}
%This follows from Theorem 
%\ref{thm:ckfuncThm}.
%\end{proof}

\paragraph{Reductions to Counting Factored Problems Imply Average Case Hardness Over Some Distribution}

Assume a problem $\# P$ exists such that an algorithm for it running in $T(n)^{1-\epsilon}$ implies a violation of \#\ckfunch[]~or \#\cfkch[]. We will evince examples of such problems in Section \ref{sec:harderProblem}. Further imagine that there is an explicit reduction that turns instances of \#\ckfunc[]~or \#\cfkc[]~into instances of $\# P$. In that case we can describe a distribution $D$ over which problem $\# P$ is $T(n)^{1-o(1)}$ hard on average from \#\ckfunch[]~or \#\cfkch[]. We can generate this distribution $D$ by taking the uniform distribution (the average-case distribution) over \#\ckfunc[]~or \#\cfkc[]~and running this distribution through our reduction. 

Thus, reductions from problems $\# P$ to \#\ckfunc[]~or \#\cfkc[]~give explicit hard average-case distributions for problems $\# P$.

\subsection{Factoring is Expressive: Worst-Case Reductions}
\label{subsec:WCReductionsCompressed}

Our factored versions of these problems are very expressive. This allows us to show hardness from these factored problems.

\subsubsection{Completeness}

We will now show that \ckov[], \ckxor[], and \cksum[]~are all complete for \ckfunc[]~for all functions $\mathfrak{f}$. We do this by showing \ckxor[]~solves \ckfunc[]. Then, the equivalence between \ckov[], \ckxor[], and \cksum[]~implies they are all complete for \ckfunc[].
We will also show, using similar techniques, that \czkc[]~is complete for \cfkc[]~for all functions $\mathfrak{f}$. 

This is a reminder of Theorem \ref{thm:ckxorIscomplete}, however, we add an additional statement to the theorem. We give an explicit function that we use to build this reduction. 

\begin{reminder}{Theorem \ref{thm:ckxorIscomplete}}
If we can solve  \#\ckxor[]~with $g$ sets of $k^3b$ length vectors in time $T(n)$ then we can solve \#\ckfunc[]~instance with $g$ sets of $b$ length vectors in time $T(n)+\tilde{O}(n)$. 

Additionally, 
Let $v_1, \ldots, v_k$ be $k$ factored vectors each with $g$ subsets of $\{0,1\}^b$. Let $f_{XOR}$ be the function that returns $1$ if the $k$ input vectors xor to zero and otherwise returns $0$. 
There is a function $\gamma_{\mathfrak{f} \rightarrow XOR,k}(\cdot, \cdot)$ that takes as input a factored vector with $g$ subsets of $\{0,1\}^b$ and an index and returns a new factored vector  with $g$ subsets of $\{0,1\}^{O(b)}$. This function $\gamma_{\mathfrak{f} \rightarrow XOR,k}$ 
runs in $\tO(2^b \cdot g)$ time for each vector and:
$$\circledcirc_{f_{XOR}}(\gamma_{\mathfrak{f} \rightarrow XOR,k}(v_1, 1), \ldots, \gamma_{\mathfrak{f} \rightarrow XOR,k}(v_k,k)) = \circledcirc_{\mathfrak{f}}(v_1, \ldots, v_k).$$
\end{reminder}
\begin{proof}
Consider a \ckfunc[]~instance and let $v_i$ be a factored vector from the $i^{th}$ list of it.
Given the factored vector $v_i$ from \ckfunc[]~we will describe how to make the factored vector $v_i'$ for our \ckxor[]~instance. 
This transformation will be $\gamma_{\mathfrak{f} \rightarrow XOR,k}(\cdot, \cdot)$. We will describe the transformation for $\gamma_{\mathfrak{f} \rightarrow XOR,k}(v_i, i)$.
We do this by doing the same transformation on each set $v_i[j]$ where $j\in[1,g]$. We transform each set by performing the same transformation on every vector $u_i \in v_i[j]$. We describe this transformation in the next paragraph. 

Given a vector $u_i$ of length $b$ we produce at most $2^{b (k-1)}$ new vectors of length $k^3b$. These vectors represent all possible $k$ tuples which include $u_i$ as as the $i^{th}$ vector. We want to include a $k$ tuple vector only if $f$ of that $k$ tuple evaluates to $1$. And we want our new long vectors to to return true if we are comparing vectors in \ckfunc[]~instance that do indeed have exactly that $k$ tuple of vectors. 

More formally, let one possible $k$ tuple that includes $u_i$ as the $i^{th}$ vector be $(w_1, \ldots, w_{i-1},u_i,w_{i+1} \ldots , w_k)$.
If 
$f(w_1, \ldots, w_{i-1},u_i,w_{i+1} \ldots , w_k)=1,$
then we create a $k^3b$-length vector for $u_i$ with this $kb$ length vector by considering every possible tuple $(x,y,z)$ where $x,y,z\in[1,k]$: We set aside $b$ bits for every possible tuple (in sorted order by the tuple). We want to use these to check if the $x^{th}$ vector and $y^{th}$ vector agree about what tuple they are considering as follows:
\begin{itemize}
\item If the tuple is $(x,x,z)$ we write the all zeros string, for the rest of the cases assume the first two indices are non-equal.
\item If the tuple is $(x,i,z)$ we write $w_z$ (or $u_i$ if $z=i$) in the $b$ bits.
\item If the tuple is $(i,y,z)$ we write $w_z$ (or $u_i$ if $z=i$) in the $b$ bits.
\item If the tuple is $(x,y,z)$ and $x,y \ne i$ then we write the all zeros vector of length $b$.
\end{itemize}

If we are comparing $k$ of these new vectors each of which representing the same tuple $(w_1, \ldots , w_k)$ then the new vector xors to zero. Consider a given group of $b$ bits that corresponds to $(x,y,z)$. Only two of our vectors have non-zero entries here, the $x^{th}$ and $y^{th}$ vectors. Both wrote down $w_z$ if they were representing the same $k$-tuple. A vector xored to itself produces the zero vector, so we get the zero vector. 

If we are comparing $k$ of these new vectors and not all of the vectors agree about what tuples they are comparing then we will not xor to the zero vector. Say the $x^{th}$ and $y^{th}$ vectors disagree about what the $z^{th}$ element of the tuple is. Then the $b$ bits corresponding to $(x,y,z)$ will still have only two vectors with non-zero contributions. We will be xoring two vectors which are not equal, this will xor to some non-zero string. 

Thus there is a one-to-one correspondence between $k$-tuples of vectors that evaluate to one in the \ckfunc[]~version and $k$-tuples of vectors that xor to the zero vector in the new \ckxor[]~version. Thus, the counts both give as output are equal. 

As a result, we can transform an instance of \ckfunc[]~with $g$ groups of $b$ length vectors into an instance of \ckxor[]~with $g$ groups of $k^3 b$ length vectors in time $O(n\cdot g \cdot 2^b \cdot 2^{(k-1)b} \cdot k^3b )$. We restrict $k$ to be constant and $b= o(\lg(n))$, thus the time for the conversion is $\tilde{O}(n)$. 
In the new version the count of the number of \ckxor[]~vectors that xor to zero is the same as the count of the number of \ckfunc[]~vectors that evaluate to $1$ on the function. So, a $T(n)$ algorithm for \ckxor[]~with $g$ groups and $k^3b$ bits implies a $T(n) +\tilde{O}(n)$ algorithm for \ckfunc[].
\end{proof}

\subsubsection{\texorpdfstring{f$k$-OV}{Lg}, \texorpdfstring{f$k$-SUM}{Lg}, and \texorpdfstring{\ckxor[]}{Lg} are Equivalent and Complete}
Intuitively, we can use our subsets of $\{0,1\}^b$ to do guesses that reduce from one problem to another. 

\begin{lemma}
If  (\#)\ckov[]~can be solved in time $T(n)$ then  (\#)\ckxor[]~can be solved in $\tO(T(n))$ time.
\label{lem:ckov>cksum}
\end{lemma}
\begin{proof}
Say we are given an \ckxor[]~instance, with $k$ lists of factored vectors, each with $g$ subsets of $b$-bit vectors. We will follow the structure of Theorem \ref{thm:ckxorIscomplete}. Say we are given a factored vector from list $i$, $v_i$. Consider the $j^{th}$ subset $v_i[j]$ of it. Consider a particular vector $u_i \in v_i[j]$. 
We will produce a new vector for every possible $k$ tuple of vectors $(w_1, \ldots, u_i, \ldots, w_k)$ such that $w_1 \oplus \ldots \oplus u_i\oplus \ldots \oplus w_k = \vec{0}$. This vector will have $k^3$ sections each of length $2b$, for a total length of vector $2k^3b$. These $k^3$ sections will correspond to every possible tuple $(x,y,z)$ where $x,y,z \in [1,k]$. The $2b$ bits will be used to check if the vector from list $x$ and the vector from list $y$ agree about the vector $w_z$. We want to only accept if there are $k$ vectors, one from each list that xor to the zero vector. Let $\bar{s}$ be the bitwise bit flip of every bit in $s$. We will use the fact that if both $s_1$ and $\bar{s_2}$ are orthogonal and $\bar{s_1}$ and $s_2$ are orthogonal then $s_1 = s_2$. This allows us to check equality.  Let $\mathbin\Vert$ be the concatenation operator (e.g $00\mathbin\Vert01 = 0001$). The $2b$ bits that correspond to $(x,y,z)$ are determined as follows:
\begin{itemize}
    \item If the tuple is $(x,x,z)$ we write the $2b$ bit all zeros string. For the rest of these assume the first two indices are not equal.
\item If the tuple is $(i,x,z)$ then write $w_z\mathbin\Vert\bar{w_z}$ (for convenience let $w_i = u_i$).
\item If the tuple is $(x,i,z)$ then write $\bar{w_z}\mathbin\Vert w_z$ (for convenience let $w_i = u_i$).
\item If the tuple is $(x,y,z)$ and $x \ne i$ and $y \ne i$ then we put the all ones string. 
\end{itemize}

Now $k$ of these constructed vectors will be orthogonal only if $w_1 \oplus \ldots \oplus w_k = \vec{0}$, all the vectors $w_i$ existed in the original lists, and the constructed vectors all agree on the tuple $(w_1, \ldots, w_k)$.

So, with our constructed vectors the count of the number of vectors that are orthogonal will remain the same. The new instance will have the same number of factored vectors, $n$, but the vectors will have $g$ subsets of $\{0,1\}^{2k^3b}$. An algorithm which runs in $T(n)$ on this \ckov~instance will run in $T(n)$ on the \ckxor~instance. 
\end{proof}

Next we reduce (\#)\ckxor~to (\#)\cksum. In \ckxor~we want to know if $k$ vectors xor to zero, which is very similar to asking if $k$ numbers sum to zero. The difference is entirely carries. So, we can pad the instance, and then guess carries.

We will use the \cksum[]~variant where we ask if $k-1$ numbers sum to equal exactly the last number. 

\begin{lemma}
If (\#)\cksum[]~can be solved in time $T(n)$ then (\#)\ckxor[]~can be solved in $\tO(T(n))$ time.

Additionally, 
Let $v_1, \ldots, v_k$ be $k$ factored vectors each with $g$ subsets of $\{0,1\}^b$. Let $f_{XOR}$ be the function that returns $1$ if the $k$ input vectors xor to zero and otherwise returns $0$. Let $f_{SUM}$ be the function that returns $1$ if the $k$ input vectors sum to zero and otherwise returns $0$.
There is a function $\gamma_{XOR \rightarrow SUM,k}(\cdot, \cdot)$ that takes as input a factored vector with $g$ subsets of $\{0,1\}^b$ and an index and returns a new factored vector  with $g$ subsets of $\{0,1\}^{O(b)}$. This function $\gamma_{XOR \rightarrow SUM,k}$ 
runs in $\tO(2^b \cdot g)$ time for each vector and:
$$\circledcirc_{f_{SUM}}(\gamma_{XOR \rightarrow SUM,k}(v_1, 1), \ldots, \gamma_{XOR \rightarrow SUM,k}(v_k,k)) = 
\circledcirc_{\mathfrak{f}}(v_1, \ldots, v_k).$$
\label{lem:cksum>ckxor}
\end{lemma}
\begin{proof}
We will describe the transformation $\gamma_{XOR \rightarrow SUM,k}$ below. 
Let $v_1,\ldots, v_k$ be $k$ factored vectors from a \ckxor[]~instance. Let $v_1[j], \ldots, v_k[j]$ be the $j^{th}$ subset of $b$-bit vectors from each of the factored vectors. Let $v_i[j][h][\ell]$ be the $\ell^{th}$ bit of the $h^{th}$ vector in the set $v_i[j]$.

We will turn every bit from $v_i[j][h][\ell]$ into $\lceil \lg(k) \rceil +1$ bits in a new number. If $i<k$ then this new longer string is  $\lceil \lg(k) \rceil$ zeros followed by the bit $v_i[j][h][\ell]$. If $i=k$, then every vector $v_k[j][h]$ turns into many vectors in the $k^{th}$ set of the \cksum instance: If $v_k[j][h][\ell]=0$ then we use our $\lceil \lg(k) \rceil +1$ bits to represent all of the \textbf{even} numbers in $[0,k-1]$. If $v_k[j][h][\ell]=1$ then we use our $\lceil \lg(k) \rceil +1$ bits to represent all of the \textbf{odd} numbers in $[0,k-1]$. So we produce $O(k^b)$ vectors for $v_k[j][h]$.

If a k-tuple of vectors forms a zero vector in k-xor then we get exactly one k-sum. The number of sets stays the same but the length of vectors in those sets grows from $b$ to $(\lceil \lg(k) \rceil +1)b$ length vectors. This is a constant and so if \cksum[]~can be solved in time $T(n)$ then \ckxor[]~can be solved in $\tO(T(n))$ time.
\end{proof}

\begin{reminder}
{Theorem \ref{thm:compressedKVarEquiv}}
If any of \#\ckov[], \#\cksum[], or \#\ckxor[]~can be solved in $T(n)$ time then all of \#\ckov[], \#\cksum[], and \#\ckxor[]~can be solved in $\tO(T(n))$ time.
\end{reminder}
\begin{proof}
This follows from Lemmas \ref{lem:ckov>cksum}, \ref{lem:cksum>ckxor}, and Theorem \ref{thm:ckxorIscomplete}.
\end{proof}

\begin{reminder}{Theorem \ref{thm:allCompressedAreComplete}}
If any of \#\ckov[], \#\cksum[], or \#\ckxor[]~can be solved in $T(n)$ time then \#\ckfunc[]~can be solved in $\tO(T(n))$ time.
%\label{thm:allCompressedAreComplete}
\end{reminder}
\begin{proof}
Use Theorem \ref{thm:ckxorIscomplete} and Theorem \ref{thm:compressedKVarEquiv}.
\end{proof}

By Theorem \ref{thm:compressedKVarEquiv} and Theorem 
\ref{thm:allCompressedAreComplete} we get the following corollary.

\begin{corollary}
\#\ckovh[], \#\cksumh[], and \#\ckxorh[]~are all equivalent. Moreover, \#\ckov[]~is implied by \#\ckfunc[]~for any function $\mathfrak{f}$.
\label{cor:equivAndComplete}
\end{corollary}
%\begin{proof}
%The first statement is implied by Theorem \ref{thm:compressedKVarEquiv}, and the second statement is implied by Theorem 
%\ref{thm:allCompressedAreComplete}.
%\end{proof}

\subsubsection{Factored \texorpdfstring{zero-$k$-clique}{Lg} is hard from \texorpdfstring{f$k$-OV}{Lg}, \texorpdfstring{f$k$-SUM}{Lg}, \texorpdfstring{f$k$-XOR}{Lg}, and \texorpdfstring{\cfkc[]}{Lg}}

\begin{lemma}
If (\#)\czkc~is solved in time $T(n)$ then (\#)\cksum~is solved in time $\tO(T(n))$.
\label{lem:czt>cksum}
\end{lemma}
\begin{proof}
Consider a \cksum instance with lists $L_1,\ldots,L_k$. We will build the $k$-partite graph of our \czkc~instance to have vertex sets $V_1,\ldots,V_k$. 
For every factored number $x_i \in L_i$ from the \cksum instance we create a node $v_i \in V_i$ where all edges going from $v_i$ to any vertex in $V_{i+1}$ have the value $x_i$ on them. All edges going from $v_i$ to nodes in $V_j$ where $j \ne i-1$ and $j\ne i+1$ are given the special factored vector where every set contains only the all zeros string. 

Now, when three nodes are selected $v_1, v_2, \ldots, v_k$ the corresponding edges have a zero sum iff the corresponding $x_1,x_2,\ldots, x_k$ are a zero sum. 

\end{proof}

\begin{theorem}
If (\#)\czkc~can be solved in $T(n)$ time then all of (\#)\ckov, (\#)\cksum, and (\#)\ckxor~can be solved in $\tO(T(n))$ time.
\label{thm:cztIsAUnit}
\end{theorem}
\begin{proof}
This follows from Lemma \ref{lem:czt>cksum} and Theorem \ref{thm:compressedKVarEquiv}.
\end{proof}

Now we will show that \czkc[]~is complete for the set of all problems \cfkc[]~for all functions $\mathfrak{f}$.

\begin{reminder}{Theorem \ref{thm:fzkcIsComplete}}
If (\#)\czkc~can be solved in $T(n)$ time then  (\#)\cfkc[]~for any $f$, can be solved in $\tO(T(n)+n^2)$ time.
\end{reminder}
\begin{proof}
We will use the transforms $\gamma_{\mathfrak{f} \rightarrow XOR,\binom{k}{2}}(\cdot, \cdot)$ and $\gamma_{XOR \rightarrow SUM,\binom{k}{2}}(\cdot, \cdot)$ from Theorem \ref{thm:ckxorIscomplete} and Lemma \ref{lem:cksum>ckxor} respectively.

Let $G$ be the $k$-partite graph we take as input from (\#)\cfkc[], now label the $\ell = \binom{k}{2}$ edge sets as $E_1,\ldots, E_\ell$. Now, for every factored vector $e_i \in E_i$ run the following transform:
$\hat{\gamma}_{\ell}(e_i,i) = \gamma_{XOR \rightarrow SUM,\binom{k}{2}}(\gamma_{\mathfrak{f} \rightarrow XOR,\ell}(e_i, i), i).$ This causes the output factored vector to have $g$ subsets of $\{0,1\}^{O(b)}$. The transformation takes $\tO(2^{O(b)}g)$ time per vector, which is $\tO(1)$ time per vector. The output vectors have the property that $$\circledcirc_{f_{SUM}}(\hat{\gamma}_{\ell}(e_1,1),\ldots, \hat{\gamma}_{\ell}(e_\ell,\ell)) = \circledcirc_{\mathfrak{f}}(e_1,\ldots, e_\ell).$$

Because $\circledcirc'_{f}$ (the function used in our factored clique problem definition) is equal to $\circledcirc_{f} \cdot isClique(e_1,\ldots,e_{\ell})$, by running this transformation we will have that:
$$\circledcirc'_{f_{SUM}}(\hat{\gamma}_{\ell}(e_1,1),\ldots, \hat{\gamma}_{\ell}(e_\ell,\ell)) = \circledcirc'_{\mathfrak{f}}(e_1,\ldots, e_\ell).$$

Thus, we can run the transformation $\hat{\gamma}_{\ell}$ in time $\tO(n^2)$ (because $n^2$ is the input size). Additionally, the output of the counting or detection variants of the \czkc[]~on the transformed input will be exactly equal to the output of \cfkc[]~on the original input. Thus, if we can solve (\#)\czkc[]~in time $T(n)$ we can solve (\#) \cfkc[]~in time $\tO(T(n) +n^2)$. 
\end{proof}

\subsubsection{Factored Zero Triangle is hard from SETH, 3-SUM and APSP}

\begin{lemma}
The \czkch[3]~is implied by any one of SETH, the $3$-SUM hypothesis, or the APSP hypothesis. 
\label{lem:czkch3Imp}
\end{lemma}
\begin{proof}
A violation of \czkch[3]~implies a violation of \ckovh[3]~by Lemma \ref{thm:cztIsAUnit}. A violation of \ckovh[3]~implies a violation of SETH \cite{ryanThesis}.

A violation of \czkch[3]~implies that a $O(n^{1-\epsilon})$ time algorithm exists for the zero triangle problem for some $\epsilon>0$. A  $O(n^{1-\epsilon})$ time algorithm for zero triangle implies a violation of the $3$-SUM hypothesis and the APSP hypothesis \cite{williams2010subcubic}.

So if any one of the three core hypotheses of fine-grained complexity (SETH, $3$-SUM, and APSP) are true then \czkch[3]~is true. 
\end{proof}

\begin{reminder}{Theorem \ref{thm:fzkcIsHardFromAll}}
If \czkc[3]~(even for $b=o(\log n)$ and $g=o(\log(n)/\log\log(n))$) can be solved in $O(n^{3-\eps})$ time for some constant $\eps>0$, then SETH is false, and there exists a constant $\eps'>0$ such that $3$-SUM can be solved in $O(n^{2-\eps'})$ time and APSP can be solved in $O(n^{3-\eps'})$ time.
\end{reminder}
\begin{proof}
This follows from Lemma \ref{lem:czkch3Imp}. A $O(n^{3-\eps})$ time algorithm implies a violation of \czkch[3]. If \czkch[3]~is false then all of SETH, the APSP hypothesis, the $3$-SUM hypothesis are false. 
\end{proof}

\section{Implications from Factored Variants}
\label{sec:harderProblem}

In this section we will show that a series of problems reduce from both counting and decision versions of \fzkc[3], \fzkc[], and \ckfunc[]. 

The reductions from the counting variant of \fzkc[3]~generate counting problems that are hard in the average-case from SETH, 3-SUM, and APSP. The reductions from the counting variant of \fzkc[]~or \ckfunc[]~generate counting problems that are hard in the average-case from SETH. As a result, in this section we produce a suite of problems that are fine-grained hard from the most popular hypotheses of fine-grained complexity. 

In this section we give explicit tight fine-grained reductions from factored problems to many other problems. We will quickly summarize the results of this section. 

We give four tight fine-grained reductions from counting versions of our factored problems. 
We reduce \#\NDMT~from \#\czkc[3], \#\kNLstC[]~from \#\czkc[]~and \#\kELstC[]~from \#\ckfunc[]. Finally, we reduce counting regular expression matching to \#\ckov[2]. 

We also give three tight fine-grained reductions that only work from the detection versions of our factored problems.
We reduce $(k+1)$L-MF to \czkc[]. We reduce Edit Distance to \ckov[2]~and \klcs~from \ckov[].

%that (counting) Partitioned Matching Triangles (\NDMT), (counting) $3$-node labeled st connectivity (\kNLstC[3]), Bounded Labeled Max Flow, and (counting) Regular Expression Matching all reduce from (counting) \fzkc[3]. We show that (counting) $k$-edge colored st connectivity (\kELstC[k]) and (counting) \klcs~reduces from (counting) \ckfunc[].
%Additionally we show that counting \kNLstC[]~ reduces from counting \fzkc. 
%We show that counting \kELstC[]~reduces from counting \fkfunc[].
%Notably, this shows that \fzkc~is a useful as a center of both worst and average case hardness.

\subsection{The Partitioned Matching Triangles Problem solves the Factored Zero Triangle Problem}

%motivation for defining such problem: without node disjointedness the problem has existed before. It is hard from SETH, APSP and 3SUM. Those reductions produced instances that are node disjoint, so \NDMT~ is hard from these problems too. There is node disjoint triangle packing \xxx{had some other name?} that people seem to have worked on \xxx{TODO the results?}. Our problem is a natural mix of these two. 

%\begin{figure}[h]
%  \centering
%  \includegraphics[width=.65\linewidth]{reductions2.eps}
%  \caption{Map of the reductions \xxx{needs update}}
%  \label{fig:reductions1}
%\end{figure}

\begin{reminder}
{Theorem \ref{thm:czkc-to-ndmt}}
If (\#)Partitioned Matching Triangles (\NDMT) can be solved in $T(n)$ time, then we can solve (\#) \czkc[3]~in time $\tilde{O}(T(n)+n^2)$. 
\end{reminder}
\begin{proof}
Let $G=(U,V,W)$ be an instance of \czkc[3], where each edge $e$ is a factored vector. For notation convenience let $uv[j]$ refer to the $j^{th}$ set of the factored vector on the edge from $u$ to $v$.

We define an instance of \NDMT~as a set of $g$ graphs $G_j$ for $j=1,\ldots,g$. We define $G_j$ as follows. For every $u\in U$ add a copy of $u$ in $G_j$ with color $u$. For every $v\in V$ and $w\in W$, add vertices $v_x$ and $w_x$ for all $x=-2^{b+1},\ldots,2^{b+1}$, with color $v$ and $w$ respectively. Note that since $b=o(\log{n})$, $G_j$ has $\tilde{O}(n)$ nodes. 

%Let For every $a\in A$, we define and put vertices $a_1,\ldots,a_g$ in $A'$ with color $a$. For every $b\in B$, we define vertices $b_{j,x}$ for all $j=1,\ldots,g$ and $x=-2^{b+1},\ldots,2^{b+1}$ and put them in $B'$ with color $b$. For every $c\in C$, we define vertices $c_{j,x}$ for all $j=1,\ldots,g$ and $x=-2^{b+1},\ldots,2^{b+1}$ and put them in $C'$ with color $c$. Note that since $g=O(\log{n}/\log\log{n})$ and $b=O(\log\log{n})$, $G'$ has $\tilde{O}(n)$ nodes. 

Now we attach $u\in G_j$ to $v_x\in G_j$ if $x\in uv[j]$. We connect $v_x$ to $w_y$ if $y-x\in vw[j]$ and we connect $w_y$ to $u$ if $-y\in wu[j]$. 
%Now we put an edge from $a_j$ to $b_{j,x}$ if $x\in ab[j]$. We put an edge from $b_{j,x}$ to $c_{j,y}$ if $y-x\in bc[j]$. Finally, we put an edge from $c_{j,y}$ to $a_j$ if $-y\in ac[j]$. 

We prove that the \czkc[3]~instance and the \NDMT~instance are equivalent. For this, consider a zero triangle $uvw$, where the vectors $x_j\in uv[j]$, $y_j\in vw[j]$ and $z_j\in wu[j]$ are picked to have sum zero for each $j$. This corresponds to the triangles $uv_{x_j}w_{y_j+x_j}$ in $G_j$ for each $j$, where all these triangles are of color $(u,v,w)$. Conversely, any set of $g$ triangles of color $(u,v,w)$ in $G_j$s should be of the form $uv_{x_j}w_{y_j}$, and hence from the definition of the \NDMT~ instance we have that for each $j$, $x_j\in uv[j]$, $y_j-x_j\in vw[j]$ and $-y_j\in wu[j]$ and so they correspond to a zero $uvw$ triangle. 
%We need to show that any zero triangle in $G$ corresponds to a $g$-tuple of node-disjoint matching triangles in the $G'$. Consider a zero triangle $abc$, where the vector $x_j\in ab[j]$, $y_j\in bc[j]$ and $z_j\in ac[j]$ are picked to have sum zero for each $j$. This corresponds to the following $g$ node-disjoint triangles of color $(a,b,c)$: $a_jb_{j,x_j}c_{j,y_j+x_j}$ for each $j$. Conversely, any $g$-tuple of node-disjoint triangles of color $(a,b,c)$ correspond to a zero triangle $abc$, where the triangle $a_jb_{j,x}c_{j,y}$ corresponds to the vectors $x\in ab[j]$, $y-x\in bc[j]$ and $-y\in ca[j]$ on the triangle $abc$.
\end{proof}

\subsection{\texorpdfstring{$k$}{Lg}-Node Labeled st Connectivity is hard from Factored \texorpdfstring{Zero-$k$-Clique}{Lg}}
We will show that counting \kNLstC~mod $2^{2k\lg^2(n)}$ is hard from $\#$\czkch. The generated graph will be a dense DAG. Recall that this implies an explicit average-case distribution over which counting \kNLstC~mod $2^{2k\lg^2(n)}$ is hard from worst case \czkch, SETH, the $3$-SUM hypothesis, and the APSP hypothesis. 

\begin{reminder}{Theorem \ref{thm:knlstc-hardness}}
If a $O(|C|^{k-2}|E|^{1-\epsilon/2})$ or $O(|C|^{k-2-\epsilon}|E|)$ time algorithm exists for (counting mod $2^{2k\lg^2(n)}$) \kNLstC~then a $O(n^{k-\epsilon})$ algorithm exists for (\#)\czkc.
\end{reminder}
\begin{proof}

Let $G=(V_1,\ldots,V_k)$ be an instance of \czkc. We reduce this instance to an instance of \kNLstC~ as follows. We begin by adding the special node $s$ and the special node $t$. %Recall that every edge in $G$ has a factored vector on it and every factored vector is formed by $g$ subsets of $\{0,1\}^b$. 
%We give every node $v$ in $G$ the color $v$ which we will use in $G*$. 
We will build $g$ gadgets and put them after each other serially. The nodes in the gadgets will be assigned colors associated to the nodes of $G$. Each gadget will be designed to check if given $k$ colors (and thus $k$ nodes in $G$) whether the $i^{th}$ subset of the $\binom{k}{2}$ factored vectors represented do have a zero sum. See Figure \ref{fig:gadgetNLstC} for a representation of our construction.

\begin{figure}[ht]
    \centering
    \includegraphics[width=\textwidth]{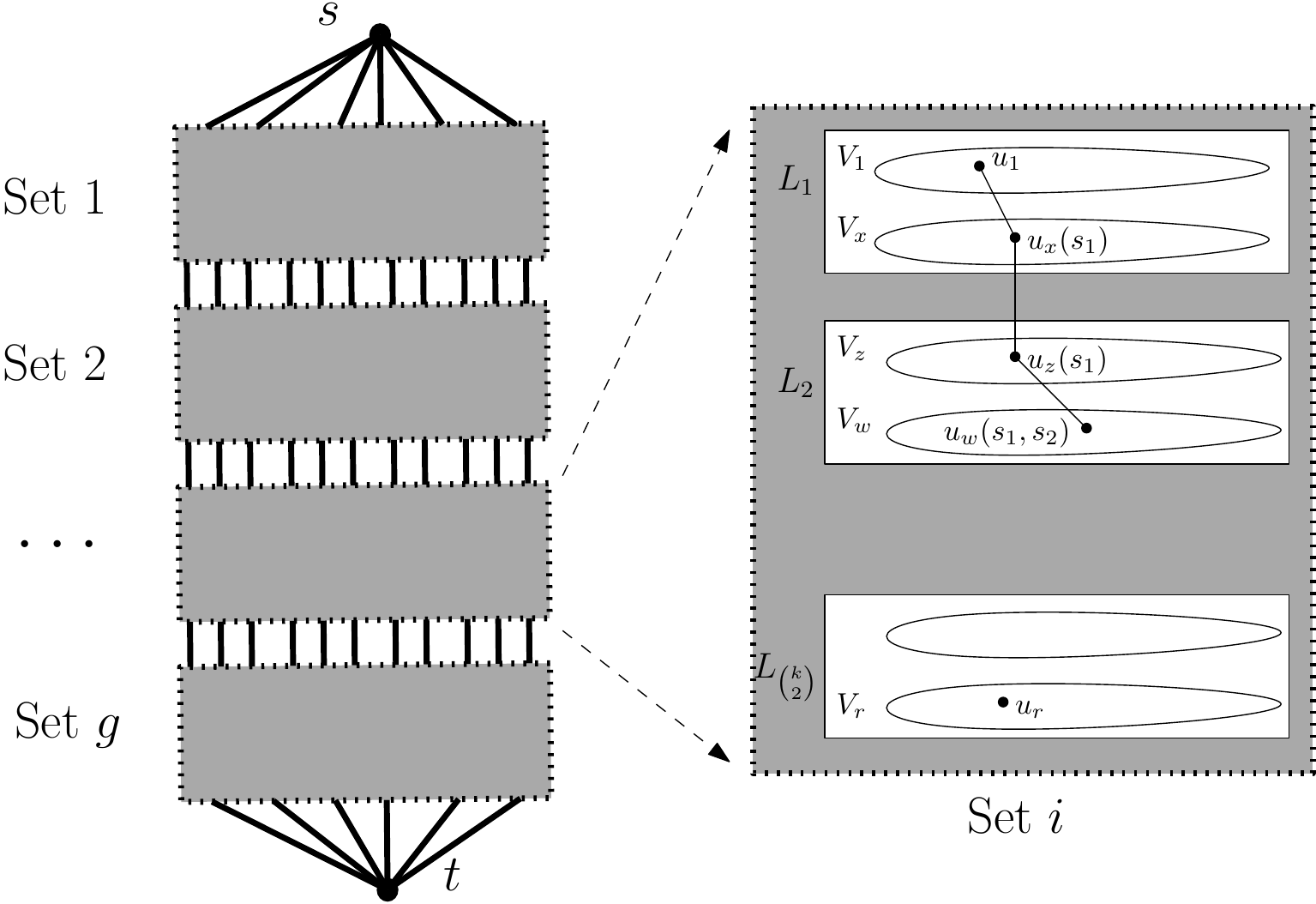}
    \caption{Left: The gadget structure for \kNLstC. Right: Inside of a set gadget.}
    \label{fig:gadgetNLstC}
\end{figure}
The gadget for set $i$ consists of 
%starts with a copy of all the nodes in $V_1$. We form this gadget by making 
$\binom{k}{2}$ layers $L_1,\ldots,L_{\binom{k}{2}}$. Each layer $L_j$ represents the edges from $V_x$ to $V_y$ for some $x,y\in \{1,\ldots,k\}$ as follows: $L_j$ consists of two layers itself, one for $V_x$ and one for $V_y$. For each vertex $u_x\in V_x$, we add $2^{b(j-1)}$ nodes $u_x(s_1,\ldots, s_{j-1})$ where $s_1,\ldots,s_{j-1} \in \{0,1\}^b$, and we color these nodes with the color $u_x$. So we have a total of $n2^{b(j-1)}$ vertices. These vertices represent that we have chosen a particular node $u_x \in V_x$ and that in the first $j$ edges we have chosen the $j-1$ vectors $s_1,\ldots,s_{j-1}$ from the sets of the previous $j-1$ edges. 
%We call the $2^{b(j-1)}$ nodes for $u_x$ by $u_x(s_1,\ldots, s_{j-1})$ where $s_1,\ldots,s_j \in \{0,1\}^b$. 
%We color the node $u_x(s_1,\ldots, s_{j-1})$ with the color of $u_x$.\\
For the second layer of $L_j$, for each $u_y\in V_y$ we add $2^{bj}$ nodes $u_y(s_1,\ldots, s_{j})$ where $s_1,\ldots,s_j \in \{0,1\}^b$, and we color these nodes with the color $u_y$. 
%We then have $n$ sets of $2^{b\cdot j}$ nodes, these represent $u_y$ and the $j$ vectors we have picked, we label these  $u_y(s_1,\ldots, s_{j-1},s_j)$.  We color the node $u_y(s_1,\ldots, s_{j-1},s_j)$ with the color $u_y$. 
We add edges between nodes  $u_x(s_1,\ldots, s_{j-1})$ and $u_y(s_1,\ldots, s_{j-1},s_j)$ iff the $i^{th}$ set of the factored vector of edge $(u_x,u_y)$ contains the string $s_j$, i.e. $s_j\in u_xu_y[i]$.

Now to specify the edges between layers, suppose that layer $L_{j+1}$ deals with the edges between $V_z$ and $V_w$.
For every $j$-tuple $(s_1,\ldots, s_j)$, add an edge from $u_y(s_1,\ldots, s_j)$ in $L_j$ to $u_z(s_1,\ldots,s_j)$ in $L_{j+1}$ for every $u_y\in V_y$ and $u_z\in V_z$. If $z=y$ we can skip this step and just use the same set nodes. 

Finally, we do something special for layer $L_{\binom{k}{2}}$. Lets say that layer $L_{\binom{k}{2}}$ summarizes the edges between $V_x$ and $V_y$. For the nodes associated to $V_y$, instead of having vertices $u_y(s_1,\ldots,s_{\binom{k}{2}})$, we put only one vertex $u_y$. %For our nodes representing  $V_x$ we do the same matching with layer $L_{\binom{k}{2}-1}$ as we do with the earlier layers. However, we have a singular representation of the nodes in $V_y$. So , a single node representing $u_y$. 
We connect $u_x(s_1, \ldots, s_{\binom{k}{2}-1})$ to $u_y$ if and only if the set $i$ of the edge $(u_x,u_y)$ has a vector $s_{\binom{k}{2}}$ such that the vectors $s_1, \ldots, s_{\binom{k}{2}-1},s_{\binom{k}{2}}$ sum to zero. Note that given a fixed choice of $s_1, \ldots, s_{\binom{k}{2}-1}$  there is a single vector $s_{\binom{k}{2}}$ such that they all sum to zero together. 

This forms a layered directed graph, where edges go from layer $L_i$ to layer $L_{i+1}$. We also assume that $L_1$ represents the edges from $V_1$ to $V_x$ for some $x$. We add an edge from $s$ to all vertices of the first layer of $L_1$ and an edge from all vertices in the last layer of $L_{\binom{k}{2}}$ to $t$. A representation of the layers is represented in Figure \ref{fig:gadgetNLstC}.

By this construction, a path with colors $u_1, u_2, \ldots, u_k$ that goes through the $j^{th}$ gadget represents a zero sum within the $j^{th}$ sets on the $\binom{k}{2}$ edges between $u_1,\ldots,u_k$. 

In our graph the number of colors $|C|$ is $O(n)$ and $|E|=O(n^2)$ so a $O(|C|^{k-2}|E|^{1-\epsilon/2})$ algorithm and a $O(|C|^{k-2-\epsilon}|E|)$ algorithm both run in $O(n^{k-\epsilon})$ time. The number of solutions to both problems is the same, thus the counts are the same. The maximum count of \czkc[]~is $2^{gb}n^k = O(2^{\lg^2(n)+k\lg(n)})$. Notably, this is less than $2^{2k\lg^2(n)}$, so the count from the \kNLstC~instance will be less than the count for the \czkc[]~instance.
\end{proof}

Note that we can count \kNLstC~mod $2^{2k\lg(n)^2}$ with $|C| =n$ and $|E|=n^2$ in time $\tilde{O}\left(n^k\right)$. 

\begin{corollary}
If  \#\czkch[]~is true then \#\kNLstC~(mod $R$) takes $|C|^{k-2 \pm o(1)}|E|^{1 \pm o(1)}$ (where $\lg(R) =n^{o(1)}$).
\label{cor:knlstcWithHypo}
\end{corollary}
\begin{proof}
By Theorem \ref{thm:knlstc-hardness} if  \#\czkch[]~is true then \# \kNLstC~(mod $R$) takes at least $|C|^{k-2 - o(1)}|E|^{1 - o(1)}$ time. 

By Theorem \ref{thm:algForknlstc} there is a $|C|^{k-2 + o(1)}|E|^{1 + o(1)}$ time algorithm for counting \kNLstC~mod $R$.
\end{proof}

\subsection{\texorpdfstring{$k$}{Lg}-Edge Labeled st Connectivity is hard from Factored \texorpdfstring{$k$}{Lg} Function Problems \texorpdfstring{(F$k$-f)}{Lg}}
In this subsection we will show hardness from 
The edge labeled version of $st$ connectivity. This reduction will get hardness from \#\ckfunc[]. Note that while \kNLstC[]~has a $\tO(C^{k-2}E)$ algorithm, \kELstC[]~has a more expensive $\tO(C^{k-1}E)$ algorithm. In this section we will show that the \kELstC[]~algorithm is optimal up to sublinear factors if \ckfunch[]~is true (note that this algorithm is thus also implied to be tight by SETH). 

While our reduction to \kNLstC[]~generated a dense graph, our reduction to \kELstC[]~generates a \emph{sparse} graph. The sparsity allows for a tight reduction to the \ckfunch[]~problem. However, because that our reduction requires sparsity to be tight, we have not been able to reduce \czkc~to \kELstC[].

\begin{reminder}{Theorem \ref{kelstcCKFUNC}}
If a $\tilde{O}(|E||C|^{k-1-\epsilon})$ or $\tilde{O}(|E|^{1-\epsilon}|C|^{k-1})$ time algorithm exists for (counting ~mod $ 2^{2k\lg^2(n)}$) \kELstC[], then a $\tilde{O}(n^{k-\epsilon})$ algorithm exists for (\#)\ckfunc[].
\end{reminder}
\begin{proof}
Given an instance of \ckfunc[]~which takes $k$ lists $V_1,\ldots,V_k$ of factored vectors, we produce an instance of \kELstC[]~with $\tilde{O}(n)$ colors and $\tO(n)$ edges. In the \ckfunc[]~instance, let $u_j[i]$ be the $i^{th}$ subset of the vector $u_j\in V_j$. We use the vectors in the \ckfunc[]~instance as colors in the \kELstC[]~instance.
%let $V_j=\{u_j^1, u_j^2,\ldots,u_j^{n}\}$, where $u_j^x$ is the $x^{th}$ factored vector in the $V_j$. Let the $i^{th}$ subset of $u_j^x$  be $u_j^x[i]$.

We start by adding two nodes $s$ and $t$. We will make $g$ gadgets, $G_1, \ldots, G_g$, where $G_i$ handles the $i^{th}$ set of the factored vectors, i.e. $u_j[i]$ for all $u_j\in V_j$ for all $j$. In each gadget $G_i$ we have $k$ layers of vertices $L^i_1, \ldots, L^i_k$, where the vertex set $L^i_j$ represents the factored vectors in $V_j$. Finally we attach these gadgets one after the other serially.%We will have a color for each factored vector, lets call the $x^{th}$ factored vector in the $j^{th}$ list $u_j^x$. Lets call the $i^{th}$ subset of $u_j^x$, $u_j^x[i]$.

For $j\le k$ the layer $L_j^i$ has two layers itself, one with $n2^{bj}$ nodes and one with $2^{bj}$ nodes. For each node $u_j\in V_j$, we add the nodes $u_j(s_1,\ldots,s_j)$ to the first layer of $L_j^i$ for all $s_1,\ldots,s_j\in \{0,1\}^b$, so adding $n2^{bj}$ nodes in total. For the second layer of $L_j^i$, we add nodes $l_i(s_1,\ldots,s_j)$ for all $s_1,\ldots,s_j\in \{0,1\}^b$. For each vertex $u_j$, we add a matching from the $2^{bj}$ nodes associated to vector $u_j$ to the nodes in the second layer, connecting $u_j(s_1,\ldots,s_j)$ to $l_i(s_1,\ldots,s_j)$. We color these edges with $u_j$. Note that this is how we achieve sparsity. Every other layer has $\tO(1)$ nodes in it. So every node (other than $s$ and $t$) has an out-degree of $\tO(1)$. 

%$j$-tuple $(s_1,\ldots,s_j)$, we connect $u_j^x(s_1,\ldots,s_j)$ to $(s_1,\ldots,s_j)$ for all $x$. 

%The first layer will be $2^{bj}$ nodes for every factored vector in the $j^{th}$ list. Each of these nodes represents all possible $j$ tuples of strings from $\{0,1\}^b$. From the $2^{bj}$ nodes representing vector $u_j^x$ we have a matching with edges ending in the second layer, a set of $2^{bj}$ nodes. These edges are colored with the color of factored vector $u_j^x$.

We add edges from the second layer of $L^i_j$ to the first layer of $L^i_{j+1}$. For $u_{j+1}\in V_{j+1}$, we connect $l_i(s_1,\ldots, s_j)\in L^i_j$ to $u_{j+1}(s_1,\ldots, s_j, s_{j+1})\in L^i_{j+1}$ if, and only if, $s_{j+1} \in u_{j+1}[i]$. We color this edge with $u_{j+1}$.

The full effect of this means that by layer $L^i_{k}$ a path from the beginning to the end of the gadget with the colors of a given set of $k$ vectors implies those vectors have the corresponding set of vectors in their $i^{th}$ sets.  We will only add outgoing edges from nodes in the second layer of $L^i_{k}$ only if $f(s_1,\ldots, s_k)=1$.

We add edges between gadgets $G_i$ and $G_{i+1}$ by adding edges between the second layer of $L^i_{k}$ and the first layer of $L^{i+1}_1$ as follows. We connect the node $l_i(s_1,\ldots,s_k)\in G_i$ to $u_1(s'_1)\in G_{i+1}$ for some $u_1\in V_1$ if and only if $f(s_1,\ldots,s_k)=1$ and $s'_1\in u_1[i+1]$. We color this edge with $u_1$.

%Specifically we add outgoing edges from the node associated with $(s_1,\ldots, s_k)$ in $V_i^k$ only if $f(s_1,\ldots, s_k)=1$; call these nodes $S_i$. Nodes in the first layer of $V_{i+1}^{k}$ are associated with both a $b$ length vector $(s_1)$ and a node $u_{1}^x$. We add edges between every node in $S_i$ to every node where $(s_1) \in u_{1}^x[i+1]$. These edges are colored the color of $u_1^x$. 

Now we deal with $s$ and $t$. We add edges from 
%We do something slightly special for the edges from $s$ to the first layer of $G_1$, we add edges from 
$s$ to all nodes $u_1\in L_1^1$ if $s_1 \in u_{1}[1]$. These edges are colored with $u_1$. Further, 
%We do something slightly special for the edges from $G_g$ to $t$. 
we add edges from the first layer of $L^g_k$ to $t$ directly, removing the second layer of $L^g_k$. We only add edges from $u_k(s_1,\ldots,s_k)$ for $u_k\in V_k$ to $t$ iff $f(s_1,\ldots,s_k)=1$. We color this edge with $u_k$.

First, note that we always add edges between two layers of size $o(n)$ and $\tO(n)$, so adding at most $\tO(n)$ edges between them. Since we have $O(1)$ layers, our graph has $\tO(n)$ edges in total. 

Given this graph setup, if we pick $k$ colors for example associated with $u_1(1), u_2(2), \ldots, u_k(k)$ then the number of paths from $s$ to $t$ using only those colors of edges will correspond to the outcome of $$\circledcirc(u_1(1), u_2(2), \ldots, u_k(k))$$ as defined in the preliminaries. As a result, the sum over all $k$ tuples of colors will be the count of the output of the \ckfunc[]~instance. The count of a \ckfunc[]~instance is at most $n^k 2^{bg} = o( 2^{2k\lg^2(n)})$. So if $R = \Omega( 2^{2k\lg^2(n)})$ the count mod $R$ and the count are the same. 
\end{proof}

\begin{corollary}
If \ckfunch[]~(\#\ckfunc[]) is true then \kELstC(\#\kELstC~mod $R$) takes $|C|^{k-1\pm o(1)}|E|^{1\pm o(1)}$ (when $\lg(R) =n^{o(1)}$).
\label{cor:kelstcWithHypo}
\end{corollary}
\begin{proof}
By Theorem \ref{kelstcCKFUNC} if \ckfunch[]~(\#\ckfunch[])~is true then \kELstC ~(\#\kELstC~mod $R$) takes at least $|C|^{k-1-o(1)}|E|^{1-o(1)}$ time. 

By Theorem \ref{thm:kelstcALG} there is a $|C|^{k-1+o(1)}|E|^{1+o(1)}$ time algorithm for counting mod $R$ \kELstC.
\end{proof}

\subsection{\texorpdfstring{$(k+1)$}{Lg} Labeled Max Flow solves Factored Zero-\texorpdfstring{$k$}{Lg}-Clique}
\begin{reminder}{Theorem \ref{klmf}}
If $(k+1)$L-MF can be solved in $T(n)$ time, then we can solve \czkc~in time $\tilde{O}(T(n)+n^2)$.
\end{reminder}
\begin{proof}
We use the set gadgets from Theorem \ref{thm:knlstc-hardness} and instead of placing them serially, we make a parallel network $G'$ as shown in Figure \ref{fig:flow}. More particularly, let $SG_1,\ldots,SG_g$ be the set gadgets from Theorem \ref{thm:knlstc-hardness}. Add $s_1,\ldots,s_g$ with a source node $s$ to the graph. Add $t_1,\ldots,t_g$ with a sink node $t$ to the graph. This completes the definition of the vertices of $G'$.

We attach $s$ to all $s_j$ and all $t_j$ to $t$ with label $l^*$ for $j=1,\ldots, g$. For each $j$, we attach $s_j$ to all the nodes in the first layer of $SG_j$, which is a copy of $V_1$. Let the label of any $s_ju$ edge be $u$ where $u\in V_1$. 
Connect all the nodes in the last layer of $SG_j$ to $t_j$ for all $j$. Suppose that the last layer of $SG_j$ corresponds to $V_y$. Let the label of any edge $u_yt_j$ be $u_y$. Let the label of any edge $(r,z)$ in any set gadget $SG_j$ be the same as the color of $r$, since $G'$ is supposed to be an edge-labeled graph. All the edges are unit capacitated. This completes the definition of $G'$.

\begin{figure}[ht]
    \centering
    \includegraphics[width=0.6\textwidth]{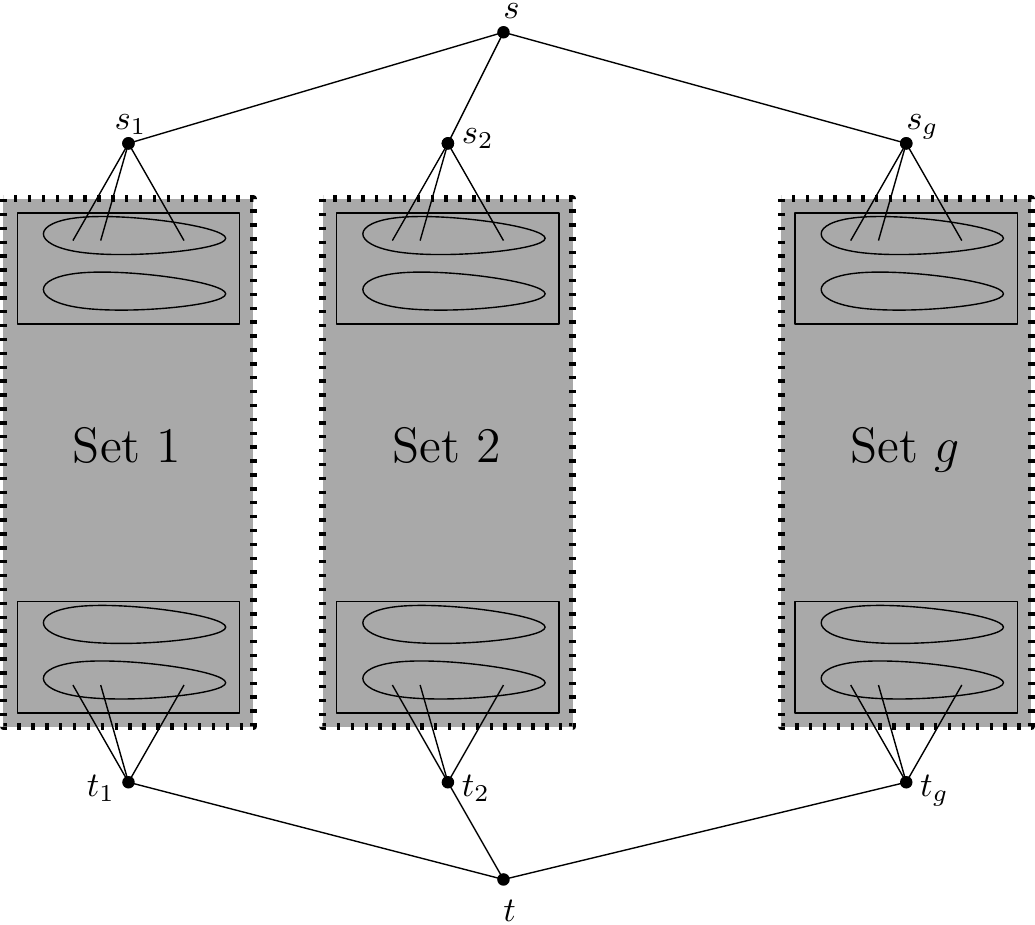}
    \caption{$(k+1)$ labeled max flow instance structure.}
    \label{fig:flow}
\end{figure}

First note that the maximum flow is at most $g$ since the outdegree of $s$ is $g$ and the graph is unit-capacitated. So the flow going through each set gadget is at most $1$, which means that there is at least one path from $s_j$ to $t_j$ through $SG_j$. From Theorem \ref{thm:knlstc-hardness} any zero weight $k$-clique corresponds to $g$ paths, one in each set gadget, using $k$ labels corresponding to the $k$ vertices of the clique. So any zero weight $k$-clique corresponds to a $(k+1)$ labeled flow of size $g$ from $s$ to $t$. Conversely, if there is a $(k+1)$ labeled flow of size $g$ from $s$ to $t$, it must correspond to (at least) one path from $s_j$ to $t_j$ in $SG_j$ for each $j$ with all the $g$ paths having the same $k$ labels, which corresponds to a zero $k$-clique by Theorem \ref{thm:knlstc-hardness}.
\end{proof}

\subsection{Regular Expression Matching is hard from Factored OV}

We are going to reduce \ckov[2]~to regular expression matching. First, we define type and depth of a regular expression. Intuitively, the structure of the operations in a regular expression is called its \textit{type}, which is represented by a tree with nodes labeled with operations. Let $\bullet$ be an arbitrary operator. A tree $T$ with root node $\bullet$ means that all the first level operations of a regular expression $E$ of are $\bullet$, i.e. $E=A_1\bullet A_2\bullet\ldots\bullet A_\ell$, where $A_i$s are regular expressions. The type of each $A_i$ can be the subtree with any of the children of the root node as its root. The \textit{depth} of a regular expression is the longest root-leaf path in the type tree of the regular expression.

We reduce (\#)\ckov[2]~to (\#)regular expression matching where the pattern is a depth 5 pattern of type $T_0$ shown in Figure \ref{fig:regex-type}, and we give an $O(mn)$ algorithm for counting such patterns in $O(mn)$ time (Theorem \ref{thm:regexcounting} in the Appendix).

%To clarify, by "counting" regular expression matching we mean the number of subset alignments of the pattern in the text, and not just the number of substrings that can be derived from the pattern. In order to avoid having big numbers in our counting versions, we take numbers mod a poly $n$ integer. By subset alignment we mean an alignment of the pattern on a substring of the text, and by full alignment we mean an alignment of the pattern on the text.

\begin{figure}[ht]
    \centering
    \includegraphics[width=0.3\textwidth]{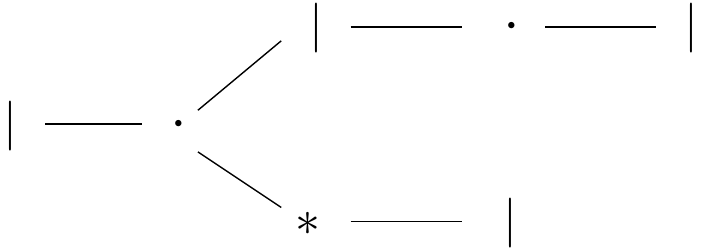}
    \caption{Type $T_0$ of the regular expression of Theorem \ref{thm:regex}. Where ``$|$" is the OR operator, ``$\cdot$" is the concatenation operator, and $*$ is the Kleen star operator.}
    \label{fig:regex-type}
\end{figure}

\begin{reminder}{Theorem \ref{thm:regex}}
Let $R$ be an integer where $\lg{(R)}$ is subpolynomial. If you can solve (\# mod ${R}$) regular expression matching in $T(n)$ time, then you can solve (\# mod $R$) \ckov[2]~in $\tO(T(n)+n)$ time. 
%Let $R$ be an integer where $\lg{R}$ is subpolynomial. If you can solve counting (mod ${R}$) regular expression matching in $T(n)$ time, then you can solve counting (mod $R$) \ckov[2]~in $\tO(T(n)+n)$ time. 
\end{reminder}

\begin{proof}
We use the proof of Theorem $1$ in \cite{regularexpression} that shows hardness for patterns of type ``$|\cdot|$", where ``$|$" is the OR operator, ``$\cdot$" is the concatenation operator and the type tree is a path of length two with node $|,\cdot$ and $|$ respectively. In \cite{regularexpression} authors start with any $2OV$ instances $(A,C)$ where $A=\{a^1,\ldots,a^n\}$ and $C=\{c^1,\ldots,c^n\}$ are sets of $n$ vectors of dimension $d$ and reduce it to an instance of regular expression matching with pattern $p$ constructed from $A$ (and independent from $C$) and text $t$ constructed from $C$ (and independent from $A$) both of $O(dn)$ size, where any orthogonal pair $(a,c)$ with $a\in A$ and $c\in C$ corresponds to an alignment of $p$ on a substring of $t$, and conversely any alignment of $p$ on $t$ corresponds to an orthogonal pair $(a,c)$. More particularly, pattern $p=VG(a^1)|\ldots|VG(a^n)$ consists of the OR of vector gadgets $VG(a)=CG(a_1)\cdot \ldots \cdot CG(a_d)$ where $CG$ is a coordinate gadget and $a_i$ is the $i$th bit of vector $a$. Each coordinate gadget is aligned on a single bit. The text $t=VG'(c^1)2\ldots 2VG'(c^n)$ consists of vector gadgets $VG'(c)=c_1c_2\ldots c_d$ which is the bit representation of the vector $c$. We have that $a^i\cdot c^j=0$ iff there is an alignment of $VG(a^i)$ on $VG'(c^j)$. As a result, the number of orthogonal pairs in $(A,C)$ is the number of subset alignments of $p$ on $t$.

We use the above construction for our factored vectors. Note that if $w$ is a factored vector, for any $j\in \{1,\ldots,g\}$ we can construct a pattern (or a text) of length $O(|w[j]|b)$ using the vectors in $w[j]$ which have length $b$.

Consider an instance $(U,V)$ of \ckov[2]~where $U=\{u^1,\ldots,u^n\}$ and $V=\{v^1,\ldots,v^n\}$ are sets of $n$ factored vectors. We construct the pattern $P$ using $U$ and the text $T$ using $V$. We first construct the pattern. Let $p_i[j]$ be the pattern corresponding to $u^i[j]$ using the construction of \cite{regularexpression} for  $i=1,\ldots,n$ and $j=1,\ldots,g$. Note that the symbols used in $p_i[j]$ are $0,1$. Let $p'_i[j]=[0|1|2]^*\cdot p_i[j]\cdot [0|1|2]^*$, where ``$*$" is the Kleen star operator.
Define the ``pattern factored vector gadget" $PVG(i)$ for $u^i$ as follows:
$$
PVG(i)=p'_i[1]\cdot 3\cdot p'_i[2]\cdot 3\cdot \ldots \cdot 3 \cdot p'_i[g]
$$
Let the pattern $P$ be the following:
$$
P=PVG(1)|PVG(2)|\ldots|PVG(n)
$$
Note that the length of $P$ is $\tO(n)$, since we have that  $|p_i[j]|=O(|u^i[j]|)=O(2^b.b)=o(n)$, $|p'_i[j]|=O(p_i[j])$ and the number of occurrence of the symbol $3$ is $(g-1)n$. The number of symbols in a $*$ expression is also $O(gn)$ since the number of $*$ expressions is $2ng$. So the total number of symbols in $P$ is $\tO(n)$, and as a result the length of $P$ is $\tO(n)$.

Now we construct the text. Let $t_i[j]$ be the text corresponding to $v^i[j]$ using the construction in \cite{regularexpression} for $i=1,\ldots,n$ and $j=1,\ldots,g$.  Note that the symbols used in $t_i[j]$ are $0,1,2$. Define the ``text factored vector gadget" $TVG(i)$ for $v_i$ as follows:
$$
TVG(i)=t_i[1]3t_i[2]3\ldots 3t_i[g]
$$
Let the pattern $T$ be the following:
$$
T=TVG(1)4TVG(2)4\ldots 4TVG(n)
$$
Similar to $P$, the length of $T$ is $\tO(n)$.

Now we have to show that there is a one to one correspondence between orthogonal vectors in the \ckov[2]~instance and the number of subset alignments of $P$ on $T$. First consider an orthogonal pair $(u_i,v_j)$, where for each $k=1,\ldots,g$ the vector chosen from $u^i[k]$ is the $r_k$th vector $u^i[k][r_k]$, and the vector chosen from $v^j[k]$ is the $s_k$th vector $v^j[k][s_k]$. So we have that for each $k$, $u^i[k][r_k].v^j[k][s_k]=0$. This means that there is a subset alignment of $p_i[k]$ on $t_j[k]$ corresponding to $u^i[k][r_k]$ and $v_j[k][s_k]$ for all $k$. We use the $[0|1|2]^*$ parts of $p'_i[k]$ to cover the rest of $t_j[k]$ and thus we get a full alignment of $p'_i[k]$ on $t_j[k]$. Having these alignments for each $k$, they extend uniquely to a full alignment of $PVG(i)$ on $TVG(i)$.

Conversely, suppose that there is a subset alignment of $P$ on $T$. Note that the first level of the pattern consists of ORs, so any alignment should choose some $i\in \{1,\ldots,n\}$ and align $PSG(i)$ on $T$. On the other hand, symbol $4$ is not used in the pattern $P$. So $P$ should be aligned on $TSG(j)$ for some $j$. So assume that in this subset alignment, $PSG(i)$ is aligned on $TSG(j)$. Since there are exactly $g-1$ symbols ``$3$" that are concatenated in $PSG(i)$ and there are exactly $g-1$ symbols ``$3$" in $TSG(j)$, the $3$s should be aligned to each other. So $p'_i[k]$ is fully aligned to $t_j[k]$. Recall that $p'_i[k]=[0|1|2]^*\cdot p_i[k]\cdot [0|1|2]^*$. So $p_i[k]$ should be aligned to $t_j[k]$. So by the construction of \cite{regularexpression} there are unique $r_k$ and $s_k$ where the vector gadget for $u^i[k][r_k]$ is fully aligned to the vector gadget for $v^j[k][s_k]$, which means that these two vectors are orthogonal. Since this is true for every $k$, $u^i$ and $v^j$ are orthogonal, and hence the number of factored orthogonal vectors in the \ckov[2]~instance equals to the number of subset alignments of $P$ on $T$.
%Now since the vector gadgets in $t_j[k]$ are separated by the symbol $2$, the $[0|1|2]^*$ parts of $p'_i[k]$ cover the rest of $t_j[k]$, and so $p'_i[k]$ fully aligns $t_j$
%$p_i[k]$ is aligned to some vector this alignment corresponds to  in $v^i[k]$ and $u^j[k]$ being orthogonal. Now since $t_j[k]$ is basically the concatenation of vectors in $t_j[k]$ separated by the symbol $2$, 
%From \cite{regularexpression}  
\end{proof}

\subsection{Longest Common Subsequence and Edit Distance}
\label{subsec:lcs}

We are going to look at the \klcs~problem in this subsection. We show that \klcs~is hard from \ckov[]. We note that this sort of proof should also work for other string similarity measures. In a work of Bringmann and K{\"{u}}nnemann they show a general framework for proving hardness for string comparisons on two strings from 2-OV \cite{alignmentGadget}. Presumably this framework can be expanded to work for \ckov[], however, generating this framework is out of the scope of this paper. We will note however that the only additional gadget you seem to need to solve \ckov[]~is a selector gadget for at most $\tO(1)$ strings each of length $\tO(1)$. This means even expensive gadgets are acceptable.

We will first show that weighted $k$-LCS is hard from \ckov[]. 
\begin{definition}[Weighted Longest Common Subsequence (WLCS) \cite{LCSisHard}]
For $k$ sequences $P_1, \ldots, P_k$ of
length $n$ over an alphabet $\Sigma$ and a weight function $w : \Sigma \rightarrow [K]$, let $X$ be the sequence that appears
in all of $P_1, \ldots, P_k$ as a subsequence and maximizes the expression $W(X) = \sum_{i=1}^{|X|} w(X[i])$. We say that
$X$ is the WLCS of $P_1, \ldots, P_k$ and write WLCS$(P_1,\ldots, P_k) = W(X)$. The Weighted Longest Common
Subsequence problem asks to output WLCS$(P_1, \ldots, P_k)$.
\end{definition}

We will then use this lemma from a previous work to show that $k$-LCS is hard from $k$-WLCS if the weights are small enough \cite{LCSisHard}.

\begin{lemma}
If the \klcs~of $k$ sequences of length $O(Kn)$ over $\Sigma$ can be computed in time $T(n)$ then the \kwlcs[]~of $k$ sequences of length $n$ over $\Sigma$ with weights $w : \Sigma \rightarrow [K]$
can be computed in $\tO(T(n) K)$ time \cite{LCSisHard}.
\label{lem:WLCSandLCS}
\end{lemma}
%\begin{proof}
%This appears in \cite{LCSisHard} in Lemma 2.
%\end{proof}

We want to use the ideas and gadgets of Abboud, Backurs and Vassilevska Williams \cite{LCSisHard}. We basically want to ask, given sets $V_1,\ldots,V_k$ of factored vectors, are there $k$ factored vectors $v_1\in V_1, \ldots, v_k\in V_k$ such that for all $i$ there exist vectors $u_1\in v_1[i],\ldots, u_k\in v_k[i]$ such that those vectors are orthogonal. 
Notably, once you have specified the factored vectors and the index $i$ what remains is a (small) orthogonal vectors instance. The construction from \cite{LCSisHard} produces a fixed longest common sub-sequence value if there is an orthogonal $k$ tuple. So, if we can construct a setup where the output of the WLCS is basically a concatenation of gadgets for each index $i$ then we will get the value we want for any given pair of vectors. We will need to add some gadgets to force the WLCS to ``pick'' a set of vectors. 

Notably, ``selector gadgets" from \cite{LCSisHard} serve the purpose of forcing the WLCS to choose which $k$ factored vectors to compare. And, if we have a gadget for every subset (so every $i\in[1,g]$) and put a high value set of symbols between them it forces the gadgets to not interact or loose that value. For this we want a ``parallel gadget". Thus, we get a WLCS that is roughly the concatenation of the WLCS of each of the $g$ gadgets. This gives us the desired result. 

\paragraph{Gadgets of General Use}

First we will describe the \textbf{selector gadget}.
\begin{lemma}
As input we are given $k$ lists $L_1, \ldots, L_k$ each of which contain $n$ strings of length at most $\ell$  (e.g. $s_{i,1}, \ldots, s_{i,n} \in L_i$ and $|s_{i,j}| \leq \ell$) with an alphabet $\Sigma$ and weights in the range $[K]$. Let $M$ be the maximum value of WLCS$(s_{1,j_1}, \ldots, s_{k,j_k})$ over all choices of $j_1,\ldots,j_k \in [1,n]$.

We can generate a $k$-WLCS instance $P_1, \ldots, P_k$ with $k$ symbols added to $\Sigma$, the new range of weights being $[2 \ell Kn]$ and length $|P_i| = O(n^2+n\ell)$, such that WLCS$(P_1, \ldots, P_k)=C_{sel}+M$ for $C_{sel} = (2kn) (2\ell Kn)$.
\label{lem:selectorGadget}
\end{lemma}
\begin{proof}
We introduce $k$ symbols $@_1, @_2, \ldots, @_k$ for this selector gadget. We assign a weight of $2\ell K n$ to all the symbols $@_j$ (note this is larger than the total weight of any given string $s_{i,j}$). 
For convenience by $@_j^{(x)}$ we mean $x$ copies of the symbol $@_1$. 

We first define a helper gadget for separating our strings
$$STG_i(s) = @_i @_{i-1}^{(2n)} \ldots @_{1}^{(2n)} s @^{(2n)}_1 \ldots @_{i-1}^{(2n)} @_i.$$

We can now define our output strings:
$$P_i =@_{k}^{(2n)} \ldots @_{i+1}^{(2n)} STG_i(s_{i,1})\ldots STG_i(s_{i,n})@_{i+1}^{(2n)} \ldots @_{k}^{(2n)} .$$

Note that every $STG_i()$ gadget is of length $O(n^2+ \ell)$ and $O(n)$ $STG_i()$ gadgets are used. The additional $@$ symbols make up at most $kn$ symbols on every string. So the total length of each string $P_i$ is at most $O(n^2+n\ell)$. The largest weight we use is for the $@$ symbols, they have a weight of $2 \ell K n$. There are $k$ $@$ symbols, so we increase the alphabet by $k$.

Let $?$ be some string made of symbols from the original alphabet (so no $@$ symbols). We will use this to make arguing easier. 
We claim that the weighted longest common subsequence will look like this:
$$@_{k}^{(y_k)} \ldots @_{2}^{(y_2)}  @_{1}^{(y_1)}?  @_{1}^{(2n-y_1)} @_{2}^{(2n-y_2)} \ldots @_{k}^{(2n-y_k)}.$$

Let us argue for this claim. First, every symbol $@_i$ appears only $2n$ times in string $P_i$ so it can not appear more often. Second, in the string $P_1$ the only place that symbols $@_k,\ldots, @_2$ appear is at the start and end of the string in the order presented above. 
For $@_k$ it appears only $2n$ times in the string $P_k$. In every other string $@_k$ only appear at the start and end of strings. To match all $2n$ copies of the $@_k$ symbol we must align a single $STG_k(s)$ gadget with the other strings $P_i$ for $i<k$. Given that we are matching a single $STG_k(s)$ string note that the only locations that $@_{k-1},\ldots, @_1$ symbols appear are around the string $s$ in decreasing and then increasing order. So, if we do try to match all $2n$ copies of every symbol $@_i$ we must get \kwlcs[]~that looks like the above. 

Now we will argue that you can match $2n$ copies of every symbol $@_i$. Consider the string $P_i$, if you pick any single $STG_i(s_{i,j})$, all the symbols $@_i$ and the ``intro'' and ``outro'' strings of $@_{k}^{(2n)} \ldots @_{i+1}^{(2n)}$ and $@_{i+1}^{(2n)} \ldots @_{k}^{(2n)}$ together make a string of the form:
$$@_{k}^{(2n)} \ldots @_{i}^{(2j-1)} \ldots  @_{1}^{(2n)}?  @_{1}^{(2n)} \ldots @_{i}^{(2n-2j+1)} \ldots @_{k}^{(2n)}.$$
If we match $k$ of these we get our claimed string where $y_i = 2j-1$. \\
Now we must argue that one wants to match all $@_i$ symbols possible. Note that every $@_i$ symbol is worth more than all non $@$ symbols in the entire string. Given this, we must prefer matching all $@$ symbols to any other goal. 

So, given that the \kwlcs[]~will have the described form we can now note the following. The $?$ that appears must be the \kwlcs[]~of $k$ strings $s_{i,j}$. In every $P_i$ in order to match all $2n$ symbols $@_i$  and only have non-$@$ symbols in the middle of the string one must select a single $STG()$ gadget to be included in the \kwlcs[]. 

So the \kwlcs[]~will include $2n$ copies of $@_i$ symbols and the \kwlcs[]~of the $k$ strings that have the largest \kwlcs[]. 
\end{proof}

Now we will describe the \textbf{parallel gadget}.  

\begin{lemma}
As input we are given $k$ lists $L_1, \ldots, L_k$ each containing $g$ strings of length at most $\ell$  (e.g. $s_{i,1}, \ldots, s_{i,g} \in L_i$ and $|s_{i,j}| \leq \ell$) with an alphabet $\Sigma$ and weights in the range $[K]$. Let  
$$M = \sum_{j=1}^{g} WLCS(s_{1,j}, \ldots, s_{k,j}).$$

We can generate a $k$-WLCS instance $P_1, \ldots, P_k$ with $1$ symbol added to the alphabet $\Sigma$, the new range of weights being $[2 \ell K g]$ and length $|P_i|=O(n\ell)$, such that WLCS$(P_1, \ldots, P_k)=C_{par}+M$ for $C_{par} = 2\ell K g (g-1)$. 

The count of \# WLCS$(P_1, \ldots, P_k)$ will be 
$$\Pi_{j=1}^g  \# WLCS(s_{1,j}, \ldots, s_{k,j}).$$
So the multiplication of all the matched $k$ tuple counts. 
\label{lem:parallelGadget}
\end{lemma}
\begin{proof}
We create a new character $\$$ with weight $2 \ell k n$ which is larger than $M$. Create each string $P_i$ as follows
$$P_i=s_{i,1} \$ s_{i,2} \$ \ldots \$ s_{i,g}.$$

Now aligning the $g-1$ symbols $\$$ has such impact it swamps everything else. So the WLCS will force comparisons of the first $k$ tuple ($s_{1,1},\ldots,s_{(k,1)}$), then the next $k$ tuple and so on. Given this, the count of the number of longest weighted subsequences is simply the multiplication of how many ways to achieve the longest subsequence for each of our $g$ $k$-tuples. 
\end{proof}

\paragraph{Building Factored Vector Gadgets}
\begin{lemma}
Let $w_1, \ldots, w_k$ be vectors of length $b = o(\lg(n))$. There are $k$ vector gadgets $VG_1(\cdot), \ldots, VG_k(\cdot)$ such that WLCS$\left(VG_1(w_1), \ldots, VG_k(w_k) \right)$ is some constant $C_{VG}$ if the $k$ vectors $w_1,\ldots, w_k$ are orthogonal and is $C_{VG}-1$ otherwise.  %Additionally the count of \# WLCS$\left(VG_1(w_1), \ldots, VG_k(w_k) \right)$ is $1$ if the WLCS is $C_{VG}$. 

This uses an alphabet of size $2k+2$, weights of size $\tO(1)$ and the length of each $w_i$ is $\tO(1)$.

\label{lem:vecGadget}
\end{lemma}
\begin{proof}
We introduce two symbols $0$ and $1$ where $w(0)=w(1)=1$. For each vector $w_i$ we construct all possible $kb$ length zero-one strings of the following form: $\{0,1\}^{(i-1)b}w_i\{1,0\}^{(k-i)b}$. That is, we generate all possible $bk$ length zero-one strings where the bits from position $b(i-1)+1$ to position $bi$ form $w_i$. Call this set of strings $S'_i(w_i)$. Now we generate the set $S_i(w_i) \subseteq S'_i(w_i)$ by including only strings where the vectors formed by the first $b$ bits, the second $b$ bits, the third $b$ bits, etc form a $k$ tuple of vectors that are $k$-orthogonal. So $S_i(w_i)$ is a representation of all tuples of $k$ vectors of length $b$ where $w_i$ is the $i^{th}$ vector and the $k$ vectors are $k$-orthogonal. 

Now note that the only way that there is one string from each set $S_1(w_1),\ldots,S_k(w_k)$ such that the weighted longest common subsequence of those strings is $kb$ is that if those strings match perfectly. The only way for there to be $k$ perfectly matching strings is if vectors $w_1,\ldots,w_k$ are orthogonal (the string would otherwise be excluded).

So we have generated $k$ lists of at most $2^{b(k-1)}$ strings of  length $bk$. To ensure all the lists are the same length we will pad all the lists to length $2^{b(k-1)}$ with empty strings. 
We can now use the selector gadget (Lemma \ref{lem:selectorGadget}) to wrap around these lists. This will add $k$ new symbols and make the gadgets have length $O(2^{2b(k-1)}+2^{b(k-1)}bk)$ with weights in range $[2bk2^{b(k-1)}]$. We note that this length is $\tO(1)$ and this weight is also $\tO(1)$ because $b$ is constrained to be $b=o(\lg(n))$. Call this construction $VG'_i(\cdot)$ for the $i^{th}$ vector. 

So right now if we have an orthogonal vector $k$ tuple we get a weighted longest common subsequence of weight $y =2k2^{b(k-1)}(2 bk 2^{b(k-1)}) + bk$. But, for some inputs the optimal weight could be much lower (like $2k2^{b(k-1)}(2 bk 2^{b(k-1)})$). 

So, we will add another layer of a selector (from Lemma \ref{lem:selectorGadget}) around $VG'_i(w_i)$ as follows: Our lists will be of length two. The $i^{th}$ list will be $VG'_i(w_i)$ and $0^{y-1}$. So if the vectors aren't orthogonal the second option will lower bound the weight of the longest subsequence. This layer of selector adds another $k$ symbols to our alphabet. It multiplies our weight by $\tO(1)$. Our weight remains $\tO(1)$. 

Now if the vectors are orthogonal we get weight $C_{sel}+y$ and $C_{sel}+y-1$ otherwise where $C_{sel}$ is set by our selector gadget.
\end{proof}

\begin{lemma}
Let $Z_1,\ldots, Z_k$ each be a subset of $\{0,1\}^b$. 
There are $k$ set vector gadgets $SVG_1(\cdot), \ldots, SVG_k(\cdot)$ such that WLCS$\left(SvG_1(Z_1), \ldots, SVG_k(Z_k) \right)$ is some constant $C_{SVG}$ if the $k$ sets of vectors have $\circ(Z_1,\ldots, Z_k)>0$ and is $C_{SVG}-1$ otherwise.  

%Additionally the count of \# WLCS$\left(SVG_1(Z_1), \ldots, SVG_k(Z_k) \right)= \circ(Z_1,\ldots,Z_k)$ if the WLCS is $C_{SVG}$. 

This uses an alphabet of size $3k+2$, weights of size $\tO(1)$ and has a length of $\tO(1)$.
\label{lem:setGadget}
\end{lemma}
\begin{proof}
We first construct vector gadgets $VG_i(w_i)$ of Lemma \ref{lem:vecGadget} for all $w_i\in Z_i$. Let the list $L_i$ consist of all the vector gadgets $VG_i(w_i)$ for all $w_i\in Z_i$. We use these lists to make the selector gadget of Lemma \ref{lem:selectorGadget}.% around a list of vector gadgets. 
%Specifically given a set of at most $2^b$ $b$ length zero one vectors $Z_i$, we input to our selector gadget a list formed by including all $VG_i(w_i)$ for all $w_i \in Z_i$.  

Note that $Z_i$ is a set of at most $2^b$ zero-one vectors of length $b$. So the expense of the selector gadget is polynomial in the length, weight, and number of input strings. All these numbers are $\tO(1)$ so the cost of this selector gadget is $\tO(1)$. 
If there is an orthogonal $k$ tuple within these sets then the optimal weight will be $C_{sel}+C_{VG}$, if there are not then the optimal weight will be $C_{sel}+C_{VG}-1$. 

This adds another set of $k$ symbols for a total of $3k+2$. 
\end{proof}

\begin{lemma}
Let $v_1,\ldots, v_k$ each be a factored vector with $g$ sets containing $b$-bit vectors. 
There are $k$ set vector gadgets $FVG_1(\cdot), \ldots, FVG_k(\cdot)$ such that WLCS$\left(FVG_1(v_1), \ldots, FVG_k(v_k) \right)$ is some constant $C_{FVG}$ if the $k$ sets of vectors have $\circledcirc(v_1,\ldots, v_k)>0$ and is $C_{FVG}-1$ otherwise.  %Additionally the count of \# WLCS$\left(FVG_1(v_1), \ldots, FVG_k(v_k) \right)= \circledcirc(v_1,\ldots,v_k)$ if the WLCS is $C_{FVG}$. 

This gadget uses an alphabet of size $3k+3$ has weights of $\tO(1)$ and has a length of $\tO(1)$.
\label{lem:factoredVecGadget}
\end{lemma}
\begin{proof}
Let $v$ be a factored vector with $g$ sets called $v[1],\ldots, v[g]$.

To build $FVG_i(v)$ we want to concatenate the gadgets $SVG_i(v[1]),\ldots,SVG_i(v[g])$. We will use the parallel gadget for this (Lemma \ref{lem:parallelGadget}).

We have that $\circledcirc(v_1,\ldots, v_k)>0$ only if for all $j\in[1,g]$ we have $\circ(v_1[j], \ldots, v_k[j])>0$. So, we want to know if the sum of all the \kwlcs[]~of all $k$ tuples of string $SVG_1(v_1[j]),\ldots,SVG_k(v_k[j])$ are $C_{SVG}$. If they are all $C_{SVG}$, then $\circledcirc(v_1,\ldots, v_k)>0$.

There are $g$ set vector gadgets each of length $\tO(1)$ and with symbols of weight $\tO(1)$. The parallel gadgets weights and length depend polynomialy on $g$ and the weights and length of the input strings. Notably, all these values are $\tO(1)$ so the length and weight of the $FVG_i(\cdot)$ will both be $\tO(1)$. 

The number of symbols increases by $1$ over the symbols in $SVG_i$. So we have $3k+3$ symbols.
\end{proof}

%%%%%%%%%%%%%%%%%%%%%%%%%%%%%%%%%%%%
\paragraph{WLCS and LCS}

We now give a reduction from \ckov[]~to \kwlcs[]~in the worst case.

\begin{theorem}
A $T(n)$ time algorithm for $k$-WLCS~with alphabet size $O(k)$ and weights in the range $[\tO(1)]$ implies a $\tO(T(n))$ algorithm for \ckov[].
\label{thm:WLCSisWorstCaseHard}
\end{theorem}
\begin{proof}

Let the \ckov[]~instance be given as $k$ lists $V_1,\ldots,V_k$ each containing $n$ factored vectors $v \in V_i$. Every factored vector has $g$ subsets of $\{0,1\}^b$. Recall that $g=o(\lg(n)/\lg\lg(n))$ and $b= o(\lg(n))$. 

We will be reducing this to an instance of \kwlcs[]~where we have $k$ strings $P_1,\ldots,P_k$. These strings will have length $\tO(n)$ and weights that range from $1$ to a number that is $\tO(1)$.   

We will produce our \kwlcs[]~instance by wrapping an alignment gadget around our factored vector gadgets from Lemma \ref{lem:factoredVecGadget}. 
We are going to use the alignment gadget from \cite{LCSisHard} (see the proof of Lemma 14 in that paper) as follows. %We repeat it here for clarity. 

We introduce $k+1$ new symbols: $8,9,3_2, 3_3, \ldots, 3_k$. Let $Q= |P_k|$. For the weights of these symbols, we set  $w(3_i) = B_i$ and we set $B= B_k>D$ where $D$ is the largest possible weight of a factored vector gadget $FVG$ which we defined in Lemma \ref{lem:factoredVecGadget}. The length of a $FVG(v)$ is $\tO(1)$ and the weight of every symbol is $\tO(1)$ so $D$ is $\tO(1)$. We set $B= B_k= (10kD)^2$. We set $B_i = (2k)^{k-i} B$. We set $w(8)=w(9) = 10k^2B_2$. So 
$$w(8)=w(9)\gg w(3_2)\gg w(3_3)\gg \ldots\gg w(3_k)\gg D.$$

We will use parentheses bellow. They do not represent symbols, they are there to assist in readability and to help convey repetitions for example $(\$\#)^3$ means $\$\#\$\#\$\#$. 

Now we produce gadgets to wrap our factored vector gadgets. Let
$$FVG_1'(s) = 8 FVG_1(s) 9 $$
and
$$FVG_i'(s) = 8 FVG_i(s) 9 (3_2 \ldots 3_i)^{Q}.$$

Define the factored vector $\vec{e}$ to be the vector formed by $g$ empty sets (so a vector that gets the worst match possible). We define the concatenation operator. Let $|_{v\in V_i} FVG'_i(v)$ be the concatenation of the $FVG'(\cdot)$ applied to every factored vector $v$ in the input list $V_i$. 

Now we will define the strings $P_i$:
$$P_i =  (3_{i+1}\ldots3_k)^Q (3_2\ldots 3_i)
(FVG_i'(\vec{e}))^{(i-1)n}
\left( |_{v\in V_i} FVG'_i(v) \right)
(FVG_i'(\vec{e}))^{(i-1)n}
(3_{i+1} \ldots 3_k)^Q $$

Given the choices of weights for the symbols $8,9,3_2,\ldots, 3_k$, a weighted longest common subsequence must contain the maximum possible number of each symbol. Given the construction of the alignment gadgets, there are weighted longest common subsequences that contain the maximum possible number of each symbol individually, simultaneously. 
%We now note that these alignment gadgets force the best $k$-WLCS to always contain the maximum number of $8,9,3_2,\ldots, 3_k$ symbols possible (the $k$-WLCS will contain the number of these symbols that appear in the string with the fewest of the symbol). 
For a more formal treatment see Lemma 14 in \cite{LCSisHard}. The length of these strings is $\tO(n)$ and the weights are of size $\tO(1)$. Recall the $FVG(\cdot)$ constructions have length $\tO(1)$ each. 

Our alphabet use for $FVG$ is $O(k)$ symbols and we have added $k+1$ symbols so the total number of symbols is $O(k)$.

Further note that the optimal $k$-WLCS will align exactly $n$ $k$-tuples of $FVG(\cdot)$s. This means the length of the optimal $k$-WLCS will be some constant $C_{tot}$, plus $n(X-1)$ if there are no $k$-tuples $(v_1,\ldots, v_k)$ such that $\circledcirc(v_1,\ldots, v_k)>0$. Otherwise the optimal  $k$-WLCS will be at least $C_{tot}+n(X-1)+1$. This allows us to solve the detection problem for \ckov[]~with one call to \kwlcs[]~on strings of length $\tO(n)$ and weights in the range $\tO(1)$. 
\end{proof}

%\begin{theorem}
%A $T(n)$ time algorithm for \# $k$-WLCS~with alphabet size $O(k)$ and weights in the range $[?]$ implies a $\tO(T(n))$ algorithm for \# %\ckov[].
%\label{thm:WLCSisAvgCaseHard}
%\end{theorem}
%\begin{proof}

%%For every vector $v_i \in V_i$ we will produce its gadget, $FCG(v_i)$, and plug these $k$ lists of $n$ strings each into our parallel gadget from Lemma \ref{lem:parallelGadget}. 

%Now our longest common subsequence is the longest common subsequence of the highest $k$ tuples of these original strings plus a constant. So, we can detect if any $\circledcirc(v_1, \ldots, v_k)>0$ by checking if the length is greater than $C= ???$. 

%As for the count, the count will be the sum over all $k$ tuples of our input strings of their counts of longest weighted common subsequences. If any k tuple has   $\circledcirc(v_1, \ldots, v_k)>0$, then this count is equal to:
%$$\sum_{v_1,\ldots,v_k \in V_1,\ldots,V_k} \circledcirc(v_1, \ldots, v_k).$$
%This is exactly the output of \# \ckov[].

%So, \# \ckov[]~is zero if the \kwlcs[]~is less than $C$. But, if the \kwlcs[]~is $C$ then \# \kwlcs[]~returns the count of  \# \ckov[].
%\end{proof}
\begin{reminder}{Theorem \ref{thm:kLCSisWorstCaseHard}}
A $T(n)$ time algorithm for \klcs~with alphabet size $O(k)$ implies a $\tO(T(n))$ algorithm for \ckov[].
\end{reminder}
\begin{proof}
We use Theorem \ref{thm:WLCSisWorstCaseHard} and Lemma \ref{lem:WLCSandLCS}. The weights of the instance produced in Theorem \ref{thm:WLCSisWorstCaseHard}  are $\tO(1)$ and the length of strings is $\tO(n)$. So we can reduce \ckov[]~to \kwlcs[]~and then reduce that instance of \kwlcs[]~to \klcs. 

%We will use the factored vector gadgets of Lemma \ref{lem:factoredVecGadget}.

%For every vector $v_i \in V_i$ we will produce its gadget, $FVG(v_i)$, and plug these $k$ lists of $n$ strings each into our alignment gadget from Lemma \ref{lem:alignmentGadget}. Recall that $k$ $FVG(v_1),\ldots, FVG(v_k)$ have a longest common subsequence of length $C_{FVG}$ if $\circledcirc(v_1, \ldots, v_k)>0$, and otherwise have length $C_{FVG}-1$

%Now our longest common subsequence is a constant, $M_{ali}$, plus the longest common subsequence of $n$ aligned $k$ tuples of factored vector gadgets. If there is no $k$ tuple of factored vectors such that $\circledcirc(v_1, \ldots, v_k)>0$, then  the length is $n(C_{FVG}-1)+C$. if there is a $k$ tuple of factored vector gadgets such that $\circledcirc(v_1, \ldots, v_k)>0$ then the longest common subsequence can align the associated vector gadgets and result in  the total longest common subsequence of length at least $C_{FVG} + (n-1)(C_{FVG}-1)+C = n(C_{FVG}-1)+C_{ali}+1$. So, determining the 

%So, we can detect if any $\circledcirc(v_1, \ldots, v_k)>0$ by checking if the length is greater than $M = n(C_{FVG}-1)+M_{ali}$ then we solve \ckov[]. 
\end{proof}

%%%%%%Edit distance text below here if there is a reduction
%LCS is a restricted version of edit distance where no substitutions are allowed. 
\paragraph{Edit Distance} 

We will use the following Lemma to obtain hardness for Edit Distance.
\begin{lemma}[Restated from Theorem C.2 from \cite{EditDistLCS}]
An algorithm for WLCS (\kwlcs~where $k=2$) that runs in $O(n^{2-\epsilon})$ time for some constant $\epsilon>0$ implies a $O(n^{2-\delta})$ time algorithm for Edit Distance for some constant $\delta>0$. 
\cite{EditDistLCS}
\label{lem:editDistandLCS}
\end{lemma}

Thus Theorem \ref{thm:WLCSisWorstCaseHard} and Lemma \ref{lem:editDistandLCS} give us the following theorem.

\begin{reminder}{Theorem \ref{thm:editDistisWorstCaseHard}}
A $T(n)$ time algorithm for Edit Distance implies a $\tO(T(n))$ algorithm for \ckov[2].
\end{reminder}
%\begin{proof}
%Use Theorem \ref{thm:WLCSisWorstCaseHard} and Lemma %\ref{lem:editDistandLCS}. 
%\end{proof}

\paragraph{\ckovh[]~and LCS and Edit Distance}

Theorem \ref{thm:kLCSisWorstCaseHard} and  Theorem \ref{thm:editDistisWorstCaseHard} give us the following corollary.
\begin{corollary}
If \ckovh[]~is true then \klcs~requires $n^{k-o(1)}$ time. 
If \ckovh[2]~is true then Edit Distance requires $n^{2-o(1)}$ time. 
\end{corollary}
%\begin{proof}
%This follows from Theorem \ref{thm:kLCSisWorstCaseHard} and  Theorem \ref{thm:editDistisWorstCaseHard}.
%\end{proof}

%removed
%\section{Framework Optimality}
%\label{sec:frameworkOptimal}
%\input{OptimalFramework.tex}

%%%% The inclusion edgesclusion thing
%\section{Sub-graphs with k-Nodes in $H$-Partite Graphs}
\section{Average Case Hardness for Subgraph Counting}
\label{sec:SubgraphsWithKNodes}

Here we demonstrate the power of the framework in Section \ref{sec:Framework} to show average case hardness for counting subgraphs $H$ with $k$ vertices, where $k=o(\sqrt{\lg(n)/\lg\lg(n)})$. % is always small enough. 
 If the sub-graph $H$ is sufficiently sparse then some larger $k$ can be tolerated. Notably, for this section, as long as the number of edges is $e = o(\lg(n)/\lg\lg(n))$ then our worst case to average case reduction has sub-polynomial overhead. 

Using the framework we can immediately show that counting subgraphs $H$ in what are roughly $H$-partite Erd{\H{o}}s-R{\'{e}}nyi graphs (see Definition \ref{def:HpartiteRandomGraph}) is hard.
We use our Inclusion/Edgesculsion Lemma from Section \ref{sec:edgesclusion} to extend this result to counting subgraphs $H$ in Erd{\H{o}}s-R{\'{e}}nyi graphs, and show that this problem is average case hard as well. We start by a few definitions.
%, and showing that  in time $T(n)$ implies that we can count subgraphs $H$ in worst-case graphs in time $\Tilde{O}(T(n)+m)$. 

%\newcommand{\chpg}{CHGKP}
%\begin{definition} The \textbf{counting $H$ sub-graphs in a $k$-partite fashion (\chpg)} problem takes as input a $k$-node graph $H$ and a $n$-node graph $G$ with a vertex set partition $V_1, \ldots, V_k$, and asks for the count of the number of sub-graphs of $G$ that have exactly one node from each of the $k$ partitions and contain the graph $H$.
 %where you are given as input:
%a graph $G$ with $n$ nodes and $m$ edges.  
%You are also given the partitioning of the vertex set of $G$: $V_1, \ldots, V_k$.  To solve \chpg~one must return the count of the number of sub-graphs of $G$ that have exactly one node from each of the $k$ partitions and contain the graph $H$.
%\end{definition}

\begin{definition}
The \textbf{counting $H$ sub-graphs in a $H$-partite fashion (\chpgh)} problem takes as input a $k$-node graph $H$ and a $H$-partite $n$-node graph $G$ with vertex set partition $V_1, \ldots, V_k$, and asks for the count of the number of sub-graphs of $G$ that have exactly one node from each of the $k$ partitions and contain the graph $H$.
%Let $H$ be a graph with $k$ nodes and at most $\binom{k}{2}$ edges.
%Let counting $H$ sub-graphs in a $H$-partite fashion (\chpgh) be the problem where you are given as input:
%a graph $G$ with $n$ nodes and $m$ edges.  
%You are also given the partitioning of the vertex set of $G$: $V_1, \ldots, V_k$.  
%The graph $G$ is $H$-partite. 
%To solve \chpgh~one must return the count of the number of sub-graphs of $G$ that have exactly one node from each of the $k$ partitions and contain the graph $H$.
\label{def:HpartiteRandomGraph}
\end{definition}

\newcommand{\uchpg}{UCHGHP}
\newcommand{\uchpgh}{UCHGHP}
\begin{definition}
The \textbf{uniform counting $H$ sub-graphs in a $H$-partite fashion (\uchpgh)} problem takes as input a $k$-node graph $H$ and an $H$-partite $n$-node graph $G$ with vertex set partition $V_1, \ldots, V_k$, where every edge between partitions that have edges in $H$ is chosen to exist iid with probability $\mu$. The problem asks for the count of the number of sub-graphs of $G$ that have exactly one node from each of the $k$ partitions and contain the graph $H$.
%return the count of the number of sub-graphs of $G$ that have exactly one node from each of the $k$ partitions and contain the graph $H$.
%The uniform counting $H$ sub-graphs in a $k$-partite fashion problem (\uchpgh) is given as input a $H$-partite graph $G$ with $n$ nodes in each partition $V_1, \ldots, V_k$. Where every edge between partitions that have edges in $H$ is chosen to exist iid with probability $\mu$.
%To solve \uchpgh~one must return the count of the number of sub-graphs of $G$ that have exactly one node from each of the $k$ partitions and contain the graph $H$.
\end{definition}

Note that \chpgh~is a worst-case problem whereas \uchpg~is an average-case problem. Notably, \uchpg~is the uniform distribution over inputs to \chpgh. 

\subsection{Reducing counting \texorpdfstring{$H$}{Lg} subgraphs in \texorpdfstring{$H$-partite}{Lg} fashion to uniform counting}

We start by reducing \chpgh~to \uchpgh. Our ultimate goal is to reduce \chpgh~to counting $H$ subgraphs in an Erd{\H{o}}s-R{\'{e}}nyi graph. %But, we will first start by reducing \chpgh~to counting $H$ subgraphs in $k$-partite Erd{\H{o}}s-R{\'{e}}nyi graphs. 

%\begin{lemma}
%If \chpgh~runs in time $T(n)$ then detecting an $H$ sub-graph in a graph $G$ runs in time $\tO(T(n)+n^2)$.
%\end{lemma}
%\begin{proof}
%We use the color coding approach of Alon et al \cite{colorcoding}.% to partition the vertex set of $G$. 
%The color coding approach will succeed deterministically by making $\lg(n)^{k^2}$ calls to \chpg. By making $k!$ calls to \chpgh~we can solve \chpg. We simply call \chpgh~for every permutation of the vertices in $H$. If permutations of the vertices of $H$ send it onto itself (e.g. in a clique all $k!$ permutations are the same) then we will divide out by this over counting. 
%\end{proof}

\begin{lemma}
Let $H$ be a $k$-node graph with vertices $V_H=\{x_1,\ldots,x_k\}$ and $G$ a $H$-partite $n$-node graph with vertex set partition $V_1,\ldots,V_k$. 
Let $\vec{E}$ be the set of variables $\{e(v_i,v_j)| i\neq j, v_i\in V_i, v_j\ \in V_j \}$ when an edge variable is a $1$ if that edge exists and $0$ if the edge is absent in $G$. 
Let $h(v_1, \ldots ,v_k)$ be a function that multiples $e(v_i,v_j)$ if $x_ix_j$ is an edge in $H$ for all $i,j \in [1,k]$ where $i \ne j$. 
%Let $g(v_1, \ldots ,v_k) = \sum_{\pi \in Perm[\{v_1,\ldots, v_k\}]}h(\pi[1], \ldots ,\pi[k])$. That is, $g$ is the sum over the function $h$ called on all $k!$ possible permutations of the variables $v_1,\ldots,v_k$. So if the subgraph of $G$ formed by the nodes $v_1, \ldots ,v_k$ has $c$ subgraphs $H$, then $g(v_1, \ldots ,v_k)=c$.
If $p$ is a prime in $[2n^k, n^{2k}]$, 
the following function returns the output of \chpgh~on $G$:
$$f(\vec{E}) = \sum_{v_1 \in V_1,\ldots,v_k \in V_k} h(v_1, \ldots ,v_k) \pmod{p}.$$
\label{lem:smallKSugraphPoly}
\end{lemma}
\begin{proof}

Consider the function $h$: If $v_1, \ldots ,v_k$ in that particular order contain the graph $H$ it returns $1$, otherwise it returns $0$. Specifically, we are checking if our particular permutation of these variables completely covers the (arbitrary) permutation of variables associated with the input sub-graph $H$. 

Now $f$ sums over all choices of $k$ nodes from each partition and counts how many instances of the sub graph appear in each. There is no double counting because every set of $k$ nodes differs by at least one node. 
\end{proof}

\begin{lemma}
The function $f$ defined in Lemma \ref{lem:smallKSugraphPoly} is a \goodPoly~for \chpgh~if the number of edges in $H$ is $o(\lg(n)/\lg\lg(n))$. 
\label{lem:gpolForSmallK}
\end{lemma}
\begin{proof}
To prove the lemma, first note that $f$ is a polynomial over a prime finite field $F_{p}$ for some prime $p \in [2n^k, n^{2k}]$,
and the number of monomials in $f$ is $O(n^k \cdot k!)$, which is polynomial. By Lemma \ref{lem:smallKSugraphPoly} the function $f$ returns the same value as \chpgh~ when it is given zero-one inputs. 
%Now if $\vec{I} = b_1,\ldots, b_n$ with $b_i\in\{0,1\}$ in the prime field, then
%$f(b_1,\ldots, b_n) = f(\vec{I})=\chpgh(\vec{I})$ by Lemma  \ref{lem:smallKSugraphPoly}. %This is because every variable in $f$ represents an edge and if every value is a zero or one the function $f$ and \chpgh~return the same value by Lemma \ref{lem:smallKSugraphPoly}.

%Let $|E_H|$ be the number of edges in $H$, note that $|E_H|\leq \binom{k}{2}$.
Let $|E_H|$ be the number of edges in $H$. The function $f$ has degree $|E_H| = o\left( \lg(n)/\lg\lg(n) \right)$. In fact given constant $k$, $f$ has constant degree. This is because $f$ is formed with a sum over monomials $h(v_1,\ldots,v_k)$, which have degree $|E_H|\le \binom{k}{2}$.% which is at most $\binom{k}{2}$. 

Finally, the function $f$ is strongly $|E_H|$-partite. 
There are $|E_H|$ partitions of edges. The function $f$ is a sum over calls to $h$ where $h$ takes as input one variable from each of those edge partitions and multiplies all of them. %So $f$ is strongly $|E_H|$ partite. 
\end{proof}

\begin{corollary}
Let $d=\binom{k}{2}$ and $k=o(\sqrt{\lg(n)/\lg\lg(n)})$.
 If an algorithm exists to solve \uchpgh~in time $T(n)$ with probability $1-1/\omega\left( \lg^{d}(n)\lg\lg^d(n) \right)$, then an algorithm exists to solve \chpgh~in time $\tO(T(n) +n^2)$ with probability at least $1-O\left(2^{-\lg^2(n)} \right)$.
\end{corollary}
\begin{proof}
If $k=o(\sqrt{\lg(n)/\lg\lg(n)})$ then the size of the edge set in $H$, $E_H$ is at most $\binom{k}{2}=d = o(\lg(n)/\lg\lg(n))$.
Using Theorem \ref{thm:framework} we simply need that a \goodPoly~for \chpgh~exists.
By Lemma \ref{lem:gpolForSmallK}, the function $f$ from Lemma \ref{lem:smallKSugraphPoly} is a \gPol{\chpgh}. 
\end{proof}

\begin{corollary}
Let $H$ be a sub-graph with an edge set $E_H$ where $|E_H| = o(\lg(n)/\lg\lg(n))$. Let $d=|E_H|$.
 If an algorithm exists to solve \uchpgh~in time $T(n)$ with probability $1-1/\omega\left( \lg^{d}(n)\lg\lg^d(n) \right)$, then an algorithm exists to solve \chpgh~in time $\tO(T(n) +n^2)$ with probability at least $1-O\left(2^{-\lg^2(n)} \right)$.
 \label{cor:uchpghImpliescpgh}
\end{corollary}
\begin{proof}
We have that $E_H$ is at most $ \binom{k}{2}=d = o(\lg(n)/\lg\lg(n))$.
Using Theorem \ref{thm:framework} we simply need that a \goodPoly~for \chpgh~exists.
By Lemma \ref{lem:gpolForSmallK}, the function $f$ from Lemma \ref{lem:smallKSugraphPoly} is a \gPol{\chpgh}. 
\end{proof}

\subsection{Inclusion-Edgesculsion}
\label{sec:edgesclusion}
In Corollary \ref{cor:uchpghImpliescpgh} we show that counting subgraphs $H$ in Erd{\H{o}}s-R{\'{e}}nyi $H$-partite graphs quickly with a high enough probability implies fast algorithms for counting $H$-subgraphs in the worst case. 
We now want to extend this to fully Erd{\H{o}}s-R{\'{e}}nyi graphs. Specifically, we want to show that counting $H$-subgraphs in Erd{\H{o}}s-R{\'{e}}nyi quickly with a high enough probability implies a fast algorithm for counting $H$-subgraphs in the worst case. 
To acheive this goal we introduce our Inclusion-Edgesclusion technique.  
We begin with a few definitions. 

\begin{definition}
Let $G$ be a $k$-partite Erd{\H{o}}s-R{\'{e}}nyi graph with every edge included with probability $1/b$ where $b$ is a constant integer. Let the vertex partitions of $G$ be $V_1,\ldots,V_k$ and the edge partitions be $E_{i,j}$ $\forall i,j \in [1,k]$ where $i<j$.

Label all $|V_i|\cdot |V_j|$ edges with numbers in $[1,b]$ as follows. Edges that exist in $G$ are labeled $1$. The rest of the edges are uniformly at random assigned labels from $[2,b]$. 
For $\ell \in [1,b]$, let $E_{i,j}^{\ell}$ be the set of all edges of label $\ell$.% where the label $\ell \in [1,b]$ was chosen for that edge. 

Let $G^{(\ell_1)(\ell_2)\ldots(\ell_{\binom{k}{2}})}$ be the graph formed by choosing edge sets $E_{1,2}^{\ell_1}$, $E_{1,3}^{\ell_2},\ldots, E_{k-1,k}^{\ell_{\binom{k}{2}}}$.
Let $S_G$ be the set of all possible $b^{\binom{k}{2}}$ graphs $G^{(\ell_1)(\ell_2)\ldots(\ell_{\binom{k}{2}})}$.
\label{def:graphInverses}
\end{definition}
Note when $b=2$ these sets of edges are $E_{i,j}^{(1)} = E_{i,j}$ and $E_{i,j}^{(2)} = \bar{E}_{i,j}$.

\begin{definition}
Let $G$ be a $k$-partite Erd{\H{o}}s-R{\'{e}}nyi graph with every edge included with probability $1/b$ where $b$ is a constant integer. Let the vertex partitions be $V_1,\ldots,V_k$. Let the edge partitions be $E_{i,j}$ $\forall i,j \in [1,k]$ where $i<j$.

Let a labeled subgraph $L$ of $H$ in $G$ be a subgraph of $H$ where every vertex is assigned a unique label from $[1,k]$. 

Define the count of the number of labeled subgraphs $L$ in $G$ to be the number of not-necessarily induced subgraphs $L$ where every vertex in $L$ with label $\ell$ comes from $V_{\ell}$ in the original graph. 
\label{def:labeledSubgraphs}
\end{definition}

We want to reduce \uchpgh~to counting subgraphs $H$ in Erd{\H{o}}s-R{\'{e}}nyi graphs. A uniformly random $H$-partite graph only has edges between partitions corresponding to edges in $H$. However, an Erd{\H{o}}s-R{\'{e}}nyi graph would have edges within partitions and between partitions that don't correspond to edges in $H$. So, if we add these random edges we will over count subgraphs $H$, including subgraphs $H$ that appear outside of the original $H$-partite graph. 

We solve this problem by creating multiple graphs. Each graph individually looks like it is sampled from the Erd{\H{o}}s-R{\'{e}}nyi distribution. However, these graphs are correlated. We use a variant of an inclusion-exclusion argument (hence the name ``inclusion-edgesclusion'') to count the subgraphs $H$ that appear in the original $H$-partite graph. 

We will start with a warm up lemma.

\begin{lemma}[Warm Up Lemma]
Let $C_G$ be the count of the number of $k$-node subgraphs $H$ in a complete $k$-partite graph with the same edge partitioning as $G$ where exactly one node of the subgraph is in each partition in $G$.

Let $C_{S_G}$ be the sum of the subgraphs $H$ in all graphs $G^{(\ell_1)(\ell_2)\ldots(\ell_{\binom{k}{2}})}$ in $S_G$ where each of the partitions of $G^{(\ell_1)(\ell_2)\ldots(\ell_{\binom{k}{2}})}$ has exactly one vertex of the subgraph. 

Then, $C_G = C_{S_G}$.
\label{lem:warmUp}
\end{lemma}
\begin{proof}
If a subgraph $H_0$ exists and has one vertex in each partition, then there is exactly one choice of $G^{(\ell_1)(\ell_2)\ldots(\ell_{\binom{k}{2}})} \in S_G$ that will contain it. The choice of $G^{(\ell_1)(\ell_2)\ldots(\ell_{\binom{k}{2}})}$ that picks the edge sets that $H_0$'s edges lay in. Every $H$ that exists in the complete graph will appear in exactly one of these $G^{(\ell_1)(\ell_2)\ldots(\ell_{\binom{k}{2}})}$, so the counts of both are the same. 
\end{proof}

What should you get out of this lemma intuitively? Consider what happens if we sum all $H$ that involve exactly one edge from $E_{i,j}^{1}, E_{i,j}^{2},\ldots$ or$, E_{i,j}^{b}$ (as defined in Definition \ref{def:graphInverses}). Then, we are getting the sum of all $H$ that would exist if $E_{i,j}$ were complete. 
We can use this idea to count the subgraphs that use \emph{particular} edge partitions, while every $E_{i,j}^{(\ell)}$ looks uniformly random. %And thus  $G^{(\ell_1)(\ell_2)\ldots(\ell_{\binom{k}{2}})}$ 
To do this count, we develop a few lemmas and then we proceed to our main counting result in Lemma \ref{lem:inductiveStep1}.

\paragraph{Counting Small Subgraphs}
We will argue that we can count labeled subgraphs $H$ recursively. We start by arguing the base cases. Below are give fast algorithms for counting small labled subgraphs. By counting labeled subgraphs $H$ in a graph $G$ with partitions $V_1,\ldots,V_k$, we mean that the vertex set of $H$ is labeled with $1,\ldots,k$, and we want every copy of $H$ to have a copy of $x_i$ in $V_i$ where $x_i$ is the vertex with label $i$ in $H$.

\begin{lemma}
Let $G$ be a graph with $n$ nodes, $m$ edges and $k$ labeled partitions of the vertices $V_1, \ldots, V_k$ ($G$ is not necessarily $k$ partite). 

Given a labeled $k$-node tree $H$ with vertices, counting the number of such labeled trees in $G$ takes $O(m+n)$ time. 
\label{lem:treesEasy}
\end{lemma}
\begin{proof}
Pick a root of the tree $H$. Let $(u_0,p_0)$ be the root and its label $p_0$. Let $U_i$ be the set of all tuples of vertices and their labels in the tree at level $i$. 
Let $h$ be the height of the tree. 

Thus, the set $U_h$ only contains leaves, and every node in $U_h$ has one sub-tree that includes it and no nodes below it. 

For all $(u_{h-1},p_{h-1}) \in U_{h-1}$, where $p_{h-1}$ is the label of $u_{h-1}$, we look at the vertex set $V_{p_{h-1}}$. For all nodes in $V_{p_{h-1}}$ we are going to count the number of labeled sub-trees that include it and the nodes below it. We can do this in linear time over the edges between the relevant partitions. Save all the computed values. 

Now, we can do this for level $h-2$, using our pre-existing counts. We can propagate these up the tree until we reach our root and count the total number of labeled trees $H$ in the graph. 
\end{proof}

%\paragraph{The Recursive Step of Inclusion-Edgesclusion}

\begin{lemma}
Let $G$ be a graph with $n$ nodes, $m$ edges and $k$ labeled partitions of the vertices $V_1, \ldots, V_k$ ($G$ is not necessarily $k$ partite). 

If we have the counts of all labeled subgraphs of $H$ in $G$ of size less than $s$ vertices, we can compute the number of labeled subgraphs in $G$ that are the union of two disconnected labeled subgraphs of $H$ of size $s$ or less.
\label{lem:disconnectedSubgraphsEasy}
\end{lemma}
\begin{proof}
Let one be labeled subgraph $L$ and the other be labeled subgraph $L'$. Given that they share no vertices, we can simply multiply the number of subgraphs $L$ and $L'$. 
\end{proof}

\begin{lemma}
Let $G$ be a graph with $n$ nodes, $m$ edges and $k$ labeled partitions of the vertices $V_1, \ldots, V_k$ ($G$ is not necessarily $k$ partite). 

We can compute all counts of subgraphs in $G$ with $2$ vertices or fewer in $\tilde{O}(m+n)$ time.
\label{lem:baseCase}
\end{lemma}
\begin{proof}
All subgraphs with $1$ edge are trees. So by Lemma \ref{lem:treesEasy} we can compute all subgraphs with $2$ edges or fewer in $\tilde{O}(m)$ time. 
\end{proof}

\paragraph{The Recursive Step of Inclusion-Edgesclusion}

This next lemma is the core step. We will use all counts of subgraphs with a small number of edges to count those with more edges. At its core this relies on the fact that if we sum together the counts of the number of subgraphs $H$ with all possible combinations of complimentary edge sets this roughly gives us a count of the number of subgraphs when that edge partition is a complete bipartite edge set. 

\begin{lemma}
Let $G$ be a labeled $k$-partite graph with $n$ nodes per partition. 

Say we are given the counts of the number of subgraphs $H$ in all graphs $S_G$ (see Definition \ref{def:graphInverses}). 

Additionally, say we are given the counts of all less than or equal to $v$ vertex labeled subgraphs of $H$ with $[0,e]$ edges.

Let $L$ be a labeled subgraph of $H$ with $v$ vertices and $e+1$ edges. 

Using both of these counts we can count the number of not-necessarily induced subgraphs $L$ in $G$ in time $O(k! \cdot 2^{k^2} + b^{k^2})$.
\label{lem:inductiveStep1}
\end{lemma}
\begin{proof}
Let $H$ have $v_H=k$ vertices and $e_H$ edges. Let the subgraph $L$ be given as a list of $v$ vertices labeled as being in partitions  $i_1,\ldots,i_v$ and $e+1$ edges between partitions $i_x$ and $i_y$ where $x,y \in [1,v]$. Let $S_E$ be the set of all such pairs $(x,y)$.  

Consider $\bar{S}_E$, the set of all pairs of partitions not in $S_E$. Then consider the subset of instances in $S_G$ where the edges between partitions in $S_E$ (for example $E_{i_1,i_2}$) are all set to be the version labeled $(1)$ ($E_{i_1,i_2}^{(1)}$).  Call this subset $S_G[L]$. 

Take the counts of the number of subgraphs $H$ that appear in all graphs in $S_G[L]$ and sum them together, call this count $c_{S_G[L]}$. What will this count contain? 
It will count the number of subgraphs $H$ that appear if the graph $G$ were to have complete bipartite graphs between all pairs of partitions in $\bar{S}_E$, weighted by how many edges in $S_E$ that subgraph uses. 
If a specific subgraph $H$ appears in the graph $G$ where $\ell$ of its edges are in the $\bar{S}_E$ partitions then it is counted $b^{\binom{k}{2}-e-1-\ell}$ times. We include that many copies of graphs in $S_G[L]$ that include this particular $H$. 

Given that $L$ is a labeled subgraph of $H$, at least one labeling of $H$ will share all $e+1$ edges and $v$ vertices of $L$. There may be many valid labelings for the  $e_H-e-1$ unaccounted for edges and $k-v$ unaccounted for vertices. 

We want to count all $H$ that happen to have a labeling that matches the $e+1$ edges of $L$, and not count those that share only some of these edges. Luckily, given the counts of all small subgraphs we can count how many subgraphs $H$ exist that match up only partially with $L$ and remove these from the count $c_{S_G[L]}$.

For a subgraph to match up only partially with $L$, it must match up with some labeled subgraph of $L$, $L'$. $L'$ must have $v$ vertices and at most $e$ edges. We have the counts of all labeled subgraphs with  $v$ vertices and at most $e$ edges. We want to remove from $c_{S_G[L]}$ the count of all subgraphs $H$ that overlap with $L'$ and share no edges with $L - L'$. 

Let $G_{L,L'}$ be a graph on $k$ vertices where all edges in $L'$ are included, all edges in $L-L'$ are excluded and all other edges are included. Let $c_{G_{L,L'}}$ be the count of the number of subgraphs $H$ that exist in this graph. Note we can compute this in $O(k!)$ and we do this computation on at most $O(2^{k^2})$ graphs. 

Let $L'$ have $e_{L'}$ edges and $v_{L'}$ vertices. Let $c_{L'}$ be the count of all labeled subgraphs $L'$ that exist in $G$. 
The count of all subgraphs $H$ which overlap exactly with $L'$ (sharing no edges with $L-L'$) that are counted in $c_{S_G[L]}$ is 
$$c_{L'} \cdot c_{G_{L,L'}} \cdot n^{k-v_{L'}} \cdot b^{\binom{k}{2}-e-e_H+e_{L'}}.$$
Lets break down this value. First, of course the number of labeled subgraphs $L'$ that appear in the original graph each contribute proportionally. A choice of a particular labeled subgraph $L'$ fixes $v_{L'}$ of the $k$ vertices, but the rest of the vertices could be any of the available $n$ vertices per partition. Now, given a fixed choice of $k$ vertices and $e_H$ edges this subgraph may still appear in multiple graphs in $S_{G}[L]$. Specifically, it will appear in all graphs where we haven't ``fixed'' the edge set. This is a total of $b^{\binom{k}{2}-e-e_H+e_{L'}}$ graphs. 

So, for all $O(2^{k^2})$ labeled subgraphs of $L$ we can compute their contribution to $c_{S_G[L]}$ and subtract out this contribution. This leaves only a count of subgraphs $H$ that overlap with $L$ exactly. 
To compute the number of subgraphs $L$ we simply divide this number by $c_{G_{L,L}} \cdot n^{k-v} \cdot b^{\binom{k}{2}-e_H}$.

The total time for this computation is, at most  $O(2^{k^2}\cdot k! + b^{k^2})$. If $k = o(\sqrt{\lg(n)})$ and $b$ is a constant, then this term is sub-polynomial.
\end{proof}

\begin{lemma}
Let $G$ be a graph with $n$ nodes, $m$ edges and $k$ labeled partitions of vertices $V_1,\ldots,V_k$. Given the count of all labeled subgraphs of $H$ in $G$ with less than $v$ vertices, we can count all labeled sub-graphs with $v$ vertices and at most $v-1$ edges in $\tilde{O}(m)$ time. 
\label{lem:inductiveStep2}
\end{lemma}
\begin{proof}
There are two cases. The subgraph is connected (only possible when we have exactly $v-1$ edges), or it is disconnected.

If the subgraph is connected then it is a tree, by Lemma \ref{lem:treesEasy}~we have can count this labeled tree in $\tilde{O}(m)$ time. 

If the subgraph is disconnected then it is made up of disconnected labeled subgraphs with less than $v$ vertices. We have the count of each of these on their own, thus by repeated applications of Lemma \ref{lem:disconnectedSubgraphsEasy} we can count these with overhead the number of subgraphs which is at most $v$, and thus also $\tilde{O}(m)$.
\end{proof}

\paragraph{Reducing to \uchpgh}

First we reduce  counting  labeled copies of $H$ in a $k$-partite Erd{\H{o}}s-R{\'{e}}nyi graph to counting $H$ in Erd{\H{o}}s-R{\'{e}}nyi graphs. We then note that by picking a particular labeling this solves the problem of \uchpgh. Finally, we use our previous reduction from \chpgh~to \uchpgh~to get our desired result: a reduction from \chpgh~to counting $H$ subgraphs in Erd{\H{o}}s-R{\'{e}}nyi graphs. 

\begin{lemma}
%Let $G$ be a labeled $k$-partite Erd{\H{o}}s-R{\'{e}}nyi graph with $n$ nodes per partition. 
%Let $1/b$ be the edge probability. 
Let $H$ have $e$ edges and $k$ vertices.
Let $A$ be an average-case algorithm for counting ``\emph{unlabeled}" subgraphs $H$ in $k$-partite Erd{\H{o}}s-R{\'{e}}nyi graphs with edge probability $1/b$ which takes $T(n)$ time with probability $1-\epsilon/\left(2^k \cdot b^{k^2}\right)$. 

The number of ``labeled" copies of subgraph $H$ in $k$-partite Erd{\H{o}}s-R{\'{e}}nyi graphs with edge probability $1/b$ can be computed in time $\tilde{O}(2^{k^2} \cdot m+2^k \cdot b^{k^2} \cdot T(n))$ with probability at least $1-\epsilon$. 
\label{lem:ERtoKpartite}
\end{lemma}
\begin{proof}
We want to count only subgraphs that use exactly one vertex from each partition. We can make $2^k$ calls to $A$ using standard inclusion/exclusion to count only subgraphs with exactly one edge in each partition. Call this algorithm $A'$.

Let $C(v,\ell)$ be a list of tuples of all labeled subgraphs $J$ with $v$ vertices and $\ell$ edges with the associated count of the number of labeled subgraphs $J$ in $G$. 

By Lemma \ref{lem:inductiveStep2} we can compute $C(v,\ell)$ in time $|C(v,\ell)|\cdot \Tilde{O}(m)$ if $\ell \leq v-1$. 

By Lemma \ref{lem:inductiveStep1} if we can compute $C(v,\ell)$ for all $\ell \leq \ell^{\star}$ then we can compute $C(v,\ell^{\star}+1)$ given calls to $A'$ on all graphs in $S_G$. Note each of these steps uses the \emph{same} set of calls to  $A'$ on all graphs in $S_G$.

We can bound $|S_G|\leq b^{k^2}$. With this we can say that we make at most $b^{k^2}$ calls to $A'$, meaning we make at most $2^k \cdot b^{k^2}$ calls to $A$. \\
We can bound the total sum of all $|C(v,\ell)|$ by $2^{k^2}$ (every possible choice of a subset of edges in the complete graph on $k$ vertices).
This gives a time bound of $\tilde{O}(2^{k^2} \cdot m+2^k \cdot b^{k^2} \cdot T(n))$.

We make $2^k \cdot b^{k^2}$ calls to $A$, if they are all correct then we give the correct answer to the labeled $H$ question. If $A$ succeeds with probability at least $1-\epsilon/\left(2^k \cdot b^{k^2}\right)$, then, by the union bound  $2^k \cdot b^{k^2}$ calls to $A$ will all succeed with probability at least $1-\epsilon$.
\end{proof}

\begin{lemma}
%Let $G$ be a $n$ node Erd{\H{o}}s-R{\'{e}}nyi graph. 
%Let $1/b$ be the edge probability where $b$ is a constant. 
Let $H$ have $e$ edges and $k$ vertices where $k =o(\lg(n)/\lg\lg(n))$.
Let $A$ be an average-case algorithm for counting subgraphs $H$ in Erd{\H{o}}s-R{\'{e}}nyi graphs with edge probability $1/b$ which takes $T(n)$ time with probability $1-2^{-2k} \cdot b^{-k^2} \cdot (\lg(e)\lg\lg(e))^{-\omega(1)}$ 

Then an algorithm exists to count subgraphs $H$ in $H$-partite graphs (\chpgh)~in time $\tilde{O}(T(n))$ with probability at least $1-O(2^{-\lg^2(n)})$. 
\label{lem:WChp-ACer}
\end{lemma}
\begin{proof}
By Lemma \ref{lem:ERtoKpartite}, $A$ implies a $\tilde{O}( T(n))$ algorithm for counting the number of labeled copies of subgraph $H$ in $k$-partite Erd{\H{o}}s-R{\'{e}}nyi graphs with edge probability $1/b$ with probability $1- 2^{-k}(\lg(e)\lg\lg(e))^{-\omega(1)}$.

We need to add random edges within each partition to get a truly Erd{\H{o}}s-R{\'{e}}nyi graph. Luckily, we can use traditional inclusion-exclusion to count how many subgraphs don't include exactly one vertex in each partition. This introduces another $2^k$ calls.  By the union bound this causes the probability of success to be at least $1- (\lg(e)\lg\lg(e))^{-\omega(1)}$.

Now note that counting labeled copies of subgraph $H$ in $k$-partite Erd{\H{o}}s-R{\'{e}}nyi graphs solves \uchpgh~with edge probability $1/b$ with a single call. Given an instance of \uchpgh~label the vertices of the subgraph $H$ in the input instance, between all other partitions add random edges with probability $1/b$.

Now apply Lemma \ref{cor:uchpghImpliescpgh}. An algorithm for  \uchpgh~that succeeds with probability $1- (\lg(e)\lg\lg(e))^{-\omega(1)}$ in time $T(n)$ implies an algorithm for \chpgh~that runs in time $\tO(T(n)+n^2)$ and succeeds with probability $1-O(2^{-\lg^2(n)})$.
\end{proof}

\begin{reminder}{Theorem \ref{thm:ACSubgraphCountToWCSubgraphCount}}
Let $H$ have $e$ edges and $k$ vertices where $k =o(\sqrt{\lg(n)})$. 
Let $A$ be an average-case algorithm for counting subgraphs $H$ in Erd{\H{o}}s-R{\'{e}}nyi graphs with edge probability $1/b$ which takes $T(n)$ time with probability $1-2^{-2k} \cdot b^{-k^2} \cdot (\lg(e)\lg\lg(e))^{-\omega(1)}$.

Then an algorithm exists to count subgraphs $H$ in $k$-partite graphs in time $\tilde{O}(T(n))$ with probability at least $1-\tO(2^{-\lg^2(n)})$. 
\end{reminder}
\begin{proof}
From Lemma \ref{lem:WChp-ACer} we know that $A$ implies a $\tilde{O}(T(n))$ time algorithm for counting $H$ in $H$-partite graphs. 

Now, given an input of a graph $G$ that is $k$-partite graph we can produce all $e$ choose $\binom{k}{2}$ graphs that have only $e$ sets of edges between the partitions. From these we can select only those that are $H$ partite (the number of these will vary based on $H$'s shape). The number of these graphs is at most $2^{k^2}$, which by our restriction on $k$ is $\tO(1)$. Call the set of these $H$ partite graphs $S_H$.

We use the result from Lemma \ref{lem:WChp-ACer} to count the results on each of these graphs. By the union bound we will get the correct answer on every graph with probability at least $1-\tO(2^{-\lg^2(n)})$. The sum these counts over all $G' \in S_H$ is equal to the number of $H$ in the original graph. 
\end{proof}

%%%% This stuff is bad
\section{Counting OV is Easy on Average}
\label{sec:OV}

Previous work has shown that detecting if there is at least one orthogonal vector in a set of $n$ vectors is possible in sub-quadratic time \cite{ryanAvgCaseOV}. So, the next natural candidate problem that we might hope to show hard with our framework would be the \emph{counting} version of average-case Orthogonal Vectors problem (OV).
However, even the counting version of orthogonal vectors has truly subquadratic algorithm, as we will prove below. 

%However, this framework allows the probability of a $1$ to be any constant. This will give us sufficient flexibility to show no \goodPoly~should exist for uniform average-case Orthogonal Vectors. 
\begin{definition}
The counting $\mu$-uniform d-dimensional Orthogonal Vectors problem (\dOVm) takes as input two lists of $n$ zero-one vectors, where each vector is $d$-dimensional. All $2 \cdot n \cdot d$ bits are chosen iid where a one is selected with probability $\mu$. The output is the \emph{count} of the number of vectors that are orthogonal (whose dot product is zero). 
\end{definition}

We will consider constant $\mu$ for this section. We built up a few lemmas to prove the following theorem. 

\begin{reminder}{Theorem \ref{thm:avgCaseCountOVisEasy}}
For all constant values of $\mu$ and all values of $d$ there exists constants $\epsilon>0$ and $\delta>0$ such that there is an algorithm for \dOVm~that runs in time $\tO(n^{2-\delta})$ with probability at least $1-n^{-\epsilon}$.
\end{reminder}

We start by showing that if vectors are very long we are unlikely to have an orthogonal vector pair. 

\begin{lemma}
A \dOVm~instance has at most a $n^2 \cdot e^{-\mu^2 \cdot d}$ probability of having at least one pair of orthogonal vectors.
\label{lem:ovProb}
\end{lemma}
\begin{proof}
Any given pair of vectors has a probability of $(1-\mu^2)^d$ of being an orthogonal pair. 
The probability that some vector is an orthogonal pair is at most $n^2 \cdot (1-\mu^2)^d$ which is at most $ n^2 \cdot e^{-\mu^2 \cdot d}$.
\end{proof}

\begin{lemma}
If $d>(1+\delta) 2\lg(e)\lg(n)$ for some constant $\delta>0$ then there is a constant $\mu = (1+\delta)^{-1/4}$ such that \dOVm~instance has at least a $1-1/n^\epsilon$ probability of having no orthogonal vectors for some constant $\epsilon$.
\label{lem:NoOVprob}
\end{lemma}
\begin{proof}
Using Lemma \ref{lem:ovProb} and plugging in our value of $d$ we have that the probability of an \dOVm~instance having an orthogonal vector is at most
$n^2\cdot (n^2)^{-(1+\delta)\mu^2}$. If $\mu = (1+\delta)^{-1/4}$ then we can bound the probability by $n^{2(1-(1+\delta)^{1/2})}$. For $\delta>0$ we have that $(1+\delta)^{1/2}>1$, and so  $1-(1+\delta)^{1/2}$ is a negative constant. Thus there is some positive constant $\epsilon$ (for example $\epsilon = -2(1-(1+\delta)^{1/2})$) such that the probability there are no orthogonal vectors in a \dOVm~instance is at least  $1-1/n^\epsilon$. 
\end{proof}

A straightforward Corollary of Lemma \ref{lem:NoOVprob} is the following.

\begin{corollary}
For all constant $\mu$ there is a constant $\delta = 1/\mu^4-1$ such that for $d>(1+\delta) 2\lg(e)\lg(n)$ a \dOVm~instance has at least a $1-1/n^\epsilon$ probability of having no orthogonal vectors for some constant $\epsilon$.
\label{cor:largeDOV}
\end{corollary}
%% no need for a proof
%\begin{proof}
%Solve for $\delta$ given $\mu$ in the statement of Lemma \ref{lem:NoOVprob}.
%\end{proof}

%\begin{lemma}
%If $c=1-\epsilon/2$ for some constant epsilon and $d= c\lg(n)$ then there is a deterministic algorithm for \#\dOVm~which runs in $\tO(n+n^{2-\epsilon})$ time.
%\label{lem:superSmalld}
%\end{lemma}
%\begin{proof}
%If $d= (1-\epsilon/2)\lg(n)$ then there are at most $n^{1-\epsilon/2}$ vectors to represent. We can then compress the input to a list of these vectors and their counts. We can then exhaustively compare all these vectors any pair that is orthogonal contributes the multiplication of the number of instances of that vector that appear in the input. We can simply return this count. This exhaustive search takes $O(n+n^{2-\epsilon})$ time.
%\end{proof}

We use the following theorem appearing in \cite{countSmallDimOV} in the proof of Theorem \ref{thm:avgCaseCountOVisEasy}.
\begin{theorem}
Given a vector of dimension $d = c\lg(n)$ there is a $\tO(n^{2-1/O(\lg(c))})$ time algorithm that succeeds with probability $1$ on instances of \dOVm~in returning the count of the number of orthogonal vector pairs for every vector if one exists, regardless of $\mu$. \cite{countSmallDimOV}
\label{thm:smallDimOV}
\end{theorem}
%\begin{proof}
%This appears in the paper \cite{countSmallDimOV}.
%\end{proof}

Finally, we return to the proof of Theorem \ref{thm:avgCaseCountOVisEasy}. We show that even the counting version of the uniform average-case OV has a subquadratic algorithm. 

%\begin{reminder}{Theorem %\ref{thm:avgCaseCountOVisEasy}}
%For all constant values of $\mu$ and all values of $d$ there exists constants $\epsilon>0$ and $\delta>0$ such that there is an algorithm for \dOVm~that runs in time $\tO(n^{2-\delta})$ with probability at least $1-n^{-\epsilon}$.
%\end{reminder}
\begin{proofof}{Theorem \ref{thm:avgCaseCountOVisEasy}}
%\begin{proofof}{Theorem \ref{thm:avgCaseCountOVisEasy}}
Let the dimension be $d = c\lg(n)$.
%By Lemma \ref{lem:superSmalld} if $c< 1-\delta/2$ then there is a deterministic $\tO(n^{2-\delta})$ time algorithm.
By Corollary \ref{cor:largeDOV} if $c > 2\lg(e)/\mu^4$ then there is some $\epsilon>0$ such that there are no orthogonal vectors with probability at least $1-n^{-\epsilon}$. Notably, this gives us an $\tO(d)$ time algorithm where we return a count of zero if the dimension is larger than $2\lg(e)\lg(n)/\mu^4$ that succeeds with probability at least $1-n^{-\epsilon}$.

When $c\leq 2\lg(e)/\mu^4$ we will run the algorithm from Theorem \ref{thm:smallDimOV}. This runs in $\tO(n^{2-1/\lg(c)})$ time and is correct with probability $1$. This is at its worst a run time of $\tO(n^{2-1/\lg(2\lg(e)\mu^{-4})})$. So $\delta = \mu^4/\lg(2\lg(e)\mu^{-4})$, $\mu$ is a constant so $\delta$ is also a constant. 
%\end{proofof}
\end{proofof}

%\begin{lemma}
%If the worst case counting OV conjecture is true for some dimension $d>(1+\delta) 2\lg(e)\lg(n)$ where $d=n^{o(1)}$ then no \goodPoly~exists for OV instances with dimension  for some constant $\delta>0$. 
%\label{lem:NoGoodPolyOV}
%\end{lemma}
%\begin{proof}

%If a \goodPoly~existed for counting OV with dimension $d$ then by Theorem \ref{thm:framework} solving \dOVm~for any constant $\mu$ with probability $$1-1/\omega\left( \lg^{o(\lg(n)/\lg\lg(n))}\lg\lg^{o(\lg(n)/\lg\lg(n))} \right)$$
%in time $T(n)$ yields a randomized algorithm for OV with dimension $d$ in time $\Tilde{O}(T(n)+nd)$.

%Note that for any constant $\epsilon>0$:
%$$n^{\epsilon} =\omega\left( \lg^{o(\lg(n)/\lg\lg(n))}\lg\lg^{o(\lg(n)/\lg\lg(n))} \right). $$

%Lemma \ref{lem:NoOVprob} implies a trivial $O(1)$ time algorithm for \dOVm~for $\mu = (1+\delta)^{-1/4}$ that succeeds with probability at least $1-1/n^{\epsilon}$, simply return $0$. 

%Thus, if a \goodPoly~existed for counting OV with a dimension $d>(1+\delta) 2\lg(e)\lg(n)$ then worst case  $d$ dimensional OV would be solvable in time $\Tilde{O}(nd)$ which simplifies to $\Tilde{O}(n)$. 
%\end{proof}

%%%%%Count to detection
\section{Counting to Detection Reduction for Average-Case \texorpdfstring{\zkc}{Lg}}
\label{sec:countToDetect}
In fine-grained complexity the primary technique used for worst-case to average-case reductions has used the technique described by \cite{BallWorstToAvg}. This technique produces average-case hardness for computing the output of functions over a finite field. These problems are fundamentally counting problems. The issue with counting problems is that they are much harder to build cryptographic objects out of. 

Here we give a reduction from Counting to Detection for \zkc~in the average case (\aczkc). Notably, such a reduction \textbf{does not exist} in the worst case in fine-grained complexity. This makes the assumption that average case \zkc~detection with high probability requires $n^{k-o(1)}$ time more plausible. The assumption that \zkc~detection is hard with probability $1/100$  can be used to make fine-grained public-key cryptography \cite{fgCrytpo} (though the assumption that average-case \zkc~is hard with probability $1-n^{-o(1)}$ should be sufficient). There is a gap here between the probabilities we are describing, $1-1/\Omega( n^k)$, and the probabilities used for fine-grained cryptography, $1 - 1/n^{o(1)}$. However, this makes a step forward in closing the gap between the problems we can show are average-hard from worst-case assumptions and those we can build cryptography from.

Let us define average-case \zkc.

\begin{definition}
An average case instance of \zkc~(\aczkc) with range $R$ takes as input a complete $k$-partite graph with $n$ nodes in each partition. Every edge has a weight chosen iid from $[0,R-1]$. A clique is considered a zero $k$ clique if the sum of the edges is zero mod $R$.
\end{definition}

The idea of our reduction from counting to detection uses the fact that average-case \zkc~is easy when $R$ is small and there are very few solutions when $R$ is large. In the worst-case we can reduce detecting \aczkc~to counting $n^{k-\epsilon}$ \aczkc s when $\epsilon>0$. So, intuitively we are using the fact that when $R$ is small we can use a fast algorithm for counting. When $R$ is larger there are $n^{k-\epsilon}$ solutions, so we can use a reduction to show that faster detection solves counting those small number of solutions. 

First we will prove that when the range is small there is a fast algorithm. Then, we will show that a search algorithm counts very well when the range is exactly $R= n^k$. We will then show that this gives a generic counting to search reduction. Next, we will provide a search to decision reduction. Finally, we will give the counting to detection statement. 

Note that throughout this section we assume the function $p(n)$ is a monotonically \emph{non-increasing function}.
Additionally, when we say an algorithm succeeds in the average case with probability $p$, this is randomness over \emph{both} the input and the random coins flipped in the algorithm. 

\paragraph{Small Range is Easy}

\begin{lemma}
There is a $\tO(R^2 n^{\omega \lceil k/3 \rceil})$ time algorithm for \aczkc. 
\label{lem:smallRangeZKC}
\end{lemma}
\begin{proof}
Take the graph as a $k$-partite graph. 
Group together $k/3$ partitions of nodes. If $k$ is not a multiple of $3$ then make groups of $\lceil k/3 \rceil$ partitions and $\lfloor k/3 \rfloor$ partitions. Then, in each group of partitions create a node for every possible set of $\lceil k/3 \rceil$ or $\lfloor k/3 \rfloor$ nodes one from each partition. The total number of nodes is $O(n^{\lceil k/3 \rceil})$. 

Consider two nodes $v$ and $u$ where $v$ represents $x$ nodes and $u$ represents $y$ nodes. 
Add an edge between $u$ and $v$ only if all $x+y$ represented nodes form a clique. The weight on the edge between $u$ and $v$ is the sum of half the weight of all edges within the clique of $x$ nodes represented by $v$, half of the weight of all the edges within the clique of $y$ nodes represented by $u$, and the weight of all edges going between the $x$ nodes in $v$ and the $y$ nodes in $u$.

Now, the weights of the edges are still in the range $[0, R]$. We want to find a zero triangle in this new graph. We can guess the edge weights of two of the edges in the triangle, which forces the third value. Then, we produce a graph with only the edges of the guessed weights, then use matrix multiplication. All told this takes
$O(R^2 \cdot (n^{\lceil k/3 \rceil})^\omega )$ time. This can be simplified to $\tO(R^2 n^{\omega \lceil k/3 \rceil})$ time. 
\end{proof}

We can have a slight improvement in the running time of Lemma \ref{lem:smallRangeZKCExtend}.

\begin{lemma}
There is a $\tO(R^2 n^{(\omega (k-2)/3)+2})$ time algorithm for \aczkc. 
\label{lem:smallRangeZKCExtend}
\end{lemma}
\begin{proof}
Let $g$ be the largest integer such that $3g \leq k$. Note that $3g\geq k-2$. 

If $3g=k$ then by Lemma \ref{lem:smallRangeZKC} an algorithm exists which runs in time  $\tO(R^2 n^{\omega \lceil k/3 \rceil})$ time, which is $\tO(R^2 n^{(\omega (k-2)/3)+2})$.

If $3g=k-1$ then pick one partition, for every node in this partition we create a zero $k-1$ clique instance and use Lemma \ref{lem:smallRangeZKC} to get a $\tO(R^2 n^{\omega \lceil (k-1)/3  +1\rceil})$ time algorithm, which is $\tO(R^2 n^{(\omega (k-2)/3)+2})$.

If $3g=k-2$ then pick one partition, for every node in this partition we create a zero $k-2$ clique instance and use Lemma \ref{lem:smallRangeZKC} to get a $\tO(R^2 n^{(\omega (k-2)/3)+2})$ time algorithm.
\end{proof}
Lemma \ref{lem:smallRangeZKCExtend} gives the following corollary.
\begin{corollary}
\label{cor:ZKCeasySmallRange}
If $R = O(n^{(k-2-\omega (k-2)/3 -\epsilon)/2})$ then there is a $\tO(n^{k - \epsilon})$ time algorithm for \aczkc~with range $R$. 
\end{corollary}
%\begin{proof}
%Use Lemma \ref{lem:smallRangeZKCExtend}. 
%\end{proof}

\paragraph{High Probability Counting for $R= n^k$}
When the range is $n^k$ we want to count efficiently with very high probability. We will do this by first proving two helper lemmas.

\begin{lemma}
\label{lem:numOfSol}
The probability that an instance of \aczkc~with range $R=n^k$ has at least $k^k\lg^{2k}(n)$ solutions is  $2^{-\Omega(\lg^2(n))}$. 
\end{lemma}
\begin{proof}
If there are at least $\left(k\lg^2(n)\right)^k$ zero cliques then there is at least one set of $\lg^2(n)$ cliques such that each zero clique has at least one node not shared by any other zero clique. After all at least $k\lg^2(n)$ distinct nodes must be involved in these $\left(k\lg^2(n)\right)^k$ zero cliques. 

If a zero clique has a node not shared with the other cliques then whether or not it is a zero clique is uncorrelated with the other zero cliques. So, the probability that there are $\left(k\lg^2(n)\right)^k$ zero cliques is at most the probability that out of $n^k$ independent trials $\lg^2(n)$ return true when the probability of a trial returning true is $1/R = 1/n^k$. By the Chernoff bound we get the probability of this event is less than $2^{- lg^2(n)/3}$ which is $2^{- \Omega(lg^2(n))}$. 
\end{proof}

\begin{lemma}
Using a search algorithm, $\mathcal{A}$,  that succeeds with probability $1-p$ on an instances of \aczkc~with $n/(k \lg^2(n))$ nodes per partition and edge weights in the range $R=\theta{(n^k)}$ in time $T(n)$ we can count the number of solutions (or list all those solutions) in a \aczkc~instance in time $\tO(T(n)+n^{k-1})$ with probability at least $1-pk^k\lg^{2k+2}(n)-2^{-\Omega(\lg^2(n))}$.
\label{lem:countFromSearch}
\end{lemma}
\begin{proof}
Let the input \aczkc~instance be the graph $G$ with edge set $E$ and vertex set $V$.
First, note that with probability $1-1/2^{\Omega(\lg^2(n))}$ there are at most $s =k^k\lg^{2k}(n)$ zero $k$-cliques (ZKCs). 

Now consider a given ZKC $c$ in $G$. Imagine creating a new instance $G'$ that is a subset of $G$ by selecting a random subset of $n/x$ nodes from each partition. The ZKC $c$ is in $G'$ with probability $x^{-k}$. Now consider a clique $c'$ which is in $G$ and shares no nodes with $c$. Given that $c$ is in $G'$ the probability that $c'$ is also in $G'$ is at most $x^{-k}$. If there are at most $\ell$ cliques in $G$ then the probability that a given clique $c$ is in $G'$ and no disjoint cliques (cliques that share no vertices with $c$) are in $G'$ is at least:
$x^{-k}(1-x^{-k})^{\ell}.$ Further note that the sub-graph $G'$ has total variation distance $0$ from \aczkc~instances with $n/x$ nodes per partition and range $n^k$. 

Consider the algorithm $\mathcal{B}_x$. It creates an empty set $S_B$ that it will fill with cliques it finds. It generates $G'$ at random by selecting a random set of $n/x$ nodes from each partition. Then it runs $\mathcal{A}$ on $G'$. If $\mathcal{A}$ returns a clique $c$, check that it is a ZKC. If it is, further exhaustively check that there is no clique that shares a node with $c$ this takes $O(k(n/x)^{k-1})$ time (you can simply check all sets of $k$ nodes involving one node in the clique). Any cliques it finds in this search are added to $S_B$ and $S_B$ is returned. $\mathcal{B}_x$ takes $O(T(n)+n^{k-1})$ time. If:
\begin{enumerate}
    \item a ZKC $c$ is in $G'$,
    \item $\mathcal{A}$ returns correctly, and
    \item there are no ZKCs in $G'$ which share no vertices with $c$
\end{enumerate} 
then $\mathcal{B}_x$ will include $c$ in $S_B$. Because $\mathcal{A}$ returned a ZKC and it was either $c$ or a clique that shared a node with $c$. In the later case our exhaustive search would find it. 
Given a specific clique $c$ and $\mathcal{A}$ returning correctly $c \in S_B$ with probability $x^{-k}(1-x^{-k})^{\ell}$.

Consider the case where $\ell \leq s$ and $x = k \lg^2(n) =s^{1/k}$. Then $\mathcal{B}_{k \lg^2(n)}$ returns a given $c$ with probability at least $s^{-1}(1-s^{-1})^{s} \geq \frac{1}{4s}$. If $\mathcal{A}$ is returning correctly every trial is independent. Thus if we run $\mathcal{B}_{k \lg^2(n)}$ $4s\lg^2(n)$ times we will find the clique $c$ with probability at least $1-(1-1/(4s))^{4s\lg^2(n)} \geq 1-2^{-\Omega(\lg^2(n))}$. The probability we find all the ZKCs (given that there are at most $s$ ZKCs) is, by union bound at least $1-s2^{-\Omega(\lg^2(n))} = 1-2^{-\Omega(\lg^2(n))}$.

After making $4s\lg^2(n)$ calls to $\mathcal{B}_{k \lg^2(n)}$ we will have made $4s\lg^2(n)$ calls to $\mathcal{A}$. Using union bound all of these will succeed with probability at least $1-4s\lg^2(n)p$. 

So the time we take is $O(4s\lg^2(n)T(n)+4s\lg^2(n)n^{k-1})$ which is $\tO(T(n)+n^{k-1})$. Our success probability requires the union of the number of cliques being less than $s$, $\mathcal{A}$ returning correctly on all calls, and the randomness in $\mathcal{B}_{k\lg^2(n)}$ allowing us to return all cliques. Thus our probability of success is at least $1-4s\lg^2(n)p-2^{-\Omega(\lg^2(n))}-2^{-\Omega(\lg^2(n))}$. This can be simplified to a success probability of $1-4k^k\lg^{2k+2}(n)p-2^{-\Omega(\lg^2(n))}$. 

\end{proof}

\paragraph{Counting to Search}

We will start by describing the self reduction for \aczkc. This is a folklore self-reduction in the worst case and was analyzed in the average case in \cite{fgCrytpo}.

\begin{lemma}
Given an instance, $I$, of average case \zkc~with range $R$ with $kn$ nodes it can be split into $(n/x)^k$ instances $I_1, \ldots, I_{(n/x)^k}$ each with $kx$ nodes such that:
\begin{enumerate}
    \item The distribution over each $I_i$ is the average case distribution with $kx$ nodes and range $R$. (Though two instances $I_i$ and $I_{i'}$ may be correlated.)
    \item The number of solutions in instance $I$ ($\#Solutions(I)$) is equal to the sum of solutions in all the instances $I_1, \ldots, I_{(n/x)^k}$ ($\sum_{i=1}^{(n/x)^k} \#Solutions(I_i)$).
\end{enumerate}
\label{lem:selfREductionACzkc}
\end{lemma}
\begin{proof}
Note the $k$-partite graph in the instance $I$ and note each partition of vertices $V_1, \ldots, V_k$. We create a random partition of each vertex set into $n/x$ sets of $x$ vertices. Name the subsets of $V_i$, $V_i[j]$ where $j\in [1,n/x]$.
The $(n/x)^k$ subinstances are formed by taking the intersection of $k$ subsets one from each of the $k$ partitions: $V_1[j_1] \cup \ldots \cup V_k[j_k]$ for all possible $k$ tuples $(j_1,\ldots,j_k) \in [n/x]^{k}$.  

For the first claim, note that for any given instance $I_i$ we simply have a random selection of $n/x$ nodes from an average case instance. So every edge is chosen iid from $[R]$. This is indeed the distribution of an average case \zkc~instance. We will note that two separate instances may be correlated. For example the instance formed by $V_1[j_1] \cup V_2[j_2] \cup V_3[j_3] \ldots \cup V_k[j_k]$ and the instance formed by $V_1[j_1] \cup V_2[j_2] \cup V_3[j_3'] \ldots \cup V_k[j_k']$ will share all edges between sections $V_1[j_1]$ and $V_2[j_2]$. Of course union bounds can still be used to bound error between these $(n/x)^k$ instances. 

For the second claim, any \aczkc~witness has $k$ nodes one from each partition: $v_1\in V_1, \ldots,  v_k\in V_k$. Every witness appears in exactly one sub-instance. A given witness $(v_1, \ldots, v_k)$ will appear only in the instance formed by a union of the subsets $V_1[j_1] \cup \ldots \cup V_k[j_k]$  where $v_i \in V_i[j_i]$ in every subset. 
\end{proof}

\begin{lemma}
Let $p(n)$ be a monotonically non-increasing function. 

Assume an algorithm exists for the search version of \aczkc~with range $R \in [k^k\lg^{2k}(n)n^k, 2k^k\lg^{2k}(n)n^k]$ that succeeds with probability at least $1-p(n)$ and runs in time $O(n^{k-\epsilon})$ where $\eps>0$. Let $k'=2+\omega(k-2)/3$. Then there is an algorithm for counting the number of \zkc~in an average case instance for any positive integer $R$ with probability at least $1-2^{-\Omega(\lg^2(n))}-p\left(n^{(k-k' -\delta)/(2k)}/(k\lg^2(n))\right)\cdot n^{(k' +\delta)/2} \cdot k^k\lg^{k2}(n)$ that runs in time $\tO(n^{k -\delta})$ for some constant $\delta>0$.
\label{lem:countToSearch}
\end{lemma}
\begin{proof}
Let us call the search algorithm $\mathcal{A}$.
There are two cases to consider. $R\leq n^k$ and $R > n^k$.

If $R > n^k$ then we can use nearly linear hashing (see \cite{Patrascu10})  to reduce our range down to $n^k$. There may be false positives here, however, the instance will look uniformly random (we are hashing large uniformly random numbers). So we can use Lemma \ref{lem:numOfSol} to say that there will be at most $s=k^k\lg^{k2}(n)$ solutions (false positives or true positives) with probability at least $1 -2^{-\Omega(\lg^2(n))}$. Now, we can use the algorithm from Lemma \ref{lem:countFromSearch} to list all solutions with probability at least $1-p(n/(k\lg^2(n)))s-1/2^{\Omega(\lg^2(n))}$. For each listed solution we can check if it is a false positive and only count the actual cliques. This will return the true number of cliques with probability at least $1-p(n/(k\lg^2(n)))s-2^{-\Omega(\lg^2(n))}$. This requires $s$ calls to $\mathcal{A}$ so it takes time $O(n^{k-\epsilon} s)$. This constrains $\delta < \epsilon$.

If $R \leq n^{(k-k' -\epsilon')/2}$ where $k'=2+\omega(k-2)/3$, then by Corollary \ref{cor:ZKCeasySmallRange} there is a $O(n^{k -\epsilon'})$ time algorithm that succeeds with probability $1$.

If $ n^{(k-k' -\epsilon')/2} < R< n^k$ then we will use the average case self reduction for \aczkc~(see Lemma \ref{lem:selfREductionACzkc} or \cite{fgCrytpo}) to reduce the problem to problems of size $x = R^{1/k}$, so now we have that $R = x^k$ where $x$ is our new smaller input size. We can now call $\mathcal{A}$ on all these instances. Note that $s\geq n^{(k-k'-\epsilon')/(2k)}$. Further note that the total number of instances is $(n/x)^k \leq n^{(k' +\epsilon')/2}$. So, if $\mathcal{A}$ succeeds with probability $1-p(n)$ then by union bound these independent instances will succeed with probability at least $1-p(x/(k\lg(x)))\cdot(n/x)^k$. Now note that this is at least $1-p(n^{(k-k' -\epsilon')/(2k)})\cdot n^{(k' +\epsilon'/(k\lg(n)))/2}$. If there is an algorithm running in time $O(x^{k-\eps})$ for all $(n/x)^k$ problems then the running time is $\tO(n^{k-\eps (k-k' -\epsilon')/(2k) })$. Notably $\eps (k-k' -\epsilon')/(2k)>0$. 

We want $\delta < \epsilon '$ and $\delta < \eps (k-k' -\epsilon')/(2k) < \eps/2$. If we choose $\delta<\epsilon'=\eps/2$ then this meets all of our constraints. In every case the algorithm succeeds with probability at least $1-2^{-\Omega(\lg^2(n))}-p(n^{(k-k' -\delta)/(2k)})\cdot n^{(k' +\delta)/2} \cdot s$ and runs in time $\tO(n^{k -\delta})$ when $\delta<\epsilon'=\eps/2$.
\end{proof}

\paragraph{Search to Decision}

\begin{lemma}
Let $p(n)$ be a monotonically non-increasing function. 

Given a detection algorithm that runs in $O(n^{k-\delta})$ time for some $\delta>0$ and has success probability at least $1-p(n)$ we can produce a search algorithm that runs in time $\tO\left(n^{k-\epsilon k} + n^{k-\delta(1-\epsilon)}\right)$ for any constant $1/2>\epsilon>0$ and has success probability at least $1- p(n^{1-\epsilon})n^{k\epsilon}$.

Specifically for $\epsilon=1/2$ this can be bounded as $\tO\left(n^{k-\delta/2}\right)$ time  and probability at least $1- p(n^{1/2})n^{k/2}$.
\label{lem:searchDecsions}
\end{lemma}
\begin{proof}
We use the classic self reduction for \aczkc~producing instances of size $n^{1-\epsilon}$. For this we randomly split each partition of vertices into $n^{\epsilon}$ groups of $n^{1-\epsilon}$ nodes. We form all $n^{k \epsilon}$ possible sub-problems and run the detection algorithm on them. On any instance that returns true we brute force the problem in $O(n^{k-k\epsilon})$ time. We of course can stop as soon as we find a clique. 

The probability that none of our $n^{k\epsilon}$ instances produces a false positive is at least $1- p(n^{1-\epsilon})n^{k\epsilon}$. If we have no false positives then our running time is $O(n^{k-k\epsilon}+ n^{k\epsilon}(n^{1-\epsilon})^{k-\delta})$. This can be simplified to $O(n^{k-k\epsilon}+ n^{k-\delta(1-\epsilon)})$ 

%Now note that $\delta\leq k$ (the algorithm can't be faster than $O(1)$ time). Now note that we can upper bound our run time by $O(n^{k-k\epsilon}+ n^{k-k\epsilon-(1-\epsilon)\delta})$. So 
\end{proof}

\paragraph{Counting to Decision}~    

\begin{lemma}
Let $p(n)$ be a monotonically non-increasing function. 

Given a decision algorithm for \aczkc~that runs in time $O(n^{k-\eps})$ for some $\eps>0$ and succeeds with probability at least $1-p(n)$ there is a counting algorithm that runs in $O(n^{k-\eps'})$ for some $\eps'>0$ and succeeds with probability at least $1 -2^{-\lg^2(n)}-p(n^{1/25})n^k$.
\label{lem:countToDecision}
\end{lemma}
\begin{proof}
Use Lemma \ref{lem:searchDecsions} when $\epsilon =1/2$ and Lemma \ref{lem:countToSearch}. 
When combing our numbers we find that the probability is at least $1-p(n^{1/25})n^k$. 

Note this is not tight, by tuning $\epsilon$ and plugging in an improved value for the matrix multiplication constant you will get a tighter result. This bound is sufficient for our purposes so we leave it as is. 
\end{proof}

\begin{reminder}{Theorem \ref{thm:countToDecision}}
Given a decision algorithm for \aczkc~that runs in time $O(n^{k-\eps})$ for some $\eps>0$ and succeeds with probability at least $1-n^{-\omega(1)}$, there is a counting algorithm that runs in $O(n^{k-\eps'})$ time for some $\eps'>0$ and succeeds with probability at least $1-n^{-\omega(1)}$, where $\omega(1)$ here means any function that is asymptotically larger than constant. 
\end{reminder}
\begin{proof}
We plug in $n^{-\omega(1)}$ for $p(n)$ in Lemma \ref{lem:countToDecision}. Note that the second error term, $2^{-\lg^2(n)}$ is $n^{-\omega(1)}$.
\end{proof}

\section{Future Work}
\label{sec:FutureWork}
Average-case fine-grained complexity still has a lot of unexplored areas. We suggest the following open problems that directly relate to results of this work.

%We could add a discussoin of (k+1)L-MF

\paragraph{General Questions}
What other natural non-factored problems are hard from factored problems (either \ckfunc[]~and \cfkc[])? 
We give three problems in section \ref{sec:harderProblem} where we only show their detection version is hard. Can one show that a counting version of $(k+1)$L-MF, $k$-LCS, or Edit Distance is hard from counting factored problems? Recall that such a reduction would imply average case hardness over some distribution for the problem reduced to. 
We show hardness for the uniform average case for  \#\ckfunc[]~and \#\cfkc[], can one show hardness for other natural worst case distributions of these problems?

\paragraph{Cryptography and Counting vs Detection} In Section \ref{sec:countToDetect} we show that detecting \zkc~with high probability in the average case implies fast algorithms for counting with high probability in the average case.
\begin{itemize}
\item\emph{Counting to detection in the high error regime:} Can you show that a detection algorithm for average-case \zkc~that succeeds with probability $1-1/(polylog(n))$ implies an algorithm for counting \zkc~with probability $2/3$?  If such a reduction exists in the high error regime you can build cryptography protocols from an assumption about the difficulty of counting \zkc~on average \cite{fgCrytpo}.
\item\emph{Worst case \zkc~to counting \zkc~on average:} Can we reduce the worst case hardness of \zkc~to average case \#\zkc? What about \ksum? If you can prove this for \zkc~and prove the previous high-error regime reduction, then you can build fine-grained cryptography from a worst-case assumption about the complexity of \zkc.  
\item \emph{Counting to detection for other problems:} A similar proof technique that we use for \zkc~should work for the \ksum[3]~problem. For this style of reduction we need: (1) an efficient average-case self-reduction for the problem, (2) the number of witnesses to be small on average when some parameter $R$ is large, and (3) an efficient algorithm when $R$ is small. All of these exist for \ksum[3], however, there isn't an efficient self reduction for \ksum[]~for $k>3$. Can another approach work to show counting to detection results for problems like \ksum[], \klcs[], etc?
\end{itemize}

\paragraph{Using/Extending the Good Low-Degree Polynomial Framework}
A few directions that could be taken with respect to our framework are the following:
\begin{itemize}
    \item Can the framework be extended to handle multiple outputs? For example, the problem of multiplying two zero-one matrices? 
    \item Can we find new problems $P$ that have $\gPol{P}$?
    \item Can the framework be improved? For example, could it be improved to handle polynomials of (slightly) greater degree? Can the strong $k$-partiteness condition be weakened?
\end{itemize}

\paragraph{LCS and Edit Distance} In Section \ref{subsec:lcs} we cover LCS and Edit Distance.  We have two open problems from this section we want to highlight.
\begin{itemize}
\item\emph{Making a framework for string distance lower bounds from factored problems:}
%``Quadratic Conditional Lower Bounds for String Problems and Dynamic Time Warping'' 
Bringmann and K{\"{u}}nnemann~\cite{alignmentGadget} create a framework for proving $n^{2-o(1)}$ lower bounds from SETH.  %\cite{alignmentGadget}.
We believe this same framework can be extended to work for \ckov[2]~by adding a requirement of a selection gadget. It also seems that this framework could be extended to contain \ckov[]. Relatedly, can  k-median distance and k-center-edit-distance be reduced to \ckov[]?
\item\emph{Getting tight hardness for \#\klcs[]~or \#\kwlcs[]:} We note that the counting versions of \klcs~and~\kwlcs[]~both have algorithms that run in time $n^{k+o(1)}$ (see Appendix \ref{app:algorithms}). Given our construction, the counting versions of \klcs~and~\kwlcs[]~count $\circledcirc(v_1,\ldots,v_k)$ is given as input the $k$ strings $FVG(v_1)$, $\ldots, FVG(v_k)$. However, unfortunately, the counting versions of \klcs~and~\kwlcs[]~do not return \# \ckov[]~given our construction of $P_1,\ldots,P_k$. This is due to using an alignment gadget instead of a selector gadget. If we used the selector gadget, the count of longest common subsequences would be the sum over the counts of all $FVG(v_1), \ldots, FVG(v_k)$ where $\circledcirc(v_1,\ldots,v_k)>0$. This would result in the count being exactly the output of \# \ckov[]. However, the strings produced by our reduction would have length $n^2$ and weights of size $\tO(n)$. So, we would get a lower bound of $n^{k/2-o(1)}$ for \#\kwlcs[], and a lower bound of $n^{k/3-o(1)}$ for \#\klcs[]. A more efficient selector gadget would yield tight lower bounds for \# \klcs~and \# \kwlcs[], including in the average-case. We suggest this as a potential topic for future work.
\end{itemize}

\section*{Acknowledgements}
We would like to acknowledge Marshall Ball for interesting and helpful early discussions. 

We would like to acknowledge all our reviewers for helpful comments. We thank all the reviewers for advice about improving the readability of the paper. We would like to extend special thanks to reviewer 1 who noted that we could improve the definition of \goodPoly~by removing a restriction! 

% Decrease the space between bibliography items.
\let\realbibitem=\bibitem
\def\bibitem{\par \vspace{-0.5ex}\realbibitem}

\bibliographystyle{alpha}
\bibliography{bib} 

\newcommand{\etalchar}[1]{$^{#1}$}
\begin{thebibliography}{GCSR13}

\bibitem[ABV15]{LCSisHard}
Amir Abboud, Arturs Backurs, and Virginia {Vassilevska Williams}.
\newblock Tight hardness results for {LCS} and other sequence similarity
  measures.
\newblock In {\em {IEEE} 56th Annual Symposium on Foundations of Computer
  Science, {FOCS} 2015, Berkeley, CA, USA, 17-20 October, 2015}, pages 59--78,
  2015.

\bibitem[ALW14]{losingWeight}
Amir Abboud, Kevin Lewi, and Ryan Williams.
\newblock Losing weight by gaining edges.
\newblock In Andreas~S. Schulz and Dorothea Wagner, editors, {\em Algorithms -
  {ESA} 2014 - 22th Annual European Symposium, Wroclaw, Poland, September 8-10,
  2014. Proceedings}, volume 8737 of {\em Lecture Notes in Computer Science},
  pages 1--12. Springer, 2014.

\bibitem[AVY18]{matchingTriangles}
Amir Abboud, Virginia {Vassilevska Williams}, and Huacheng Yu.
\newblock Matching triangles and basing hardness on an extremely popular
  conjecture.
\newblock {\em SIAM Journal on Computing}, 47(3):1098--1122, 2018.

\bibitem[BBB19]{UniformCliqueABB}
Enric Boix{-}Adser{\`{a}}, Matthew Brennan, and Guy Bresler.
\newblock The average-case complexity of counting cliques in
  erd{\H{o}}s-r{\'{e}}nyi hypergraphs.
\newblock In David Zuckerman, editor, {\em 60th {IEEE} Annual Symposium on
  Foundations of Computer Science, {FOCS} 2019, Baltimore, Maryland, USA,
  November 9-12, 2019}, pages 1256--1280. {IEEE} Computer Society, 2019.

\bibitem[BI15]{BackursI15}
Arturs Backurs and Piotr Indyk.
\newblock Edit distance cannot be computed in strongly subquadratic time
  (unless {SETH} is false).
\newblock In Rocco~A. Servedio and Ronitt Rubinfeld, editors, {\em Proceedings
  of the Forty-Seventh Annual {ACM} on Symposium on Theory of Computing, {STOC}
  2015, Portland, OR, USA, June 14-17, 2015}, pages 51--58. {ACM}, 2015.

\bibitem[BI16]{regularexpression}
Arturs Backurs and Piotr Indyk.
\newblock Which regular expression patterns are hard to match?
\newblock In {\em 2016 IEEE 57th Annual Symposium on Foundations of Computer
  Science (FOCS)}, pages 457--466. IEEE, 2016.

\bibitem[BI18]{BackursI18}
Arturs Backurs and Piotr Indyk.
\newblock Edit distance cannot be computed in strongly subquadratic time
  (unless {SETH} is false).
\newblock {\em {SIAM} J. Comput.}, 47(3):1087--1097, 2018.

\bibitem[BK15]{alignmentGadget}
Karl Bringmann and Marvin K{\"{u}}nnemann.
\newblock Quadratic conditional lower bounds for string problems and dynamic
  time warping.
\newblock In Venkatesan Guruswami, editor, {\em {IEEE} 56th Annual Symposium on
  Foundations of Computer Science, {FOCS} 2015, Berkeley, CA, USA, 17-20
  October, 2015}, pages 79--97. {IEEE} Computer Society, 2015.

\bibitem[BRSV17]{BallWorstToAvg}
Marshall Ball, Alon Rosen, Manuel Sabin, and Prashant~Nalini Vasudevan.
\newblock Average-case fine-grained hardness.
\newblock In Hamed Hatami, Pierre McKenzie, and Valerie King, editors, {\em
  Proceedings of the 49th Annual {ACM} {SIGACT} Symposium on Theory of
  Computing, {STOC} 2017, Montreal, QC, Canada, June 19-23, 2017}, pages
  483--496. {ACM}, 2017.

\bibitem[BRSV18]{BallRSV18}
Marshall Ball, Alon Rosen, Manuel Sabin, and Prashant~Nalini Vasudevan.
\newblock Proofs of work from worst-case assumptions.
\newblock In Hovav Shacham and Alexandra Boldyreva, editors, {\em Advances in
  Cryptology - {CRYPTO} 2018 - 38th Annual International Cryptology Conference,
  Santa Barbara, CA, USA, August 19-23, 2018, Proceedings, Part {I}}, volume
  10991 of {\em Lecture Notes in Computer Science}, pages 789--819. Springer,
  2018.

\bibitem[CGI{\etalchar{+}}16]{CarmosinoGIMPS16}
Marco~L. Carmosino, Jiawei Gao, Russell Impagliazzo, Ivan Mihajlin, Ramamohan
  Paturi, and Stefan Schneider.
\newblock Nondeterministic extensions of the strong exponential time hypothesis
  and consequences for non-reducibility.
\newblock In Madhu Sudan, editor, {\em Proceedings of the 2016 {ACM} Conference
  on Innovations in Theoretical Computer Science, Cambridge, MA, USA, January
  14-16, 2016}, pages 261--270, 2016.

\bibitem[CR02]{caprara2002packing}
Alberto Caprara and Romeo Rizzi.
\newblock Packing triangles in bounded degree graphs.
\newblock {\em Information Processing Letters}, 84(4):175--180, 2002.

\bibitem[CW16]{countSmallDimOV}
Timothy~M. Chan and Ryan Williams.
\newblock Deterministic apsp, orthogonal vectors, and more: Quickly
  derandomizing razborov-smolensky.
\newblock In Robert Krauthgamer, editor, {\em Proceedings of the Twenty-Seventh
  Annual {ACM-SIAM} Symposium on Discrete Algorithms, {SODA} 2016, Arlington,
  VA, USA, January 10-12, 2016}, pages 1246--1255. {SIAM}, 2016.

\bibitem[GCSR13]{mfml}
Donatella Granata, Raffaele Cerulli, Maria~Grazia Scutell{\`a}, and Andrea
  Raiconi.
\newblock Maximum flow problems and an np-complete variant on edge-labeled
  graphs.
\newblock {\em Handbook of Combinatorial Optimization}, pages 1913--1948, 2013.

\bibitem[GO95]{C3sum}
Anka Gajentaan and Mark~H. Overmars.
\newblock On a class of o(n2) problems in computational geometry.
\newblock {\em Comput. Geom.}, 5:165--185, 1995.

\bibitem[GR18]{GoldreichR18}
Oded Goldreich and Guy~N. Rothblum.
\newblock Counting t-cliques: Worst-case to average-case reductions and direct
  interactive proof systems.
\newblock In Mikkel Thorup, editor, {\em 59th {IEEE} Annual Symposium on
  Foundations of Computer Science, {FOCS} 2018, Paris, France, October 7-9,
  2018}, pages 77--88. {IEEE} Computer Society, 2018.

\bibitem[GR20]{goldreichrothblum20}
Oded Goldreich and Guy~N. Rothblum.
\newblock Worst-case to average-case reductions for subclasses of {P}.
\newblock In {\em Computational Complexity and Property Testing - On the
  Interplay Between Randomness and Computation}, volume 12050 of {\em Lecture
  Notes in Computer Science}, pages 249--295. Springer, 2020.

\bibitem[IF92]{LCSdef}
Robert~W Irving and Campbell~B Fraser.
\newblock Two algorithms for the longest common subsequence of three (or more)
  strings.
\newblock In {\em Annual Symposium on Combinatorial Pattern Matching}, pages
  214--229. Springer, 1992.

\bibitem[IP01]{cseth}
Russell Impagliazzo and Ramamohan Paturi.
\newblock On the complexity of $k$-{SAT}.
\newblock {\em J. Comput. Syst. Sci.}, 62(2):367--375, 2001.

\bibitem[Kus19]{EditDistLCS}
William Kuszmaul.
\newblock Dynamic time warping in strongly subquadratic time: Algorithms for
  the low-distance regime and approximate evaluation.
\newblock In Christel Baier, Ioannis Chatzigiannakis, Paola Flocchini, and
  Stefano Leonardi, editors, {\em 46th International Colloquium on Automata,
  Languages, and Programming, {ICALP} 2019, July 9-12, 2019, Patras, Greece},
  volume 132 of {\em LIPIcs}, pages 80:1--80:15. Schloss Dagstuhl -
  Leibniz-Zentrum f{\"{u}}r Informatik, 2019.

\bibitem[KW19]{ryanAvgCaseOV}
Daniel~M. Kane and R.~Ryan Williams.
\newblock The orthogonal vectors conjecture for branching programs and
  formulas.
\newblock In Avrim Blum, editor, {\em 10th Innovations in Theoretical Computer
  Science Conference, {ITCS} 2019, January 10-12, 2019, San Diego, California,
  {USA}}, volume 124 of {\em LIPIcs}, pages 48:1--48:15. Schloss Dagstuhl -
  Leibniz-Zentrum fuer Informatik, 2019.

\bibitem[LLV19]{fgCrytpo}
Rio LaVigne, Andrea Lincoln, and Virginia {Vassilevska Williams}.
\newblock Public-key cryptography in the fine-grained setting.
\newblock In Alexandra Boldyreva and Daniele Micciancio, editors, {\em Advances
  in Cryptology - {CRYPTO} 2019 - 39th Annual International Cryptology
  Conference, Santa Barbara, CA, USA, August 18-22, 2019, Proceedings, Part
  {III}}, volume 11694 of {\em Lecture Notes in Computer Science}, pages
  605--635. Springer, 2019.

\bibitem[Pat10]{Patrascu10}
Mihai Patrascu.
\newblock Towards polynomial lower bounds for dynamic problems.
\newblock In Leonard~J. Schulman, editor, {\em Proceedings of the 42nd {ACM}
  Symposium on Theory of Computing, {STOC} 2010, Cambridge, Massachusetts, USA,
  5-8 June 2010}, pages 603--610. {ACM}, 2010.

\bibitem[{Vas}18]{virgiSurvey}
Virginia {Vassilevska Williams}.
\newblock On some fine-grained questions in algorithms and complexity.
\newblock In {\em Proceedings of the ICM}, volume~3, pages 3431--3472. World
  Scientific, 2018.

\bibitem[VW10a]{CAPSP}
Virginia {Vassilevska Williams} and Ryan Williams.
\newblock Subcubic equivalences between path, matrix and triangle problems.
\newblock In {\em 51th Annual {IEEE} Symposium on Foundations of Computer
  Science, {FOCS} 2010, October 23-26, 2010, Las Vegas, Nevada, {USA}}, pages
  645--654. {IEEE} Computer Society, 2010.

\bibitem[VW10b]{williams2010subcubic}
Virginia {Vassilevska Williams} and Ryan Williams.
\newblock Subcubic equivalences between path, matrix and triangle problems.
\newblock In {\em Foundations of Computer Science (FOCS), 2010 51st Annual IEEE
  Symposium on}, pages 645--654. IEEE, 2010.

\bibitem[VW13]{vw09j}
Virginia {Vassilevska Williams} and Ryan Williams.
\newblock Finding, minimizing, and counting weighted subgraphs.
\newblock {\em {SIAM} J. Comput.}, 42(3):831--854, 2013.

\bibitem[VW18]{vw10j}
Virginia {Vassilevska Williams} and R.~Ryan Williams.
\newblock Subcubic equivalences between path, matrix, and triangle problems.
\newblock {\em J. {ACM}}, 65(5):27:1--27:38, 2018.

\bibitem[Wil07]{ryanThesis}
Ryan Williams.
\newblock Algorithms and resource requirements for fundamental problems.
\newblock {\em Ph. D. dissertation, Ph. D. Thesis}, 2007.

\end{thebibliography}
\appendix

\section{Removed Algorithms}
\label{app:algorithms}

%\subsection{Graphs}

%\begin{definition}
%Let $H=(V_H,E_H)$ be a $k$-node graph with $V_H=\{x_1,\ldots,x_k\}$. 

%An $H$-partite graph is a graph with $k$ partitions $V_1,\ldots,V_k$. This graph must only have edges between nodes $v_i \in V_i$ and $v_j \in V_j$ if e $(x_i,x_j)\in E_H$. (See Figure \ref{fig:Hpartite})
%\end{definition}

%\begin{figure}[ht]
%\centering
%\includegraphics[width=0.5\textwidth]{HpartiteImage.png}
%\caption{An example of the corresponding $H$-partite graphs. }
%\label{fig:Hpartite}
%\end{figure}

\subsection{Algorithms for Factored Problems}
\label{subsubsec:WCReductionsCompressed}

These algorithms are straightforward be case they are simply brute force. 

\begin{lemma}
\ckfunc[]~can be solved in $\tO(n^k)$ time.
\end{lemma}
\begin{proof}

For \ckfunc[], we want to run $\circledcirc(v_1,\ldots,v_k)$ on every set of $k$ vectors. 
To do this we need to compute $\circ(v_1[i],\ldots,v_k[i])$ for all $i \in [0,g]$. Running $\circ(\cdot)$ takes $O( \Pi_{j\in[1,k]}|v_j[i]|)$ time. We can use the upper bound $|v_j[i]| \leq 2^b$. So computing $\circ(\cdot)$ takes at most $O(2^{bk})$ time. Thus, computing $\circledcirc(\cdot)$ takes at most $O(g \cdot 2^{bk})$. 

Thus, computing \ckfunc[]~ takes at most $O(n^k \cdot g \cdot 2^{bk})$ time. We bounded $g$ and $b$ to be $o(\lg(n))$ so $g \cdot 2^{bk}$ is subpolynomial. 
Thus, \ckfunc[]~ takes at most $\tO(n^k)$ time.
\end{proof}

\begin{lemma}
\cfkc[]~(\#\cfkc[]) can be solved in time $\tO(n^k)$.
\end{lemma}
\begin{proof}
For every $k$ tuple of nodes in the graph we want to evaluate $\circledcirc( \cdot )$. If we can evaluate $\circledcirc( \cdot )$ in $\tO(1)$ time then we can count or detect in  $\tO(n^k)$ time.

Evaluating $isClique(\cdot) $ can be done in $\tO(1)$ time. Evaluating the multiplication, given the results of the function $\circ(\cdot)$ can be done in $\tO(g)$ time.
Evaluating $\circ(\cdot)$ should require at most $\tO(2^{\binom{k}{2}b})$ time as the function $f$ has a total truth table size of $2^{\binom{k}{2}b}$ and we simply need to evaluate how many entries of the truth table are $1$ while we simultaneously have that vector. 

Finally, we note that $g = \tO(1)$ and $2^{\binom{k}{2}b} = \tO(1)$. So evaluating $\circledcirc( \cdot )$ can be done in $\tO(1)$ time. 
\end{proof}

\subsection{Algorithms for Problems Harder than Factored Problems }
\begin{theorem}
Counting partitioned matching triangles  (\#\NDMT) can be solved in $\tilde{O}(n^3)$ time. 
\end{theorem}
\begin{proof}
Let the $g$ graphs be $G_1,\ldots,G_g$ in our \NDMT~instance, and let $n_r^j$ be the number of nodes of color $r$ in $G_j$. For all triple of colors $(c_1,c_2,c_3)$ and all $j$, we count the number of
triangles of these colors in $G_j$. We can do this by inspecting every triple of nodes of color $(c_1,c_2,c_3)$ in time $n_{c_1}^jn_{c_2}^jn_{c_3}^j$. Since $\sum_r n_r^j=n$ for all $j$, we have that $\sum_{j=1}^g \sum_{(c_1,c_2,c_3)} n_{c_1}^jn_{c_2}^jn_{c_3}^j = O(gn^3)=\tilde{O}(n^3)$.
\end{proof}

\begin{theorem}
(Counting mod $R$) \kNLstC~has a $\tO(|C|^k + |C|^{k-2}|E|)$ time algorithm for all $k\geq 2$ (when $\lg(R)$ is sub-polynomial). 
\label{thm:algForknlstc}
\end{theorem}
\begin{proof}
We can guess the $k$ colors and use BFS to discover if s is connected to t. This takes $O(|E|)$ time.

If we are counting instead of detecting paths from $s$ to $t$ then we want to extend the BFS approach by associating an additional number to each node. Every node will keep a value of the number of paths from $s$ to that node  mod $R$. These numbers will require a sub-polynomial number of bits to represent as $\lg(R)$ is bounded to be subpolynomial. In a layered graph we can compute the number of paths from $s$ to a node $v$ in layer $i$ by summing the number of paths from $s$ to $u$ for all $u$ that are neighbors of $v$ that are in layer $i-1$. We can go through the graph by computing these numbers layer by layer staring at layer $L_0$. This also takes $O(|E|)$ time. 

Let $E_{c_i,c_j}$ be the set of all edges between nodes of colors $c_i$ and $c_j$. Let $E_{s,c_i}$ be the set of all edges between $s$ and nodes with color $c_i$. Let $E_{t,c_i}$ be the set of all edges between $t$ and nodes with color $c_i$. 

Our running time is:
$$\sum_{c_1, c_2,\ldots,c_k \in C} \left( \left( \sum_{i\in [1,k]} |E_{s,c_i}|+ |E_{t,c_i}|\right) + \sum_{i,j \in [1,k]} |E_{c_i,c_j}|\right). $$

We know that $|E_{s,c_i}|+|E_{t,c_i}| \leq 2$ from the problem definition. We also know that $|E_{c_i,c_i}| =0$. So we can simplify to:
$$\tO \left( k|C|^k + k^2 \sum_{c_1, c_2,\ldots,c_k \in C} |E_{c_1,c_2}| \right). $$
Then we can use the fact that $\sum_{c_1,c_2 \in C} |E_{c_1,c_2}| = |E|$ and that $k$ is a constant to get:
$$\tO \left( |C|^k + |C|^{k-2}|E|\right). $$
\end{proof}

\begin{theorem}
There is an algorithm for (counting mod $R$) \kELstC[]~that runs in time $\tilde{O}(|C|^{k-1}|E|)$ (when $\lg(R) =n^{o(1)}$).
\label{thm:kelstcALG}
\end{theorem}
\begin{proof}
We do an exhaustive search for all $k$ colors. Once we guess $k$ colors $c_1, \ldots ,c_k$ then we simply run a $O(|E|)$ time algorithm for (directed/undirected) reachability on this input. If we are counting paths mod $R$ then we use the fact that the graph is a directed acyclic graph to count the number of paths from $s$ to every node, so we can go through a normal breadth first search but keeping the count mod $R$. Because $\lg(R) = n^{o(1)}$ we can track these sums in sub-polynomial time. 

We start by sorting our edges by their color (so that given a guess of colors we can in $k \lg(n)$ time give pointers to the full set of all edges of that color). Let $e(c)$ be the number of edges of color $c$. Then our running time can be given as:
$$\sum_{c_1,\ldots,c_k \in [1,|C|]} e(c_1) + \ldots + e(c_k) + k\lg(n).$$
Consider a particular one of the additive parts of this sum: $\sum_{c_1,\ldots,c_k \in [1,n]} e(c_i).$ This can be re-written as:
$$\sum_{c_1,\ldots, c_{i-1}, c_{i+1}, \ldots ,c_k \in [1,|C|]} \left( \sum_{c_i \in [1,|C|]} e(c_i) \right).$$
Which is
$$\sum_{c_1,\ldots, c_{i-1}, c_{i+1}, \ldots ,c_k \in [1,|C|]} |E| = |C|^{k-1}|E|.$$

So the total running time is $k |C|^k + k|C|^{k-1} |E| \lg(R) \lg(n) + |E|\lg(n)$. The $k |C|^k$ time comes from running all of our small instances.  The $k \lg(n)$ coming from the need to give pointers into where our $k$ colors of edges are stored. The factor of $\lg(R)$ comes from tracking the count mod $R$ in the counting version.  And finally, the $|E|\lg(n)$ comes from sorting our edges according to color. 
\end{proof}

\begin{theorem}
There is an algorithm for detecting $(k+1)$L-MF* on an $n$-node graph that runs in $\tO(n^k)$ time.
\end{theorem}
\begin{proof}
First we run a max flow algorithm on the graph to obtain the value $|F|$ of the max flow. Since the graph is unit-capacitated, this can be done in $O(n\sqrt{m})=O(n^2)$ time. 

Recall that the edges connected to the source $s$ and the sink $t$ have a special label $l^*$. So this label must be among the $k+1$ labels. Now for any choice of $k$ labels $l_1,\ldots,l_k$, we consider the subgraph induced on the edges with labels in $l_1,\ldots,l_k,l^*$, and we run a max flow algorithm on this graph. If the max flow value on this graph equals $|F|$, we are done. Otherwise if all these graphs have maximum flow less than $|F|$, there is no max flow with $k+1$ labels. Note that the max flow in each small graph takes $o(n)$ time since for each label the number of edges with that label is $o(n)$.
\end{proof}

Now we turn to regular expressions matching problem, and state an efficient algorithm for counting the number of alignments of the pattern on sub-strings of the text. First we state two lemmas.

\begin{lemma}
\label{lem:nfa-fast}
Let $M$ be an NFA with no cycles of length more than $1$. Let a computation of a string $t$ in $M$ be a sequence of states from the start state to the accept state of $M$ that produces $t$.
Then given a text $T$ and a fixed integer $R$ where $\log{R}$ is sub-polynomial, there is an algorithm that computes the number of computations of substrings of $T$ in $M$ mod $R$ in $O(m|T|)$ time, where $m$ is the number of edges of $M$. 
\end{lemma}

\begin{proof}
All numbers are taken mod $R$. Let $Q$ be the set of states of $M$, and let $\Delta: Q \times \Sigma\rightarrow P(Q)$ be the transition function of $M$, where $\Sigma$ is the alphabet, and $P(Q)$ is the power set of $Q$. Recall that we have an edge from state $s$ to state $s'$ if $s'\in \Delta(s,\sigma)$ for $\sigma \in \Sigma \cup \{\epsilon\}$,  
where $\epsilon$ is the empty string. Note that for any state $s\in Q$, $s\notin \Delta(s,\epsilon)$. We can assume that there is only one accept state with no outgoing edge (and hence no self-loops).

Since $M$ has no cycles other than self-loops, it has a topological ordering $s_1,s_2,\ldots,s_r$ where $s_1$ is the start state, $s_r$ is the accept state, $r$ is the number of states of $M$ and there is no edge from state $s_i$ to $s_j$ if $i>j$. We compute the number of computations of substrings of $T$ in $M$ by dynamic programming. Let $|T|=n$, and let $T_i$ be the postfix of $T$ starting at $i$ for $i=1,\ldots,n+1$, where $T_{n+1}$ is the empty string. Let $M_i$ be the NFA obtained from $M$ by having $s_i$ as the start state. 
For $i = 1,\ldots,r$ and $j=1,\ldots,|T|$, let $f(i,j)$ be the number of computations of prefixes of $T_j$ by $M_i$. So $\sum_{j=1}^n f(1,j)$ is what we have to compute. 

As the base case, we have that $f(r,n+1)=1$. Let $N_{out}(s_i)$ be the set of outgoing neighbors of $s_i$, i.e. we have that $s_j\in N_{out}(s_i)$ if there is an edge from $s_i$ to $s_j$. Similarly we define $N_{in}(s_i)$ to be the set of incoming neighbors of $s_i$. 

Fix $i,j$. Suppose that we have computed $f(i',j')$ for all $i'\ge i$ and $j'\ge j$ where $i'+j'>i+j$. We compute $f(i,j)$ as follows.
$$
f(i,j)=\sum_{s_{\ell}\in \Delta(s_i,T[j])} f(i+1,\ell) + \sum_{s_\ell\in \Delta(s_i,\epsilon)} f(i,\ell) 
$$
Note that $s_i\notin \Delta(s_i,\epsilon)$, so we can compute this sum, which takes $O(|N_{out}(s_i)|)$ to compute. Hence the computation of all $f(i,j)$s takes $O(mn)$ time. 
\end{proof}

\begin{lemma}
\label{lem:regex-has-nfa}
%Let $E$ be a regular expression using with concat ($``\cdot"$), OR ($``|"$) and Kleen star ($``*"$) as its operations, where the expressions that the star operator is applied on is only a
If $E$ is a regular expression of the type $T_0$ (see Figure \ref{fig:regex-type}), there is an NFA equivalent to $E$ that has no cycle of length more than $1$. This MFA has $O(|E|)$ edges.
\end{lemma}
\begin{proof}
Let a sub-type of a regular expression type $T$ be a type shown by a sub-tree of the tree of $T$. We show that for any regular expression of type $T_0$ or any sub-type of $T_0$, there is an NFA equivalent to $E$ that has no cycle length more than $1$. Recall that $\epsilon$ is the empty string.

So let $E$ be a regular expression of any sub-type of $T_0$. We construct the NFA of $E$ in a recursive manner. As the base case, suppose that $E$ has length $1$. So it consists of only one symbol $a$, for which a two state NFA suffices: Let $s_1$ be the starting state and $s_2$ be the accept state, and let $e$ be an edge from $s_1$ to $s_2$ with value $a$ (equivalently, the transition function $\Delta$ is $\Delta(s_1, a)=\{s_2\}$).

If $E$ has length more than $1$, it is of the form $A\bullet B$ or $A*$, where $\bullet$ is one of the operators concatenation ($``\cdot"$) or OR ($``|"$), and $A$ and $B$ are two regular expressions of a sub-type of $T_0$. Let $M_A$ and $M_B$ be the NFAs corresponding to $A$ and $B$ respectively, with $s_A,s_B$ as the corresponding start states and $t_A,t_B$ as the corresponding accept states.

So we have three cases:

\begin{enumerate}
    \item Concatenation: suppose that $\bullet=\cdot$. Define $M$ to be the MFA that consists of $M_A$ and $M_B$, with an edge added from $t_A$ to $s_B$ with value $\epsilon$. Let $s_A$ be the start state of $M$ and $t_B$ be the accept state of $M$.
    \item Or: suppose that $\bullet=|$. Let $s$ be a new state, which has an edge of value $\epsilon$ to $s_A$ and $s_B$. Mark $s$ as the start state of $M$. Let $t$ be a new state, where there is an edge from $t_A$ and $t_B$ to $t$ with value $\epsilon$. Let $t$ be the accept state.
    \item Star: Suppose that $E=A^*$. Since $E$ is of a subtype of $T_0$, $A$ must be of type $``|"$. So it is the OR of some symbols. Let the set of these symbols be $Q_A$. Then define $M$ to have $3$ states, $s_E$ as the start state, $t_E$ as the accept state, and $s$ as a middle state where there is a self-loop from $s$ to itself with all symbols in $Q_A$ as its values, an edge from $s_A$ to $s$ and an edge form $s$ to $t_A$ with empty string $\epsilon$ as their value. 
\end{enumerate}
It is straightforward to see that this NFA is equivalent to $E$, so that each alignment of $E$ on a text is equivalent to a computation of the text by the NFA $M$. Note that in each case we add $O(1)$ edges. So the total number of edges is $O(|E|)$.
\end{proof}

Combining Lemma \ref{lem:nfa-fast} and \ref{lem:regex-has-nfa} gives us the following Theorem.

\begin{theorem}
\label{thm:regexcounting}
Given a regular expression $E$, a text $T$ and a fixed integer $R$ where $\log{R}$ is sub-polynomial, there is an algorithm that counts the number of alignments of $E$ on substrings of $T$ mod $R$ in $O(|T||E|)$ time.
\end{theorem}

\begin{theorem}
\label{thm:algKWLCS}
There is an algorithm for \#\kwlcs~mod $R$ which runs in $\tO(n^k)$ time when $\lg(R) = o(\lg(n))$. 
\end{theorem}
\begin{proof}
Take $P_1, \ldots, P_k$ to be the input sequences. Recall that $w(P_\ell[i])$ is the weight of the symbol at position $i$ in the $\ell^{th}$ string. 

We will use dynamic programming. We will have a cell in our table for every $k$ tuples of locations in the strings $i_1, \ldots, i_k$. Every cell will contain two pieces of information: 
\begin{itemize}
    \item $\ell(i_1,\ldots, i_k)$ the length of the longest common subsequence(s) of the substrings $P_1[:i_1], \ldots, P_k[:i_k]$.
    \item $C(i_1, \ldots, i_k)$ is the count of the number of longest common subsequences mod $R$. This will have a $n^{o(1)}$ bit representation due to our restriction on $R$. 
\end{itemize}

We start by initializing all cells associated with locations $i_1, \ldots, i_k$ where any $i_j=0$. These cells are initialized to $\ell(i_1,\ldots, i_k)=0$  and $C(i_1, \ldots, i_k)=1$, as there is only one way to have a zero length string. 

Let the total sum of a cell be $\sum_{j=1}^k i_j$, we will fill cells out in order by there total sum, starting with zero and moving to $kn$. Any cell that has a $i_j$ value equal to zero will be left with its initialization. 

When filling the cell there are two cases: when $P_1[i_1]=P_2[i_2]= \ldots = P_k[i_k]$, and when that isn't true. We define some helpful notation. Let $\vec{v} =i_1, \ldots, i_k $ and let $eq_\ell(\vec{v},\vec{u})$ be a function that returns $1$ if $\ell(\vec{v}) = \ell(\vec{u})$. Let $\vec{v}_{(-1)}$ be the vector $i_1-1, \ldots, i_k-1$. Let $S(\vec{v})$ be a set of all vectors $\vec{u}$ such that for all indices $j$ we have that $\vec{u}[j] = \vec{v}[j] + \{0,-1\}$ excluding $\vec{v}$ and $\vec{v}_{(-1)}$. So all the smaller neighboring vectors of $\vec{v}$, excluding the strictly smaller vector (note these may differ from $\vec{v}$ in $1, 2, \ldots, k-1$ locations). By our order of computation all cells associated with $S(\vec{v})$ and $\vec{v}_{(-1)}$ will have been computed by the time we are computing the cell $\vec{v}$.

\textbf{
We will start with the case where  $P_1[i_j]\ne P_2[i_{j'}]$.} Our length is the maximal length seen so far. 
$$\ell(i_1,\ldots, i_k) = \max_{\vec{u} \in S(\vec{v})}\left(\ell(\vec{u})\right).$$
This is maximizing over all possible previous choices of longest common subsequence. We know our current last symbols can't all be included in the LCS. 

For setting $C$:
We want to look only at entries that are longest common subsequences, so naively you might think to just sum all the counts from the earlier cells that hit our max length of $\ell(\vec{v})$. But, we will have an inclusion exclusion issue. Consider the case of $k=2$, i.e. $2$-LCS. If $C(i-1,j)= x$, $C(i,j-1)= y$, and $C(i-1,j-1)= z$ then $C(i,j) = x+y-z$. This is because $x$ captures both all the longest sequences between $P_1[:i-1]$ and $P_2[:j-1]$ as well as those that use the symbol in location $P_1[i]$. The parallel statement is true for $y$. So we are double counting those longest common subsequences that appear in both $P_1[:i-1]$ and $P_2[:j-1]$, so we subtract out that double counting. In order to handle this smoothly we will define a more involved version of $S(\vec{v})$. Let $S_r(\vec{v})$ contain the subset of vectors $\vec{u} \in S(\vec{v}) \cup \vec{v}_{(-1)}$ where $\left(\sum_{j=1}^k \vec{v}[j]\right)-\left(\sum_{j=1}^k \vec{u}[j]\right) = r$. So $S_r(\vec{v})$ is the set of vectors that have $r$ indices that are smaller than $\vec{v}$. Now, after all this lead up, our value for $C$ is the following:
$$C(\vec{v}) = \sum_{r=1}^{k} (-1)^r \sum_{\vec{u}\in S_r(\vec{v})}  eq_{\ell}(\vec{v},\vec{u})C(\vec{u}).$$
We need to mod this by $R$ so that the total bits in the representation is not too large. 

So in $\tO(2^kR)$ time per cell we can compute \#\klcs. There are a total of $n^k$ cells so the total time for this algorithm is $\tO(n^k)$.

\textbf{
Now we will deal with the case of $P_1[i_1]=P_2[i_2]= \ldots = P_k[i_k]$.} First let us set $\ell(\cdot)$: 
$$\ell(i_1,\ldots, i_k) = \ell(i_1-1,i_2-1,\ldots, i_k-1) +w(P_1[i_1]).$$
This works because we have a matching symbol. Our new longest common subsequence at this location will have a length one longer than the longest sequence that existed using none of the current symbols. 

For setting $C$:
We want to count two non-overlapping sets. One set is the weighted longest common subsequences at location $\vec{v}_{(-1)}$. The other set is all the strings that use some but not all of the symbols from our current location $\vec{v}$. For counting this we need inclusion exclusion like before. 
$$C(\vec{v}) = C(\vec{v}_{(-1)}) + \sum_{r=1}^{k} (-1)^r \sum_{\vec{u}\in S_r(\vec{v})}  eq_{\ell}(\vec{v},\vec{u})C(\vec{u}).$$
This counts all longest sequences that include the current symbols indicated by $\vec{v}$ by including the count of $C(\vec{v}_{(-1)})$, it also counts all alternate ways to achieve a longest common subsequence of this length using at least one of these symbols by the summation. We need to mod this by $R$ so that the total bits in the representation is not too large. 
\end{proof}

\begin{corollary}
\label{cor:algKLCS}
There is an algorithm for \# \klcs~mod $R$ which runs in $\tO(n^k)$ time when $\lg(R) = o(\lg(n))$. 
\end{corollary}
\begin{proof}
The \#\klcs~problem is a special case of \#\kwlcs~problem where $w(\cdot)$ is the constant function that returns $1$.
\end{proof}

\section{Framework for Generating Uniform Average Case Hardness}
\label{sec:Framework}
\subsection{Preliminaries}
\subsubsection{Notation}
\begin{definition}
We use $x \sim \mathbb{F}^n_{p}$ to mean that $x$
is drawn uniformly at random from all $p^n$ values in the support of $\mathbb{F}^n_{p}$.
\end{definition}

\subsubsection{Getting Nearly Uniform Bit Strings from Finite Field Elements}
Adser{\`{a}} et. al show that counting cliques is hard on average over the uniform distribution where every edge exists iid \cite{UniformCliqueABB}. 

\begin{theorem}
Let $Z_i = Ber[\mu]$ where $\mu\in (0,1)$. Then let $Y \equiv  \sum_{i = 0}^{t} Z_i \cdot 2^i  \pmod{p}$. Let the total variation distance between $Y$ and $\Unif[0,p-1]$ be $\Delta$. Then there exists a constant $C$ such that if $t \geq C\cdot \mu^{-1} \cdot (1-\mu)^{-1} \cdot \lg(p/\eps^2) \cdot \lg(p)$, then $\Delta \leq \eps$ \cite{UniformCliqueABB}.
\end{theorem}

\begin{theorem}
If you are given an input with $n$ numbers $x_1, \ldots,x_n$ each chosen from $\Unif[1,p-1]$ there exists a sampling procedure which runs in time $O(n\lg^3(n)t(1/p-\eps)^{-1})$ that, with probability at least $1-2^{-\lg^2(n)}$, produces a new set of numbers $I = x'_1,\ldots,x'_n$ such that:
\begin{enumerate}
    \item $x'_i \equiv x_i \mod p$ for all $i$. 
    \item Each $x'_i$ is $t$ bits long where $t \geq C\cdot \mu^{-1} \cdot (1-\mu)^{-1} \cdot \lg(p/\eps^2) \cdot \lg(p)$.
    \item $I$ is total variation distance $n\epsilon$ from the distribution where every bit of $x'_i$ is iid sampled from $Ber[\mu]$.   
\end{enumerate}
(inspired by \cite{UniformCliqueABB})
\label{thm:SampleProcudure}
\end{theorem}
\begin{proof}
Let $Z_i = Ber[\mu]$ where $\mu\in (0,1)$. Then let $Y$ be the distribution formed by $\sum_{i = 0}^{t} Z_i \cdot 2^i \pmod{p}$. 

Consider the procedure to generate $x'_i$ where we sample a number $y$ from $Y$, if $y \equiv x_i \pmod{p}$ then $x'_i =y$, else repeat. We take $O(t)$ time to produce a sample. We succeed with the probability that  $y \equiv x_i \pmod{p}$. This probability is at least $\frac{1}{p}-\eps$, because $\eps$ is the total variation distance of $Y$ and $\Unif[1,p-1]$. Thus, the time to produce a single sample in expectation is $O(t (1/p-\eps)^{-1})$. To fail $\Theta(t (1/p-\eps)^{-1} \lg^3(n))$ times in a row will happen with probability at most $1-2^{-2\lg^{2}(n)}$. If we fail $\Theta(t (1/p-\eps)^{-1} \lg^3(n))$ times in a row simply halt the program and throw an error. 

We run this procedure for all $n$ numbers, thus taking at most $O(n t (1/p-\eps)^{-1} \lg^3(n))$ time to succeed with probability at least $1-n2^{-2\lg^{2}(n)} \geq 1-2^{-\lg^{2}(n)}$. 

The total variation distance from each individual $x'_i$ to the uniform distribution is $\eps$ and there are $n$ inputs in total. Thus, the total variation distance is at most $n\epsilon$ by the union bound.
\end{proof}

\begin{corollary}
If you are given an input with $n$ numbers $x_1, \ldots,x_n$ each chosen from $\Unif[1,p-1]$ there exists a sampling procedure which runs in time $O(n\lg^3(n)t(1/p-1/n^3)^{-1})$ that, with probability at least $1-2^{-\lg^2(n)}$, produces a new set of numbers $I = x'_1,\ldots,x'_n$ such that:
\begin{enumerate}
    \item $x'_i \equiv x_i \pmod{p}$ for all $i$. 
    \item Each $x'_i$ is $t$ bits long where $t \geq C\cdot \mu^{-1} \cdot (1-\mu)^{-1} \cdot (\lg(p) + 6\lg(n)) \cdot \lg(p)$.
    \item $I$ is total variation distance $1/n^{2}$ from the distribution where every bit of $x'_i$ is iid sampled from $Ber[\mu]$.   
\end{enumerate}
\label{cor:SampleButNoEps}
\end{corollary}
\begin{proof}
Simply plug in $\eps = 1/n^3$ to Theorem \ref{thm:SampleProcudure}.
\end{proof}

\subsection{The framework}
In this section we are going to show that any problem $P$ with a  $\gPol{\cdot}$ is hard over the uniform average case. 
We define $\gPol{\cdot}$ in Definition \ref{def:goodPoly}.

First, we want to convert our problem over a polynomial large finite field to a problem over many $O(\lg(n))$ sized finite fields. We will use the Chinese Remainder Theorem (CRT) to do this. 

\begin{lemma}
\label{lem:CRT}
Let $P$ be some problem with output in range $[1,n^c]$. Let $P_p$ be the same problem as $P$, but where $P_p(\vec{I}) \equiv P(\vec{I}) \pmod{p}$. 

Let $f$ be a \gPol{$P$}. Let $f_1, \ldots, f_{s}$ be a set of $s$ polynomials where  $s =O(\lg(n)/\lg\lg(n))$. We define $f_i$ as the same polynomial as $f$, but over finite field $F_{p_i}$ where $p_i = \Theta(\lg(n))$ and all $p_i$ are distinct. 

Then, for all $i$,  $f_i$ is a \gPol{$P_{p_i}$}.

Finally, given $f_i(\vec{I})$ for all $i \in [1,s]$ we can return $P(\vec{I})$. 
\end{lemma}
\begin{proof}

%If $f$ has a polynomial number of monomials then so does $f_i$ for all $i$.\\
If $f(\vec{I}) =  P(\vec{I})$ then trivially $f(\vec{I}) \equiv P(\vec{I})\pmod{p}$. As a result $f_i(\vec{I}) \equiv P_{p_i}(\vec{I}) \equiv P(\vec{I}) \pmod{p_i}$.\\
If $f$ has degree $d$ then $f_i$ also has degree $d$ (it certainly has at most $d$, because $f$ is strongly $d$-partite they will in fact be equal). \\
If $f$ is $d$-partite then so is $f_i$. \\
Thus, $f_i$ is a \gPol{$P_{p_i}$}.

Given $f_i(\vec{I})$ for all $i \in [1,s]$ we know $P_{p_i}(\vec{I})$ for all $i \in [1,s]$. We can use the Chinese Remainder Theorem to find the value of $P$ as long as $\Pi_{i=1}^s p_i \geq n^c$. By the prime number theorem there is a sufficiently large constant $c'$ such that there are more than $2c \lg(n)/\lg\lg(n)$ primes between  $\lg(n)$ and $c'\lg(n)$. If we choose these primes to be $p_1,\ldots, p_{s=2c \lg(n)/\lg\lg(n)}$ then $\Pi_{i=1}^s p_i \geq n^{2c} \geq n^c$.
\end{proof}

Now we want to apply a worst-case to average case reduction for each $f_i$ separately. We can use Lemma 1 from \cite{BallWorstToAvg} to achieve this. 

\begin{lemma}
Consider positive integers $n$, $d$, and $p$, and an $\epsilon \in (0, 1/3)$ such that $d > 9$, $p$ is
prime and $p > 12d$. Suppose that for some polynomial $f : \mathbb{F}^n_p \rightarrow \mathbb{F}_p$ of degree at most\footnote{Ball et al. simply say a polynomial of degree $d$, however, unsurprisingly, their proof does not require the polynomial be of degree at least $9$ to work.} $d$, there is an
algorithm $A$ running in time $T(n)$ such that when $x$ is drawn uniformly at random from all inputs $\mathbb{F}^n_p $:
$$Pr[A(x) = f(x)] \geq 1-\eps.$$

Then there is a randomized algorithm $B$ that runs in time $O(nd^2log^2(p) + d^3 + T(n)d)$ such that
for \emph{any} $x \in \mathbb{F}^n_p$:
$$Pr[B(x) = f(x)] \geq 2/3.$$
\cite{BallWorstToAvg}
\label{lem:ballWCtoAC}
\end{lemma}

Notably, we demand that $d= o(\lg(n)/\lg\lg(n))$ and we use $p = \Theta(\lg(n))$, so  $p > 12d$. The running time, given these choices, is $\tO(n + T(n))$ time. 

\begin{corollary}
Assume an $f$ exists that is \gPol{$P$}. Then, let $f_1, \ldots , f_s$ be the polynomials described in \ref{lem:CRT}.
Let $A$ be an algorithm  that runs in time $T(n)$ such that when $x \sim \mathbb{F}^n_{p_i}$:
$$Pr[A(x) = f_i(x)] \geq 3/4,$$
for all $i$. 
Then there is a randomized algorithm $B$ that runs in time $\tO(n + T(n))$ such that
for \emph{any} $\vec{I} \in \{0,1\}^n$:
$$Pr[B(\vec{I}) = P(\vec{I})] \geq 1-O\left(2^{-\lg^2(n)}\right).$$
\label{cor:RandomFsolvesP}
\end{corollary}
\begin{proof}
%Let us consider a separate polynomial for each $f_i$, call them $A_i$ and $B_i$. 
%We are going to use Lemma \ref{lem:ballWCtoAC} for each $f_i$, using algorithm $A$
We use Lemma \ref{lem:ballWCtoAC} for each polynomial $f_i$. It follows that having an algorithm $A$ for computing $f_i$ over the uniform input $\mathbb{F}^n_{p_i}$ that succeeds with probability $3/4$ implies that a randomized algorithm $B'_i$ exists that succeeds with probability $2/3$.

We can now create an algorithm $B_i$ by running $B'_i$ for $\Theta(\lg^3(n))$ times and pick the most common output, this will return the correct answer with probability at least $1-2^{-\lg^{2.5}(n)}$. 

%The $A$ of the Lemma implies all $A_1, \ldots, A_s$ exist and thus all $B_1,\ldots, B_s$ exist. 
Now if all of $B_1,\ldots, B_s$ return the correct answer then we can use the CRT trick of Lemma \ref{lem:CRT} to compute the value of $P(\vec{I})$. All of $B_1,\ldots, B_s$  return the correct answer with probability at least $1-s2^{-\lg^{2.5}(n)} = 1-O(\lg(n)/\lg\lg(n))2^{-\lg^{2.5}(n)} < 1-O\left(2^{-\lg^{2}(n)}\right)$
\end{proof}

So, we now want to show that solving random instances of $P$ can solve random instances $f_i(x)$ where $x \sim \mathbb{F}^n_{p}$. To do this we will use the sampling procedure described in Corollary \ref{cor:SampleButNoEps}. We will also use the fact that $P_{p_i}(x)=f_i(x)$ when $x$ is a zero and one input.

\begin{lemma}
Assume a $d$ degree polynomial $f$ exists that is \gPol{$P$}. Then, let $f_1, \ldots , f_s$ and $p_1, \ldots , p_s$ be the polynomials and primes described in Lemma \ref{lem:CRT}.

Let $A$ be an algorithm  that runs in time $T(n)$ such that when $\vec{I}$ is formed by $n$ bits each chosen iid from $Ber[\mu]$ where $\mu\in (0,1)$ is a constant, then:
$$Pr[A(\vec{I}) = P(\vec{I})] \geq 1-1/\omega\left( \lg^d(n)\lg\lg^d(n) \right).$$
 
Then there is a  $B$ that runs in time $\tO(n + T(n))$ such that when $x \sim \mathbb{F}^n_{p_i}$:
$$Pr[B(x) = f_i(x)] > 3/4,$$
for all $f_i$. 
\label{lem:UnifPtoACf}
\end{lemma}
\begin{proof}
Let $D_{\mu}$ be the distribution over inputs where each of the $n$ bits is chosen iid from $Ber[\mu]$, that is one is chosen with probability $\mu$ and zero is chosen with probability $1-\mu$.
Recall that when we say $\vec{Z} \sim D_{\mu}$ we mean that $\vec{Z} $ is drawn from the distribution $D_{\mu}$. We will use an abuse of notation where we run $f_i(\vec{Z})$, when we do this we mean that one should interpret the $n$ length bit vector as $n$ values from $F_{p_i}$ where $0$ maps to $0 \in F_{p_i}$ and $1$ maps to $1 \in F_{p_i}$. Additionally when we have a vector $v$ we will use $v[j]$ to represent the $j^{th}$ number in $v$.

In this proof we will show how to use $P(\vec{Z})$ to solve instances of $f_i(\vec{Z})$ for all $i$. Note that we can simply take the output of $P(\vec{Z})$ modulo $p_i$.
So we want to use $f_i(\vec{Z})$ where $\vec{Z}  \sim D_{\mu}$ to solve $f_i(z)$ where $z \sim \mathbb{F}_{p_i}^n$.

Let $f'$ be the function $f_i$ but taken over the integers instead of $F_{p_i}$. Note that this is the same $f'$ regardless of $f_i$. We have that if $x \in \mathbb{F}_{p_i}^n$ then $f'(x) \equiv f_i(x) \pmod{p_i}$. 
Furthermore, if we make a new input $x'$ where $x'[j] \equiv x[j] \pmod{p_i}$ for all $j\in [1,n]$ then  $f'(x) \equiv f_i(x) \pmod{ p_i}$. 
So, given an input $x \sim \mathbb{F}_{p_i}^n $ we will take the sampling procedure of Corollary \ref{cor:SampleButNoEps} and make a new input $x'$, where $x'[j]$ is a $t =O(\mu^{-1} \cdot (1-\mu)^{-1} \cdot (\lg(p_i) + 6\lg(n)) \cdot \lg(p_i))$ bit number. Note that because $\mu$ is constant and neither zero nor one and $p_i =\Theta(\lg(n))$ then $t =O(\lg(n)\lg\lg(n))$. 
Furthermore, any given number $x'[j]$ has the property that the distribution over its binary representation has total variation distance $\leq 1/n^3$ from the distribution where all $t$ bits are chosen iid from $Ber[\mu]$. 
Thus, all $tn$ bits in our new input $x'$ have total variation distance at most $1/n^2$ from the distribution where all $tn$ bits are chosen iid from $Ber[\mu]$.

Now, we can compute the value of $f'(x')$ with $t^d$ calls to $f'$ where every call has a zero one input. Every monomial is formed by one variable from each of the $d$ partitions. Let $m$ be the number of monomials. So we can write our polynomial $f'$ as follows :
$$f'(x') = \sum_{j=1}^m y_{k_{j,1}}\cdot y_{k_{j,2}}\cdots y_{k_{j,d}},$$
where $y_{k_{j,\ell}}$ is a variable from the $\ell^{th}$ partition $S_{\ell}$. The input $x'$ is formed with $n$ of these input variables $y_{k_{j,\ell}}$.

We can break down this multiplication for every bit. Let $y_{k_{j,\ell}}[r]$ be the $r^{th}$ bit of $y_{k_{j,\ell}}$. Now we can rewrite our sum. Recall that $g_{f'}(v_1,\ldots,v_d)$ is the function such that $f'$ can be written as a sum of calls to $g_{f'}$, where $v_\ell$ is a variable from partition $S_\ell$: 
$$f'(x') =\sum_{j=1}^m \left( \sum_{r_1,\ldots,r_d \in[0,t-1]} 2^{r_1 +\ldots + r_d} \cdot y_{k_{j,1}}[r_1]\cdot y_{k_{j,2}}[r_2]\cdots y_{k_{j,d}}[r_d] \right).$$
%$$f'(x') =\sum_{j=1}^m \sum_{r_1,\ldots,r_d \in[0,t-1]} g_{f'}(2^{r_1},2^{r_2}, \ldots, 2^{r_d}) \cdot y_{k_{j,1}}[r_1]\cdot y_{k_{j,2}}[r_2]\cdots y_{k_{j,d}}[r_d].$$
Put in words, we can multiply $d$ numbers each of $t$ bits by making a weighted sum over the $t^d$ multiplications of the bits of the $d$ numbers. 

%\xxx{TODO: this is place that is effected by our definition of what strongly k-partitie means. This ability to correctly do the re-weighting is what we need in order to}

%We re-weight by the amount dictated by the function $g_{f'}$\footnote{ For a concrete example consider the case of $g_{f'}(v_1,\ldots,v_d) = v_1\cdot v_2 \cdots v_d$ then every monomial is simply the multiplication of exactly one variable from each partition (as is the case with the k-clique polynomial from \cite{UniformCliqueABB} for example). The re-weighting is $2^{r_1 +\ldots + r_d}$. }. 

Now, we want to create $t^d$ inputs $\hat{x}_1,\ldots, \hat{x}_{t^d}$. They are formed by taking all possible choices of $r_1,\ldots,r_d$ where each $r_{\ell}$ is an integer in $[0,t-1]$. Given a choice of $r_1,\ldots,r_d$ we create a new input $\hat{x}_j$ by taking all variables in $S_{\ell}$ and making their value in $\hat{x}_j$ be the  $r_{\ell}^{th}$ bit of that variable in $x'$.

Now, call $A(\hat{x}_j)$ for all $j\in[1,t^d]$. Note that  $P(\hat{x}_j) \equiv f_i'(\hat{x}_j)  \equiv f_i(\hat{x}_j) \pmod{p_i}$. So, if $A(\hat{x}_j) = P(\hat{x}_j)$ for all $j\in [1,t^d]$ then we can return the value of $f'(x') \equiv f_i(x) \pmod{p}$. 

By the definition of $A$ in this Lemma, $A$ must succeed on any individual random input $x\sim D_{\mu}$ with probability $1-1/\omega(\lg^d(n)\lg\lg^d(n))$. The total variation distance of any $\hat{x}_j$ from $D_{\mu}$ is at most $1/n^2$. So $A$ must succeed on any one given random input $\hat{x}_j$ with probability $1-1/\omega(\lg^d(n)\lg\lg^d(n))-1/n^2$ which is $1-1/\omega(\lg^d(n)\lg\lg^d(n))$.  

Our inputs $\hat{x}_j$ are not iid from each other, however, if $A$ is correct with probability $1-q$ on a given input from $\hat{x}_j$ then $A$ must be correct with probability at least $1-qt^d$ on $t^d$ inputs $\hat{x}_j$ at once. 

So, $A$ will return correct answers for all $t^d$  inputs $\hat{x}_j$ at once with probability at least $1-1/\omega(1)$. Given these correct answers we can compute $f_i(x)$, for all $f_i$. So, an algorithm $B$ exists that makes $t^d$ calls to $A$ and takes $n \lg^4(n)t$ time to produce our new sampled input $x'$ from $x$. \\
$B$ returns $f_i$ correctly with probability at least $1-1/\omega(1) > 3/4$.\\
$B$ takes a total time of $O(t^d T(n) + n)$. We have that $t=O(\lg(n)\lg\lg(n))$ and $d= o(\lg(n)/\lg\lg(n))$ (by our definition of \gPol{P}). Thus, $t^d = n^{o(1)}$. So we have that $B$ runs in time $\tO( T(n) + n)$.
\end{proof}

This next theorem gives a worst case to average case reduction for $P$.

%\begin{theorem}
\begin{reminder}{Theorem \ref{thm:framework}}
Let $\mu$ be a constant such that $0 < \mu <1$.
Let $P$ be a problem such that a function $f$ exists that is a \gPol{$P$}, and let $d$ be the degree of $f$. 
Let $A$ be an algorithm  that runs in time $T(n)$ such that when $\vec{I}$ is formed by $n$ bits each chosen iid from $Ber[\mu]$:
$$Pr[A(\vec{I}) = P(\vec{I})] \geq 1-1/\omega\left( \lg^d(n)\lg\lg^d(n) \right).$$
Then there is a randomized algorithm $B$ that runs in time $\tO(n + T(n))$ such that
for \emph{any} for $\vec{I} \in \{0,1\}^n$:
$$Pr[B(\vec{I}) = P(\vec{I})] \geq 1-O\left(2^{-\lg^2(n)}\right).$$
\end{reminder}
%\label{thm:framework}
%\end{theorem}
\begin{proof}
We will use Lemma \ref{lem:UnifPtoACf} and Corollary \ref{cor:RandomFsolvesP} to get this result. 

Note that the algorithm $A$ here can be used as the algorithm $A$ in Lemma \ref{lem:UnifPtoACf}.

Furthermore, note that the algorithm $B$ of Lemma \ref{lem:UnifPtoACf} has the same requirements as the algorithm $A$ of Corollary \ref{cor:RandomFsolvesP}.

So, given the algorithm $A$ of this theorem we can produce the algorithm $B$ from Corollary \ref{cor:RandomFsolvesP}. 

The algorithm $B$ of Corollary \ref{cor:RandomFsolvesP} has the same properties of the algorithm $B$ described in this theorem. 

Thus, algorithm $A$ implies that an algorithm $B$ exists. 
\end{proof}

%Consider problems $P$ that in the worst case require super linear time and have a \goodPoly~$f$. These problems $P$ then require, up to subpolynomial factors, the same amount of time even when every bit in the input is selected uniformly at random.\\
%In section \ref{sec:SubgraphsWithKNodes}~we show that all small sub-graph counting problems, $P$, have \gPol{$P$}~functions $f$, showing that all small subgraph counting problems require the same time in the worst case as they do over Erd\H{o}s-R{\'e}nyi graphs.\\
%In section \ref{sec:compressedVariants} we show that the factored problems we defined (\ckfunc[]~and \cfkc[]) are average-case hard over the uniform input from their worst case variants. 

%Furthermore, consider problems $P$ where there is a super polynomial gap between their worst case running time and their running time when every bit in their input is selected uniformly at random. For these problems we know that no \gPol{$P$}~exists.\\
%In section \ref{sec:OV}~we will show that counting-OV can be solved with high probability in $\tO(n^{2-\epsilon})$ time for some constant $\eps>0$ when every bit is selected uniformly at random with constant probability probability $p$. However, conditioned on the SETH (and indeed weaker assumptions), counting-OV requires $n^{2-o(1)}$ time. Thus, there is no \gPol{counting-OV} if SETH is true. 

%\section{A How To Guide for Applying the Tools in This Paper}
%\label{app:HowTo}
%\input{HowTo}

%\section{Constant Probability Framework}
%\label{sec:ConstProb}
%\input{ConstFramework.tex}

\end{document}